\documentclass[a4paper]{lmcs}
\pdfoutput=1

\usepackage{lastpage}
\lmcsdoi{16}{1}{3}
\lmcsheading{}{\pageref{LastPage}}{}{}%
{Feb.~23,~2018}{Jan.~14,~2020}{}

\usepackage{amsmath}
\usepackage{amssymb}
\usepackage{amsthm}
\usepackage{latexsym}
\usepackage{graphicx}
\usepackage{color}
\usepackage{amsfonts}
\usepackage{stmaryrd}
\usepackage{url}
\usepackage{ulem}
\usepackage{mathpartir}
\usepackage{enumerate}
\usepackage{prooftree}
\usepackage{framed}
\usepackage{tikz}

\newcommand{\technicalreport}[1]{\ignore{#1}}

\theoremstyle{plain}  
\newtheorem{property}{Property}

\long\def\ignore#1{\relax}

\newcommand{\modifref}[1]{#1}
\newcommand{\modifrefb}[1]{#1}

\newcommand{\compF}[2]{\mathtt{#1}\left(#2\right)}

\newcommand{\cmin}[2]{#1\setminus\!\!\!\setminus\, #2}

\newcommand{\ty}[2]{{#1} \ftype {#2}}

\newcommand{\tyj}[3]{{#2} \vdash #1:#3}
\newcommand{\tyjj}[3]{{#2} \Vdash #1:#3}

\newcommand{\tingD}[2]{{#1} \tri {#2}}

\newcommand{\lam}[1]{\lambda #1}
\newcommand{\lx}{\lambda x}

\newcommand{\fv}[1]{\mathtt{fv}(#1)}

\newcommand{\fn}[1]{\mathtt{fn}(#1)}

\def\l{\lambda}
\def\Gam{\Gamma}
\def\Sig{\Sigma}
\def\Del{\Delta}
\def\Lam{\Lambda}
\def\sig{\sigma}
\def\del{\delta}
\def\al{\alpha}
\def\gam{\gamma}

\newcommand{\sep}{\hspace*{0.5cm}}
\newcommand{\msep}{\hspace*{0.2cm}}

\newcommand{\Rew}[1]{\rightarrow_{#1}}

\newcommand{\Rewplus}[1]{\rightarrow^{+}_{#1}}

\newcommand{\Rewn}[2][*]{\rightarrow^{#1}_{#2}}
\newcommand{\SN}[1]{\mathcal{SN}(#1)}
\newcommand{\SNSP}{\mathcal{SN}(\lmu)}
\newcommand{\WN}[1]{\mathcal{WN}(#1)}

\newcommand{\HN}[1]{\mathcal{HN}(#1)}
\newcommand{\HNO}{\mathtt{HO}}
\newcommand{\HNT}{\mathtt{HT}}
\newcommand{\HNCom}{\mathtt{HC}}

\newcommand{\isubs}[1]{\{ #1 \}}

\newcommand{\es}{ \emptyset}

\newcommand{\A}{\mathcal{A}}
\newcommand{\R}{\mathcal{R}}

\newcommand{\mult}[1]{[#1]}
\newcommand{\umult}[1]{\langle #1 \rangle}

\newcommand{\cset}[1]{ \{ #1 \} }
\newcommand{\deft}[1]{{\bf #1}}

\newcommand{\ih}{{\it i.h.}}
\newcommand{\pair}[2]{\langle #1, #2 \rangle}
\def\l{\lambda}

\newcommand{\commands}[1]{\mathcal{C}_{#1}}
\newcommand{\terms}[1]{\mathcal{T}_{#1}}

\newcommand{\objects}[1]{\mathcal{O}_{#1}}

\newcommand{\proj}[1]{\mathtt{P}(#1)}
\newcommand{\ie}{{\it i.e.}}
\newcommand{\cf}{{\it cf.}}

\newcommand{\eg}{{\it e.g.}}

\newcommand{\slist}{\mathtt{L}}

\newcommand{\cwc}[1]{{[} \! {[} #1 {]}  \! {]}}

\newcommand{\B}{\mathtt{B}}
\newcommand{\Mu}{\mathtt{M}}

\newcommand{\Gc}{\mathtt{w}}
\newcommand{\Gcs}{\mathtt{w}_{v}}
\newcommand{\Gcr}{\mathtt{w}_{n}}

\newcommand{\cntrs}{\mathtt{c}_{v}}
\newcommand{\cntrr}{\mathtt{c}_{n}}

\newcommand{\ders}{\mathtt{d}_{v}}
\newcommand{\derr}{\mathtt{d}_{n}}
\newcommand{\sm}{\setminus}

\newcommand{\rrule}[1]{\mapsto_{#1}}
\newcommand{\tri}{\triangleright}

\newcommand{\dom}[1]{\mathtt{dom}(#1)}

\newcommand{\kK}{{k \in K}}
\newcommand{\kL}{{k \in L}}
\newcommand{\lL}{{\ell \in L}}

\newcommand{\sz}[1]{\compF{sz}{#1}}

\newcommand{\ax}{\texttt{ax}}
\newcommand{\axw}{\texttt{axw}}

\newcommand{\introarrow}{\Rightarrow_{\mathtt{i}}}

\newcommand{\appu}{\Rightarrow_{{\mathtt{e}_1}}}
\newcommand{\appd}{\Rightarrow_{{\mathtt{e}_2}}}
\newcommand{\app}{\Rightarrow_{{\mathtt{e}}}}
\newcommand{\appet}{\Rightarrow_{\mathtt{e}*}}
\newcommand{\appempty}{\Rightarrow_{{\mathtt{empty}}}}
\newcommand{\subs}{\mathtt{s}}
\newcommand{\repl}{\mathtt{r}}

\newcommand{\Id}{\mathtt{I}}
\newcommand{\combK}{\mathtt{K}}

\renewcommand{\al}{\alpha}
\newcommand{\ga}{\gamma}
\newcommand{\co}[1]{[#1]}
\newcommand{\coal}{\co{\al}}
\newcommand{\ire}[2]{\{ #1 /\! \! /  #2\} }
\newcommand{\bta}{a}
\newcommand{\btb}{b}
\newcommand{\btc}{c}

\newcommand{\muu}{\#_{\mathtt{i}}}
\newcommand{\mud}{\#_{\mathtt{e}}}
\newcommand{\mudempty}{\#_{\mathtt{empty}}}

\newcommand{\rew}{\rightarrow}                  

\newcommand{\sigk}{\sigma_k}                   
\newcommand{\msigk}{\mult{\sig_k}_{\kK}}

\newcommand{\taul}{\tau_\ell}

\newcommand{\mtau}{\mult{\tau}}

\newcommand{\Gaml}{\Gamma_{\ell}}

\newcommand{\Dell}{\Delta_{\ell}}

\newcommand{\Phit}{\Phi_t}

\newcommand{\Phiul}{\Phi_u^{\ell}}
\newcommand{\Phiuk}{\Phi_u^{k}}

\newcommand{\ar}[1]{\mathtt{ar}(#1)}

\newcommand{\calC}{\mathcal{C}}
\newcommand{\calR}{\mathcal{R}}

\newcommand{\IM}{\mathcal{I}}

\newcommand{\IMl}{\mathcal{I}_{\ell}}

\newcommand{\IMu}{\IM_u}

\newcommand{\UM}{\mathcal{U}}
\newcommand{\VM}{\mathcal{V}}
\newcommand{\VMC}{\mathcal{V}_{\Com}}

\newcommand{\VMal}{\VM_{\al}}
\newcommand{\VMalp}{\VM_{\al'}}

\newcommand{\UMi}{{\UM}_i}

\newcommand{\VMl}{{\VM}_{\ell}}

\newcommand{\inter}{\wedge}
\newcommand{\union}{\vee}

\newcommand{\Com}{c}
\renewcommand{\bot}{\umult{\,}}

\newcommand{\emul}{\mult{\, }} 
\newcommand{\eumul}{\umult{\,}}

\renewcommand{\emph}[1]{{\it  #1}}

\newcommand{\IMk}{\IM_k}      

\newcommand{\UMk}{\UM_k}

\newcommand{\VMk}{\VM_k}

\newcommand{\UMl}{\UM_{\ell}}
\newcommand{\uVMk}{\union_{\kK} \VMk}
\newcommand{\uIVMk}{ \umult{\IMk \rew \VMk}_{\kK}}

\newcommand{\Om}{\Omega}
\newcommand{\Gamk}{\Gam_k}            

\newcommand{\Gamstt}{\Gam_{\cxtt}}
\newcommand{\Gamscc}{\Gam_{\cxcc}}

\newcommand{\Delk}{\Del_k}

\newcommand{\Delstt}{\Del_{\cxtt}}
\newcommand{\Delscc}{\Del_{\cxcc}}

\newcommand{\Phik}{\Phi_k}

\newcommand{\Theuu}{\Theta^1_u}                 
\newcommand{\Theud}{\Theta^2_u}                 
\newcommand{\Theu}{\Theta_u}

\newcommand{\TypCom}{\#}

\newcommand{\Empty}{\mathcal{E}}
\newcommand{\e}{{\it e}}

\newcommand{\AM}{\mathcal{A}}

\newcommand{\arob}{\symbol{64}}
\newcommand{\muju}[3]{#1 \vdash #2  \mid  #3} 
\newcommand{\muJu}[3]{#1 \Vdash #2  \mid  #3}

\newcommand{\JM}{\mathcal{J}}
\newcommand{\JMk}{\mathcal{J}_k}
\newcommand{\JMl}{\JM_{\ell}}

\newcommand{\many}{\inter}

\newcommand{\Delo}{\Del_0}
\newcommand{\Delu}{\Del_u}
\newcommand{\Gamu}{\Gam_u}
\newcommand{\Gamo}{\Gam_0}
\newcommand{\Gamt}{\Gam_t}
\newcommand{\Delt}{\Del_t}
\newcommand{\Phio}{\Phi_0}

\newcommand{\VMo}{\VM_0}

\newcommand{\alu}{\ire{\al}{u}}

\newcommand{\subxu}{\isubs{x/u}}

\newcommand{\rempl}[3]{\langle #1 / \!\! / #2 . #3 \rangle} 
\newcommand{\ctx}{\mathtt{T}}

\newcommand{\otx}{\mathtt{C}}

\newcommand{\cxtt}{\mathtt{TT}}
\newcommand{\cxtc}{\mathtt{TC}}
\newcommand{\cxct}{\mathtt{CT}}
\newcommand{\cxcc}{\mathtt{CC}}
\newcommand{\cxot}{\mathtt{OT}}
\newcommand{\cxoc}{\mathtt{OC}}

\newcommand{\cxterm}{\mathtt{T}}
\newcommand{\cxcommand}{\mathtt{C}}
\newcommand{\cxobject}{\mathtt{O}}

\newcommand{\lmuex}{\lambda\mu_{\mathtt{s}}}  
\newcommand{\nonelmuex}{\lambda {\overline{\mu\mathtt{s}}}}
\newcommand{\lmu}{\lambda {\mu}}

\newcommand{\choice}[1]{#1^*}
\newcommand{\Any}{\mathcal{A}}
\newcommand{\ComTyp}{\mathcal{C}}

\newcommand{\mrlsym}{\eta} 
\newcommand{\mrl}[1]{\mrlsym(#1)} 
\newcommand{\lexnl}[1]{  \langle \mrl{#1}, |#1| \rangle  }

\renewcommand{\vec}[1]{\overline{#1}}
\newcommand{\X}{\mathcal{X}}

\newcommand{\Phiu}{\Phi_u}
\newcommand{\Phiv}{\Phi_v}
\newcommand{\PhiCom}{\Phi_\Com}
\newcommand{\Gamv}{\Gam_v}
\newcommand{\GamCom}{\Gam_\Com}
\newcommand{\Delv}{\Del_v}
\newcommand{\DelCom}{\Del_\Com}
\newcommand{\IMt}{\IM_t}
\newcommand{\IMv}{\IM_v}

\newcommand{\Thetu}{\Theta^t_u}
\newcommand{\Thevu}{\Theta^v_u}
\newcommand{\Thealu}{\Theta^{\al}_u}

\newcommand{\Gamtu}{\Gam^t_u}
\newcommand{\Gamvu}{\Gam^v_u}
\newcommand{\Gamalu}{\Gam^{\al}_u}
\newcommand{\Gamuu}{\Gam^{1}_u}
\newcommand{\Gamud}{\Gam^{2}_u}

\newcommand{\Deluu}{\Del^{1}_u}
\newcommand{\Delud}{\Del^{2}_u}

\newcommand{\Deltu}{\Del^t_u}
\newcommand{\Delvu}{\Del^v_u}
\newcommand{\Delalu}{\Del^{\al}_u}

\newcommand{\orc}{\mathtt{OR}}
\newcommand{\andc}{\mathtt{AND}}
\newcommand{\Ua}{\mathcal{A}}
\newcommand{\Ub}{\mathcal{B}}
\newcommand{\Uy}{\UM_{y}}

\newcommand{\calH}{\mathcal{H}}
\newcommand{\calS}{\mathcal{S}}

\newcommand{\Hl}{\calH_\l}
\newcommand{\Sl}{\calS_\l}
\newcommand{\Hpl}{\calH'_\l}
\newcommand{\Spl}{\calS'_\l}
\newcommand{\Hlmu}{\mathcal{H}_{\lmu}}
\newcommand{\Slmu}{\mathcal{S}_{\lmu}}
\newcommand{\Slmuex}{\mathcal{S}_{\lmuex}}

\newcommand{\iIMsk}{\inter_{\kK} \choice{\IMk} }
\newcommand{\uVMkt}{\union_{\kK_t} \VMk}
\newcommand{\uVMkal}{\union_{\kK_{\al}} \VMk}
\newcommand{\uIVMkt}{\umult{\IMk \rew \VMk}_{\kK_t}}
\newcommand{\uIVMkal}{\umult{\IMk \rew \VMk}_{\kK_{\al}}}

\newcounter{example}[section]
\newenvironment{example}[1][]{\refstepcounter{example}\par\medskip
   \noindent \textbf{Example~\theexample. #1} \rmfamily}{\medskip}

\newcommand{\sset}{\mathcal{S}}

\newcommand{\phdot}{\phantom{.}}
\newcommand{\ftype}{\Rightarrow}
\newcommand{\Blind}{\xi }

\begin{document}

\title[Non-Idempotent Types for Classical Calculi]{Non-Idempotent Types for Classical Calculi in Natural Deduction  Style}

\author[D. Kesner]{Delia Kesner\rsuper{a}}
\author[P. Vial]{Pierre Vial\rsuper{b}}

\address{\lsuper{a}Universit\'e de Paris and Institut Universitaire de France (IUF), France}
\email{kesner@irif.fr}
\address{\lsuper{b}Inria (LS2N CNRS), France}
\email{pierre.vial@inria.fr}

\begin{abstract}
In the first part of this paper, we define two resource aware typing systems for the $\lmu$-calculus based
on non-idempotent {\it intersection} and {\it union} types.  The
non-idempotent approach provides very simple combinatorial arguments
--based on decreasing measures of type derivations-- to characterize
head and strongly normalizing terms.  
Moreover, typability provides
upper bounds for the lengths of 
the head-reduction  and  the maximal reduction sequences to  normal-form.  

In the second part of this paper, the $\lmu$-calculus is refined 
to  a small-step calculus called $\lmuex$, which is inspired by  the {\it substitution at a distance} paradigm.  The  $\lmuex$-calculus turns out to
  be compatible with a natural extension of the non-idempotent
  interpretations of $\lmu$, \ie\ $\lmuex$-reduction preserves and
  decreases typing derivations in an extended appropriate typing system.  We
  thus derive a simple arithmetical characterization of strongly
   $\lmuex$-normalizing terms by means of typing.  
 \end{abstract}


\maketitle





\section{Introduction}

The Curry-Howard Isomorphism is
the well-known relationship between programming languages and logical
systems.  While Curry first introduced the analogy between
Hilbert-style deductions and combinatory logic, Howard highlighted the
one between simply typed lambda calculus and natural deduction. Both
examples use {\it intuitionistic} logic. The extension of the
Curry-Howard Isomorphism to {\it classical} logic took more than two
decades, when Griffin~\cite{Griffin} observed that Felleisen's $\calC$ operator can be typed with the double-negation elimination. A
major step in this field was done by Parigot~\cite{Parigot92}, who
proposed the $\lmu$-calculus as a simple term notation for classical
natural deduction proofs.  The $\lmu$-calculus is an extension of the
simply typed $\l$-calculus that encodes usual {\it control operators}
as the Felleisen's {$\calC$} operator mentioned so far. Other calculi
were proposed since then, as for example Curien-Herbelin's
$\overline{\lambda}\mu\tilde{\mu}$-calculus~\cite{CH00}
based on classical sequent calculus.

The Curry-Howard correspondence has already contributed to the
understanding of many aspects of programming languages by establishing
a rich connection between logic and computation. However, there are
still some crucial aspects of computation, like the use of
resources (\eg\  time and
space), that still need to be logically understood in the classical setting.  
Establishing the
foundations of resource consumption is nowadays 
a big challenge facing the programming language community. It would lead to a new generation of programming
languages and proof assistants, with a clean type-theoretic account of
resource capabilities. \\

{\bf From qualitative \dots } Several notions of type assignment systems
for $\lambda$-calculus have been defined since its creation, including
among others {\it simple types} and {\it polymorphic types}.  However,
even if polymorphic types are powerful and convenient in programming
practice, they have several drawbacks.  \modifref{For example, it is not
possible to assign a type to a term of the form $(\l z.  \l y. y(z\,\Id)(z\,\combK))(\l x. x\,x)$, where $\Id = \l w. w$ and
$\combK = \l x. \l y.  x$, which can be understood as
a meaningful program specified by a terminating term}.  {\it
  Intersection} types, pioneered by Coppo and
Dezani~\cite{CDC78,CDC80}, introduce a new constructor $\cap $ for
types, allowing the assignment of a type of the form $( (\sigma \ftype
\sigma)\cap\sigma ) \ftype \sigma$ to the term $\lambda x. xx$. The
intuition behind a term $t$ of type $\tau_1 \cap \tau_2$ is that $t$
has both types $\tau_1$ and $\tau_2$. The symbol $\cap$ is to be
understood as a mathematical intersection, so in principle,
intersection type theory was developed by using {\it idempotent}
($\sigma \cap \sigma = \sigma$), {\it commutative} $(\sigma \cap \tau
= \tau \cap \sigma)$, and {\it associative} $((\sigma \cap \tau) \cap
\delta = \sigma \cap (\tau \cap \delta))$ laws.

Intersection types have been used as a {\it behavioural} tool to
reason about several operational and semantic properties of
programming languages.  For example, a $\l$-term/program $t$ is
strongly normalizing/terminating if and only if $t$ can be assigned a
type in an appropriate intersection type assignment system.
Similarly, intersection types are able to describe and analyze models
of $\lambda$-calculus~\cite{BarendregtCoppoDezani83}, characterize
{\it solvability}~\cite{CDV81}, {\it head normalization}~\cite{CDV81},
{\it linear-head normalization}~\cite{KV14}, and {\it
weak-normalization}~\cite{CDV81,Krivine93} among other properties.\\

{\bf \dots to quantitative Intersection types:} This technology turns
out to be a powerful tool to reason about {\it qualitative} properties
of programs, but not about {\it quantitative} ones.  Indeed, for
example, there is a type system characterizing head normalization
(\ie\ $t$ is typable in this system if and only if $t$ is head
normalizing) and which gives simultaneously a proof that $t$ is
head-normalizing if and only if the head reduction strategy terminates
on $t$. But the type system gives no information about the number of
head-reduction steps that are necessary to obtain a head normal form.  Here is where {\it non-idempotent} types come into
play, thus making a clear distinction between $\sigma \cap \sigma$ and
$\sigma$, because intuitively, using the resource $\sigma$ twice or
once is not the same from the quantitative point of view.  This change
of perspective can be related to the
essential spirit of Linear Logic~\cite{Girard87}, which removes the
contraction and weakening structural rules in order to provide an
explicit control of the use of logical resources, \ie\ to give a full
account of the number of times that a given proposition is used to
derive a conclusion.\\

\modifrefb{\textbf{The case of the $\l$-calculus:}
Non-idempotent   types  were  pioneered by 
 Philippa Gardner~\cite{Gardner}, Assaf Kfoury~\cite{kfoury96}.   
 But is Daniel  de Carvalho~\cite{Carvalho07} who
first established
in his PhD  thesis a relation between the size  of a typing derivation
in a  non-idempotent intersection type system  for the lambda-calculus
and     the      head/weak-normalization     execution     time     of
head/weak-normalizing  lambda-terms, respectively. 
Relational models of $\lambda$-calculi based on  non-idempotent types have been investigated by de Carvalho and Ehrhard in~\cite{Carvalho07,Carvalho18,Ehrhard12}. The results of de Carvalho are distilled in~\cite{Carvalho18}.
}
Non-idempotency  is  used  to   reason  about  the  longest  reduction
sequence   of   strongly  normalizing   terms   in  both   the
lambda-calculus~\cite{bernadetleng11,DebeneRonchiITRS12,Bernadet-Lengrand2013} and    in
different            lambda-calculi            with           explicit
substitutions~\cite{Bernadet-Lengrand2013,KV14}.   Non-idempotent  types also
appear in  linearization of the  lambda-calculus~\cite{kfoury96}, type
inference and inhabitation~\cite{KfouryWells2004,NeergaardM04,BKRDR14}, different characterizations     of     solvability~\cite{PaganiRonchi10}, 
verification  of higher-order programs~\cite{OngRamsay11}.\\

{\bf The case of the $\lmu$-calculus:}
 It is essential to go beyond the  $\l$-calculus to focus on the challenges posed by the \emph{advanced features} of modern higher-order programming languages and proof assistants. We want in particular to associate quantitative
information to languages being able to express control operators, as
they allow to enrich declarative programming languages with imperative
features.\\

{\bf Related works:} 
The non-idempotent intersection and union types  for
 $\lmu$-calculus that we present in this article can be seen as a quantitative
refinement of Girard's translation of classical logic into linear logic.
Different qualitative and/or  quantitative models for classical calculi were
proposed in~\cite{Selinger01,BakelBarbaneradeLiguoro11,Vaux07,AEtlca15}, thus
limiting the characterization of operational properties to
head-normalization. Intersection and
union types were also studied in the framework of classical
logic~\cite{Lau04,vB11,KikuchiSakurai14,DGL08}, but  no work
addresses the problem from a quantitative perspective.
Type-theoretical characterization of strong-normalization for
classical calculi were provided both for
$\lmu$~\cite{BakelBarbaneradeLiguoro13} and
\modifref{$\bar\lambda\mu\tilde{\mu}$-calculus~\cite{DGL08}}, but the (idempotent)
typing systems do not allow to construct decreasing measures for
reduction, thus a resource aware semantics cannot be extracted from
those interpretations.   Combinatorial strong normalization proofs
for the $\lmu$-calculus were proposed for example in~\cite{DavidNour2003}, but 
they do not provide any  explicit decreasing measure, 
and their use of 
structural induction on simple types  does not work anymore with 
intersection types, which are more powerful than simple types as
they do not only ensure termination but also characterize it.  
Upper bounds for the $\lmu$-calculus are studied in~\cite{BN18}
  by passing through standard reduction and the non erasing $\lambda \mu I$-calculus.  Different small step semantics for classical calculi
were developed in the framework of
neededness~\cite{AriolaHS11,PedrotSaurin16}, without resorting to any
resource aware semantical argument. 
\\

{\bf Contributions:} Our first
contribution is the definition of a resource aware type system for the
$\lmu$-calculus based on non-idempotent {\it intersection} and {\it
  union} types.  The non-idempotent approach provides very simple
combinatorial arguments, only based on a decreasing {\it measure}, to
characterize head and strongly normalizing terms by means of typability.
Indeed, we show that for every typable term $t$ with type
  derivation $\Pi$, if $t$ reduces to $t'$, then $t'$ is typable with
  a type derivation $\Pi'$ such that the measure of $\Pi$ is strictly
  greater than that of $\Pi'$.  In the well-known case of the
$\lambda$-calculus, such a  measure is simply
based on the structure of type tree derivations and it is given by the 
number of its nodes, which strictly
decreases along reduction.
 However, in the $\lmu$-calculus, the creation of nested
applications during $\mu$-reduction may increase the number of nodes of the
  corresponding type derivations, so that such a simple 
definition of  measure is not decreasing 
anymore. We then need to also take into account the structure ({\it multiplicity}) of certain types appearing
  in the type derivations, thus ensuring an overall decreasing of
  the measure during reduction.   This 
first result has been previously presented in~\cite{KV17}.

The second contribution of this paper is the \modifref{definition of a
  new small-step operational semantics} for $\lmu$, called $\lmuex$,
inspired from the {\it substitution at a distance
  paradigm}~\cite{AK10}, which is compatible with the non-idempotent
typing system characterizing strong normalization for $\lmu$, \modifrefb{ in that the latter extends to $\lmuex$}. \modifref{The operational semantics
  of $\lmuex$ is linear, \ie\ a single reduction step only implements
  substitution/replacement on one (free) occurrence of some
  variable/name at a time.}  We then extend the typing system for $\lmu$ \modifref{characterizing strong normalization}, so that the small-step reduction calculus $\lmuex$ preserves (and decreases the size of) typing derivations. We
generalize the type-theoretical characterization of strong
normalization to this explicit classical calculus, thus particularly
simplifying existing proofs of strong normalization for small-step
operational semantics of classical calculi~\cite{PhDPolonoski}.


\section{The $\lmu$-Calculus}
\label{s:calculus}

This section gives the syntax (Section~\ref{s:syntax}) and the
operational semantics (Section~\ref{s:operational}) of the
$\lmu-$calculus~\cite{Parigot92}. But before this we first introduce
some preliminary general notions of rewriting that will be used all
along the paper, and that are applicable to any system $\R$. We denote
by $\Rew{\R}$ the (one-step) reduction relation associated to system
$\R$.  We write $\Rewn{\R}$  for the
reflexive-transitive  closure of $\Rew{\R}$, and
$\Rew{\R}^n$ for the composition of $n$-steps of $\Rew{\R}$, thus $t
\Rew{\R}^n u$ denotes a finite $\R$-reduction sequence of length $n$
from $t$ to $u$.  A term $t$ is in $\R$-normal form, written $t \in
\R$-nf, if there is no $t'$ s.t. $t \Rew{\R} t'$; and $t$ has an
$\R$-normal form iff there is $t' \in \R$-nf such that $t \Rewn{\R}
t'$.  A term $t$ is said to be strongly $\R$-normalizing, written $t
\in \SN{\R}$, iff there is no infinite $\R$-sequence starting at $t$.
\modifrefb{When $\R$ is finitely branching and strongly $\R$-normalizing,
$\eta_{\R}(t)$ denotes the maximal length of an $\R$-reduction
sequence starting at $t$, we simply write $\mrl{t}$ if $\R$ is clear from the context. }

\subsection{Syntax}
\label{s:syntax}

We consider a countable infinite set of \deft{variables} $x, y, z, \ldots$
 and  \deft{continuation names} $\al, \beta, \ga, \ldots$. The set of 
\deft{objects} ($\objects{\lmu}$), \deft{terms} ($\terms{\lmu}$) and 
\deft{commands} ($\commands{\lmu}$) of the $\lmu$-calculus
are  given by the following grammars
\begin{center}
$ \begin{array}{llll}
   ({\bf objects})\    & o & ::= & t \mid \Com \\
   ({\bf terms})\      & t,u,v & ::= & x \mid \l x. t \mid tu  \mid \mu \al. \Com \\
   ({\bf commands})\   & \Com & ::= & \co{\al} t \\
   \end{array} $
\end{center} 

We write $\terms{\l}$ for the the set of $\l$-terms, which is a subset of $\terms{\lmu}$.
  We abbreviate $(\ldots ((t u_1)
  u_2) \ldots u_n)$ as $t u_1 \ldots u_n$ or $t \vec{u}$ when $n$ is
  clear from the context.  The grammar extends $\l$-terms
  with two new constructors: commands $\co{\al} t$ and
  $\mu$-abstractions $\mu \al. \Com$.  \technicalreport{The \deft{size of an object}
  $o$ is denoted by $|o|$.}  \deft{Free} and \deft{bound variables} of
  objects are defined as expected, in particular $\fv{\mu \al. \Com}
  := \fv{\Com }$ and $\fv{\co{\al} t} := \fv{t}$.  
  \deft{Free names}
  of objects are defined as expected, in particular $\fn{\mu \al. \Com }  :=  \fn{\Com } \sm \cset{\al}$ and $ \fn{\co{\al} t}  :=  \fn{t} \cup  \cset{\al}$.
\deft{Bound names} are defined accordingly.

We work with the standard notion of \deft{$\alpha$-conversion}
\ie\ renaming of bound variables and names, thus for example
$\co{\del}(\mu\al.\co{\al}(\l x. x))z \equiv
\co{\del}(\mu\beta.\co{\beta}(\l y. y))z$.  \deft{Substitutions} are
(finite) functions from variables to terms specified by the
  notation $\isubs{x_1/u_1, \ldots, x_n/u_n}$ $(n \geq 0)$.
\deft{Application} of the \deft{substitution} $\sigma$ to the object
$o$, written $o \sig$, may require $\alpha$-conversion in order to
avoid capture of free variables/names, and it is defined as expected.
\deft{Replacements} are (finite) functions from names to terms
specified by the notation $\{\al_1/\! \! / u_1, \ldots, \al_n /\! \! / u_n\}$ $(n
\geq 0)$.  Intuitively, the operation $\ire{\al}{u}$ passes the term
$u$ as an argument to any command of the form $\co{\al} t$, so that
  it \emph{replaces} every occurrence of $\co{\al} t$ in a
  term by $\co{\al} tu$.  Formally,
the \deft{application} of the \deft{replacement} $\Sigma$ to the object
$o$, written $o \Sig$, may require $\alpha$-conversion in order to
avoid the capture of free variables/names, and is defined as follows:
\begin{center}
$\begin{array}{rl@{\hspace{1cm}}rl}
   x\ire{\al}{u}          := & x & 
    (\l z. t)\ire{\al}{u}  := & \l z. t\ire{\al}{u}\\
   (\co{\al} t)\ire{\al}{u}   := & \co{\al} (t \ire{\al}{u}) u &  
   (tv)\ire{\al}{u}  := & t\ire{\al}{u} v\ire{\al}{u}\\
   (\co{\ga} t)\ire{\al}{u}   := & \co{\ga} t\ire{\al}{u} \quad (\gamma \not= \alpha) &
       (\mu \ga. \Com)\ire{\al}{u} := & \mu \ga. \Com\ire{\al}{u} \\ 
      \end{array}$
\end{center}
For example,  if $\Id = \l z. z$, then 
$
(x (\mu\al\co{\al} y) (\l z. zx))\isubs{x/\Id} = 
\Id  (\mu\al\co{\al} y) (\l z. z \Id)
$, and\\ 
\modifref{$(\co{\al}x (\mu\beta.\co{\al} y)) \ire{\al}{\Id} =
 \co{\al}(x \mu\beta.\co{\al} y\Id)) \Id 
$}. 

\technicalreport{   
Substitution and replacement enjoy the following well-known interaction properties. 
\begin{lem} \mbox{}
\label{l:properties-substituton-replacement}
\begin{enumerate}
\item If $x\notin\fv{v}$ and $x\neq y$ then
$
o\isubs{x/u}\isubs{y/v}=o\isubs{y/v}\isubs{x/u\isubs{y/v}}
$.
\item If $\al\notin\fn{v}$ and $\al\neq\beta$ then
$
o\ire{\al}{u}\ire{\beta}{v}=o\ire{\beta}{v}\ire{\al}{u\ire{\beta}{v}}
$.
\item \label{l:un} If $x\notin\fv{v}$, then 
  $o\isubs{x/u}\ire{\al}{v}=o\ire{\al}{v}\isubs{x/u\ire{\al}{v}}$.
  \item \label{l:deux} If $\al\notin\fn{v}$, then
   $o\ire{\al}{u}\isubs{x/v}=o\isubs{x/v}\ire{\al}{u\isubs{x/v}}$.
\end{enumerate}
\end{lem}
}

\subsection{Operational Semantics}
\label{s:operational}

We consider the following set of contexts:
\begin{center}
$\begin{array}{llll}
  \textbf{(term contexts)} &\cxterm  & := &   \Box  \mid \cxterm\,t \mid
                                   t\,\cxterm  \mid  \modifref{\l x. \cxterm} \mid \mu \al.\cxcommand \\
\textbf{(command contexts)}&  \cxcommand   & : = & \co{\al}\cxterm  \\
\textbf{(contexts)} & \cxobject & : =& \cxterm \mid \cxcommand
  \end{array}$
\end{center}
 The hole $\Box$ can be replaced by a term:  indeed, $\cxterm[t]$ and $ \cxcommand[t]$ denote 
 the replacement of 
 $\Box$ in the context by the term $t$.

The $\l\mu$-calculus is given by the set of objects introduced in
Section~\ref{s:syntax} and the \deft{reduction relation} $\Rew{\lmu}$, sometimes simply written $\Rew{}$, 
which is the closure by  all contexts of the following rewriting rules:
\begin{center}
$ \begin{array}{lll}
   (\l x. t) u & \rrule{\beta} & t \isubs{x/u} \\
   (\mu \al. \Com)u & \rrule{\mu} & \mu \al. \Com  \ire{\al}{u} \\
   \end{array} $
\end{center} 
defined by means of the substitution and replacement application notions 
given in Section~\ref{s:syntax}.  
A \deft{redex} is a term of the form $(\l x. t)u$ or $(\mu \al. \Com)u$.

An alternative specification of the $\mu$-rule~\cite{Andou03} is given
by $(\mu \al. \Com)u \rrule{\mu} \mu \gamma. \Com
\ire{\al}{\gamma.u}$, where $\ire{\al}{\gamma.u}$ denotes the
\deft{fresh replacement} meta-operation assigning
$\co{\gamma}(t\ire{\al}{\gamma.u}) u$ to $\co{\al}t$ (thus changing
the name of the command), in contrast to the standard
replacement operation $\ire{\al}{u}$ introduced in
Section~\ref{s:syntax}.  We remark however that the resulting
terms $\mu \al. \Com \ire{\al}{u}$ and $ \mu \gamma. \Com
\ire{\al}{\gamma.u}$ are $\alpha$-equivalent; thus \eg\ $\mu
\al. (\co{\al} x )\ire{\al}{u} = \mu \al. \co{\al} x u \equiv \mu
\gamma. \co{\gamma} x u = \mu \gamma. (\co{\al} x )
\ire{\al}{\gamma.u}$. We will come back to this alternative definition
of $\mu$-reduction in Section~\ref{s:lmuex}.

\modifref{
   A simple example  is given by
    the following $\lmu$-reduction sequence: 
\[ \begin{array}{lll}
  (\l x.\mu \al.\co{\al} x (\l y.\mu \delta. \co{\al} y))\ \Id\ \Id & \Rew{\beta} \\
     (\mu
     \al. \co{\al}(\Id (\l y.\mu \delta.\co{\al}y)))
      \Id   & \Rew{\beta} \\
     (\mu \al. \co{\al}(\l y.\mu \delta.\co{\al}y))
      \Id  & \Rew{\mu} \\
     \mu \al. \co{\al}(\l y.\mu \delta.\co{\al}y  \Id)  \Id  & \Rew{\beta} \\
      \mu \al. \co{\al}( \mu \delta. \co{\al} \Id  \Id )  & \Rew{\beta} \\
      \mu \al. \co{\al}( \mu \delta. \co{\al} \Id  )     &  \\
\end{array} \]

Another typical example, given by Parigot~\cite{Parigot92}, which
illustrates the expressivity of the $\lmu$-calculus is the control
operator {\bf call-cc}~\cite{Griffin}, coming from \texttt{Scheme} and
enabling backtracking, specified in the $\lmu$-calculus by the term
$ \l x.\mu \al.\co{\al} x (\l y.\mu \delta. \co{\al} y)$.  The term
{\bf call-cc} is assigned $((A\rew B)\rew A)\rew A$ (Peirce's Law) in
the simply typed $\lmu$-calculus. } \ignore{By letting $\cxcommand$ be
the command context $\co{\gamma} \Box u_1\dots u_n$, the following
example shows how the execution of ${\bf call\mbox{-}cc}\ t$ in the
environment $\cxcommand$ allows $t$ to capture the environment
$\cxcommand$: }

A reduction step $o \Rew{} o'$ is said to be \deft{erasing} iff
\modifref{$o =\cxobject[(\l x. u) v]$ and
$x \notin \fv{u}$ and $o'=\cxobject[u\isubs{x/v}] = \cxobject[u]$, or
$o= \cxobject[(\mu \al. \Com )u]$ and $\al \notin
\fn{\Com}$ and $o'=\cxobject[\mu \al. \Com \ire{\al}{u}]=\cxobject[\mu \al. \Com]$}. Thus \eg\ $(\l x. z)y \Rew{\beta} z$ and $(\mu
\al. \co{\beta}x)\Id \Rew{\mu} \mu \al. \co{\beta}x$ are erasing
steps.  A reduction step $o \Rew{} o'$ which is not erasing is called
\deft{non-erasing}.  Note that reduction is stable by substitution and
replacement.

\ignore{More precisely,
\begin{lem}\mbox{}
  \label{l:stability}
  \begin{itemize}
  \item If $o \Rew{} o'$, then $o \isubs{x/u}
    \Rew{} o' \isubs{x/u}$ and $o \ire{\al}{u} {\Rew{}} o' \ire{\al}{u}$.
  \item If $u \Rew{} u'$, then $o \isubs{x/u} =  o \isubs{x/u'}$
    if $x \notin \fv{o}$ and $o \isubs{x/u} \Rewplus{}  o \isubs{x/u'}$
    if $x \in \fv{o}$.
  \item If $u \Rew{} u'$, then $o \ire{\al}{u} = o' \ire{\al}{u}$
    if $\al \notin \fn{o}$ and $o\ire{\al}{u} {\Rewplus{}} o' \ire{\al}{u}$
    if $\al \in \fn{o}$.
  \end{itemize}
\end{lem}

These stability properties give the following corollary.
\begin{cor}
\label{c:o-subs-sn-implique-o-sn}
If $o\isubs{x/u} \in \SNSP$ (resp. $o\ire{\al}{u}\in \SNSP$) , then $o \in \SNSP$.
\end{cor}
}

A \deft{head-context} is a context defined by the following grammar:
\begin{center}
$\begin{array}{llll}
   \HNO  & ::= & \HNT \mid \HNCom \\ 
   \HNT & ::= & \Box\, t_1\ldots t_n~ (n\geqslant 0) \mid \l x. \HNT \mid  \mu \al. \HNCom \\
   \HNCom & ::= & \co{\al} \HNT  \\
\end{array} $
\end{center}
A  \deft{head-normal form} is an object of the form $\HNO[x]$, where $x$
is any variable replacing the constant $\Box$. Thus for example $\mu \al. \co{\beta} \l y. x (\l z. z)$ is a head-normal form.
An object $o \in \objects{\lmu}$ is said to 
be \deft{head-normalizing}, written $o \in\HN{\lmu}$, if $o \Rewn{\lmu} o'$,
for some head-normal form $o'$. 
Remark that $o \in\HN{\lmu}$ does not imply $o \in\SN{\lmu}$ while the
converse necessarily holds. We write $\HN{\l}$ and $\SN{\l}$ when $t$ is
restricted to be a $\l$-term and the reduction system is restricted to
the $\beta$-reduction rule.

 A redex $r$ in an object of the form $o= \HNO[r]$ is called the
 \deft{head-redex} of $t$.  The reduction step $o \Rew{\lmu} o'$
 contracting the head-redex of $o$ is called a \deft{head-reduction
   step}.
 The \deft{head reduction strategy}  is a  deterministic
  strategy on the set $\objects{\lmu}$
 performing always head-steps\modifref{, so that the head reduction strategy only stops on head normal forms}.
  If the head-strategy starting at $o$
 terminates, then $o \in \HN{\lmu}$, while the converse direction is
 not straightforward (\cf\ Theorem~\ref{t:hn-lmu}).


\section{Quantitative Type Systems for the $\l$-Calculus}
\label{s:types-lambda}

As mentioned before, our results rely on typability of $\lmu$-terms in
suitable systems with non-idempotent types. Since the $\lmu$-calculus
embeds the $\l$-calculus, we start by recalling the
well-known~\cite{Gardner,Carvalho07,BKRDR14} quantitative type systems
for $\l$-calculus, called here $\Hl$ and $\Sl$. \modifref{We then reformulate
them, using a different syntactical view,} resulting in the
typing systems $\Hpl$ and $\Spl$, that are subsumed by the formalisms we adopt in Section~\ref{s:types-lambda-mu} for $\lmu$.

We start by fixing a countable set of \deft{base types} $\bta, \btb, \btc \ldots$,
then we introduce two different categories of types specified by the following grammars:
$$\begin{array}{llll}
(\mbox{{\bf Intersection Types}})\ & \IM       & :: = & \mult{\sigk}_{\kK}  \\
(\mbox{{\bf Types}})\           & \sig, \tau & ::= &  \bta \mid \IM \ftype  \sig  \\
\end{array}$$
An intersection type $\mult{\sigk}_{k \in \cset{1..n}}$ is a {\it
  multiset} that can be understood as a type $\sig_1 \cap \ldots \cap
\sig_n$, where $\cap$ is associative and commutative, but {\it
  non-idempotent}. Thus, $\emul$ is the empty intersection type.
\modifrefb{\label{text:non-deter-rel-for-inter} The \deft{
    non-deterministic} \deft{choice} \emph{relation} $\calR^*$ between
  intersection types is defined by $\msigk \calR^* \msigk$ when $K\neq
  \es$ and $\msigk \calR^* \mtau$ (for any type $\tau$) when $K=\es$.
  By making a slight abuse of notation, this choice relation is going
  to be  used as a non-deterministic \emph{operation}
  $\choice{\_}$ (rather than a relation $\calR^*$) as follows:
    \begin{center}
     $\choice{\mult{\sigk}_{\kK}} := \left \{ \begin{array}{ll}
                                            [\tau] & \mbox{ if } K = \es \mbox{ and $\tau$ is any arbitrary type} \\
                                              \mult{\sigk}_{\kK} & \mbox{  if } K \neq \es 
                                            \end{array} \right. $
\end{center}
  }

     \modifref{\deft{Variable assignments} (written $\Gam$) are total functions
     from variables to intersection types. We may write $\emptyset$ to
     denote the variable assignment that associates the empty
     intersection type $\emul$ to every variable.  The \deft{domain of
       $\Gam$} is given by
     $\dom{\Gam} := \cset{x \mid \Gam(x) \neq \emul}$, so that when
     $x\notin\dom{\Gam}$, then $\Gam(x)$ stands for $\mult{\,
     }$. We write $x:\IM$  (even when $\IM=\emul$)
       for the assignment mapping $x$ to $\IM$ and all $y\neq x$ to
       $\emul$.  We write $\Gam \inter \Gam'$ for
     $x\mapsto \Gam(x) + \Gam'(x)$, where $+$ is multiset union, so
     that $\dom{\Gam \inter \Gam'}=\dom{\Gam}\cup\dom{\Gam'}$.
     When $\dom{\Gam}\cap \dom{\Gam'}$, then $\Gam\inter \Gam'$
     may also
       be written as  $\Gam;\Gam'$.  We write $\cmin{\Gam}{x}$ for the
     assignment defined by $(\cmin{\Gam}{x})(x) = \mult{\,}$ and
     $(\cmin{\Gam}{x})(y) = \Gam(y)$ if $ y \neq x$.}

To  present/discuss different typing systems, 
we consider the following derivability notions.
A \deft{type judgment} is a triple
$\Gam \vdash t:\sig$, where $\Gam$ is a variable assignment,
$t$ a term and $\sig$ a type.  
A \deft{(type) derivation} in system $\X$ is a tree obtained by
applying the (inductive) rules of the type system $\X$, \modifref{where each node corresponds to the application of some typing rule}.  We write 
$\Phi \tri_\X \Gam \vdash t: \sig$ if $\Phi$ is a type derivation 
concluding with the type judgment $ \Gam \vdash t:\sig$,  and just $\tri_\X\ \Gam \vdash t: \sig$ if there exists  $\Phi$ such that $\Phi \tri_\X \Gam \vdash t: \sig$.  A term $t$
is \deft{$\X$-typable} iff there is a derivation in $\X$ typing $t$, \ie\ if
there is $\Phi$ such that $\Phi \tri_\X \Gam \vdash t: \sig$. We may
omit the index $\X$ if the name of the system is clear from the
context.

\subsection{Characterizing Head $\beta$-Normalizing $\l$-Terms}

We discuss in this section typing systems being able to characterize
head $\beta$-normalizing $\l$-terms. 
We first consider 
system $\Hl$ in Figure~\ref{fig:lambda-head}, 
first appearing in~\cite{Gardner}, then in~\cite{Carvalho07}.

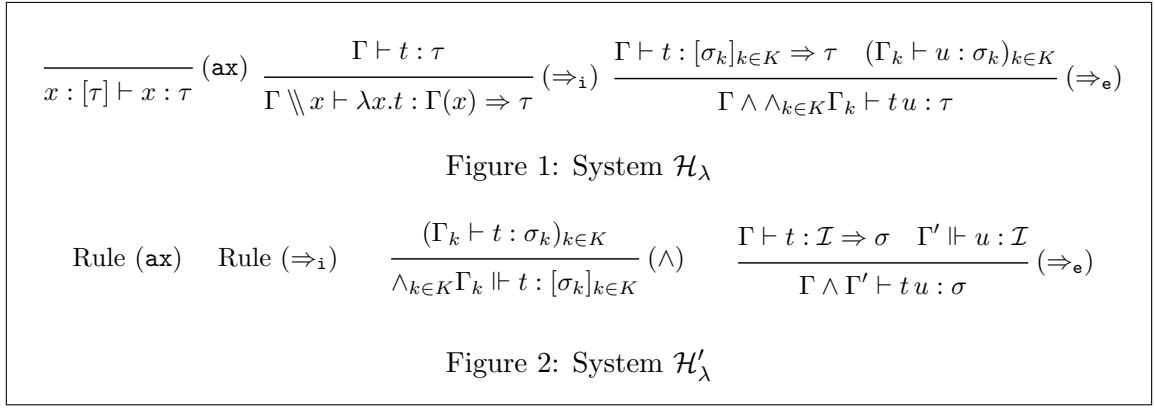
\begin{figure}[h]
\begin{framed}
{\small
\[ \begin{array}{c} 
\infer[ (\ax)]{ \phantom{} }{\tyj{x}{x:\mult{\tau}}{\tau}} \msep 
\infer[ (\introarrow) ]{ \tyj{t}{\Gam}{\tau}  }{ \tyj{\lambda{x}.t}{\cmin{\Gam}{x}}{\ty{\Gam(x)}{\tau}}}
     \msep
\infer[(\app)]{\tyj{t}{\Gam}{\ty{\mult{\sigk}_{\kK}}{\tau}}  \quad 
       (\tyj{u}{\Gamk}{\sigk})_{\kK}}
       {\tyj{t\,u}{\Gam \modifref{\inter \inter_{\kK}} \Gamk }{\tau}}
\end{array} \]
\caption{System $\Hl$}
\label{fig:lambda-head}
\[ \begin{array}{c}
\mbox{ Rule } (\ax)  \quad 
\mbox{ Rule } (\introarrow) \qquad 
\infer[(\many)]{(\Gamk \vdash t: \sigk )_{\kK}}
          {\inter_{\kK}\Gamk \Vdash t: \mult{\sigk}_{\kK} }  \qquad 
\infer[ (\app)  ]{ \Gam \vdash t: \IM \ftype  \sig  \quad 
       \Gam'  \Vdash u :  \IM  } 
       {\Gam \inter \Gam' \vdash t\,u:\sig } 
  \end{array} \]  
}
\caption{System $\Hpl$}
\label{fig:intermediate-lambda-head}
\end{framed}
\end{figure}

Notice that $K=\es$ in rule $(\app)$  allows to type an application
$t\,u$ without necessarily typing the subterm $u$. Thus, 
if $\Omega = (\l x.xx)(\l x.xx)$, then 
from the judgment $x:\mult{\sig}\vdash x:\sig$ we can derive  for example
$x:\mult{\sig} \vdash (\l y. x) \Omega: \sig$.

System $\Hl$ characterizes head $\beta$-normalization:

\begin{thm}
\label{l:typable-hn-lambda}
Let $t \!  \in \terms{\lambda}$.  Then
$t$ is $\Hl$-typable iff $t \! \in \HN{\l}$ iff the head-strategy
terminates on $t$. 
\end{thm}

Moreover, the implication {\it typability implies termination of the head-strategy} can be shown
by simple arithmetical arguments  provided by the {\it quantitative}
flavour of the typing system $\Hl$, in contrast to classical 
reducibility arguments usually invoked in other cases~\cite{GiraFont,Krivine93}. 
Actually, the arithmetical arguments 
give the  following quantitative property:

\begin{thm}
    If $t$ is $\Hl$-typable
with tree derivation $\Pi$, then the \textit{size} (number of nodes) of $\Pi$ 
gives an upper bound to the length of the head-reduction strategy starting at $t$. 
\end{thm}

To reformulate system $\Hl$ in a different way, we now
distinguish two sorts of judgments: \deft{regular judgments} of the
form $\Gam \vdash t:\sig$ assigning {\it types} to terms, and
\deft{auxiliary judgments} of the form $\Gam \Vdash t:\IM$ assigning
     {\it intersection types} to terms.

An equivalent formulation of system $\Hl$, called $\Hpl$, is given in
  Figure~\ref{fig:intermediate-lambda-head}, (where we always use the name $(\app)$ for the rule typing the application term, even if the rule is different from that in system $\Hl$).  There are two inherited forms of type derivations:  \deft{regular} (resp. \deft{auxiliary}) \deft{derivations} are those that conclude with regular (resp. auxiliary) judgments.  
\modifref{Notice that $K=\es$ in rule $(\many)$ gives $\es \Vdash u: \emul$ for {\it any  term $u$, 
\eg\ $\es \Vdash \Omega:\emul$, so
that one can also derive $x:\mult{\tau} \vdash (\l y.x)\Omega:\tau$ in this system.} Notice also  that systems $\Hl$ and $\Hpl$ are {\it relevant}, 
\ie\ they lack weakening.} Equivalence between $\Hl$ and  $\Hpl$ gives the following result: 
\begin{cor}
\label{l:typable-hn-lambda-bis}
Let $t \!  \in \terms{\lambda}$.  Then
$t$ is $\Hpl$-typable iff $t \! \in \HN{\l}$
iff the head-strategy
terminates on $t$. 
\end{cor} 
Auxiliary judgments turn out to substantially lighten the notations and 
to make the statements (and their proofs) more readable.

\subsection{Characterizing \modifref{Strongly $\beta$-Normalizing} $\l$-Terms}
\label{s:strong-lambda}

We now discuss typing systems being able to characterize \modifref{strongly $\beta$-normalizing} $\l$-terms. We first consider system $\Sl$ in
Figure~\ref{fig:lambda-strong}, which appears
in~\cite{BucciarelliKesnerVentura} (slight variants appear
in~\cite{DebeneRonchiITRS12,Bernadet-Lengrand2013,KV14}).  Rule
$(\appu)$ forces the {\it erasable argument} (the subterm $u$) to be
typed, even if the type of $u$ (\ie\ $\sig$) is not being used in the
conclusion of the judgment.  Thus, in contrast to system $\Hl$, every
subterm of a typed term is now typed.
\begin{figure}[h]
\begin{framed}
{\small 
\[  \begin{array}{c}
\infer[(\ax)]{ \phantom{} }{\tyj{x}{x:\mult{\tau}}{\tau}} \quad 
\infer[(\introarrow) ]{ \tyj{t}{\Gam}{\tau}  }{
  \tyj{\lambda{x}.t}{\cmin{\Gam}{x}}{\ty{\Gam(x)}{\tau}}}\quad 
\infer[(\appu)]{\tyj{t}{\Gam}{\ty{\mult{\;}}{\tau}}  \quad  
        \tyj{u}{\Del}{\sig} } 
      {\tyj{t\,u }{\Gam  \inter   \Del}{\tau}} \\[6mm] 
\infer[(\appd)]{\tyj{t}{\Gam}{\ty{\mult{\sigk}_{\kK}}{\tau}}  \quad 
       (\tyj{u}{\Delk}{\sigk})_{\kK} \quad 
       K  \neq \es}
      {\tyj{t\,u}{\Gam  \modifref{\inter  \inter_{\kK}} \Delk }{\tau}}
\end{array}\]
\caption{System $\Sl$}
\label{fig:lambda-strong}
\[ \begin{array}{c}
\mbox{ Rule } (\ax) \quad 
\mbox{ Rule } (\introarrow)  \quad 
\infer[(\many)]{(\Gamk \vdash t: \sigk)_{\kK} 
                                    }
          {\inter_{\kK}\Gamk \Vdash t: \mult{\sigk}_{\kK} } \msep 
\infer[(\appet)]{\tyj{t}{\Gam}{\ty{\IM}{\tau}}  \msep 
       \Del \Vdash u: \choice{\IM}  } 
      {\tyj{tu}{\Gam \inter \Del }{\tau}}
   \end{array}  \]    
}
\caption{System $\Spl$}
\label{fig:intermediate-lambda-strong}
\end{framed}
\end{figure}
System $\Sl$  characterizes strong $\beta$-normalization: 
\begin{lem}
\label{l:typable-sn-lambda}
Let $t \in \terms{\lambda}$.  Then
$t$ is $\Sl$-typable iff $t \in\SN{\l}$. 
\end{lem}

As before, the implication {\it typability implies normalization} can be show
by simple arithmetical arguments provided by the {\it quantitative}
flavour of the typing system $\Sl$. 
Actually, the \textit{proof} of Lemma~\ref{l:typable-sn-lambda} gives the following bound:
\begin{thm}
    If $t$ is $\Sl$-typable
with tree derivation $\Pi$, then the \textit{size} (number of nodes) of $\Pi$ 
gives an upper bound to the length of any reduction path starting at $t$. 
\end{thm}

An equivalent formulation of system $\Sl$,
called $\Spl$, is given in
Figure~\ref{fig:intermediate-lambda-strong}. As before, we use 
regular as well as auxiliary judgments.  Notice that \modifref{$K=\es$ in rule
$(\many)$} is still possible, but derivations of the form $\Vdash t:
\emul$, representing untyped terms, will never be used.  The 
choice operation $\choice{\_}$ (defined at the beginning of Section~\ref{s:types-lambda}) in
 rule $(\appet)$ is used to impose arbitrary types for erasable terms,
\ie\ when $t$ has type $\emul \ftype  \tau$, then $u$ needs to be typed
with an arbitrary type $\mult{\sig}$, thus the auxiliary judgment typing $u$ on the right premise of $(\appet)$ can never assign $\emul$ to $u$. This should be understood as a form of \emph{restricted} weakening on the argument $u$, 
, which is only
  performed \emph{before} the application of rule $(\appet)$ and only
  when the functional type is of the form $\emul \rew \UM$. 
  \begin{example}
  \label{ex:derivation}
Here is an example of type derivation in system $\Spl$.
  $$ \infer[(\appet)]{\infer[(\introarrow)]{\infer[(\ax)]{ \phantom{} }
                      {x:\mult{\sig}\vdash x:\sig}}
               {x:\mult{\sig} \vdash \l y. x: \emul \ftype  \sig} \quad 
                \infer[(\many)]{\infer[(\ax)]{\phantom{} }
                             {z: \mult{\tau} \vdash z: \tau }}
                      {z: \mult{\tau} \Vdash z: \mult{\tau} }}
  {x:\mult{\sig}, z: \mult{\tau}  \vdash (\l y. x) z: \sig}$$
\label{pageofexample}
\end{example}

\noindent Since $\Sl$ and $\Spl$ are equivalent, we also have: 
\begin{cor}
\label{l:typable-sn-lambda-bis}
Let $t \!  \in \terms{\lambda}$.  Then
$t$ is $\Spl$-typable iff $t \! \in \SN{\l}$. 
\end{cor}


\section{Quantitative Type Systems for the $\lmu$-Calculus}
\label{s:types-lambda-mu}

We present in this section two quantitative systems for the
$\lmu$-calculus, systems $\Hlmu$ (Section~\ref{s:head-lambdamu}) and
$\Slmu$ (Section~\ref{s:strong-lambdamu}), characterizing, respectively,
head and strong $\lmu$-normalizing objects. Since $\l$-calculus is
embedded in the $\lmu$-calculus, then the starting points to design
$\Hlmu$ and $\Slmu$ are, respectively, systems $\Hpl$ and $\Spl$, introduced in Section~\ref{s:types-lambda}.

\subsection{Types}
\label{s:types-lmu}

We consider a countable set of \deft{base types} $\bta, \btb, \btc \ldots$
and the following categories of types: 
$$\begin{array}{llll}
(\mbox{{\bf Object  Types}})          & \Any & : = & \ComTyp \mid \UM \\
(\mbox{{\bf Command Type}})       & \ComTyp & := & \TypCom \\
(\mbox{{\bf Union  Types}})\        & \UM, \VM & ::= &  \umult{\sigk}_{\kK}\\ 
(\mbox{{\bf Intersection Types}})\ & \IM, \JM       & :: = & \mult{\UMk}_{\kK}  \\
    (\mbox{{\bf Types}})\           & \sig, \tau  & ::= &  
                                                         \bta \mid \IM \ftype  \UM  \\
(\mbox{{\bf Blind Types}})\           & \Blind & ::=  & \bta  \mid \mult{ \, } \ftype   \umult{\Blind}   \\
\end{array}$$
The constant $\TypCom$ is used to type commands, \ignore{Unknown types
  $\Empty$ are a sort of minimal types, having only functional
  information (the arrows).} union types to type terms, and
intersection types to type variables (thus left-hand sides of
arrows). \modifrefb{We assume $K$ to be any finite set.} Both $\mult{\sigk}_{k \in
  \cset{1..n}}$ and $\umult{\sigk}_{k \in \cset{1 .. n}}$ can be seen
as {\it multisets}, representing, respectively, $\sig_1 \cap \ldots
\cap \sig_n$ and $\sig_1 \cup \ldots \cup \sig_n$, where $\cap$ and
$\cup$ are both associative, commutative, but {\it non-idempotent}. We
may omit the indices in the simplest case: thus $\mult{\UM}$ and
$\umult{\sig}$ denote singleton multisets.  We define the operator
$\inter$ (resp. $\union$) on intersection (resp. union) multiset types
by : $ \mult{\UMk}_{\kK} \inter \mult{\VMl}_{\lL}:= \mult{\UMk}_{\kK}
+ \mult{\VMl}_{\lL}$ and $ \umult{\sigk}_{\kK} \union
\umult{\taul}_{\lL}:= \umult{\sigk}_{\kK} + \umult{\taul}_{\lL}$,
where $+$ always means multiset union.  The \deft{non-deterministic}
\deft{choice} \modifrefb{operation (resulting from a non-deterministic choice relation, as in Section~\ref{text:non-deter-rel-for-inter})} is now not only defined on intersection but also on   union types:
\begin{center} $\begin{array}{lll}
           \choice{\mult{\UMk}_{\kK}} & := & \left
           \{ \begin{array}{ll} \mult{\UM} & \mbox{ if } K = \es
             \mbox{ and $\UM\neq \eumul$ is any arbitrary non-empty
               union type} \\ \mult{\UMk}_{\kK} & \mbox{ if } K \neq
             \es
                                            \end{array} \right. \\
  \choice{\umult{\sigk}_{\kK}} & := & \left \{ \begin{array}{ll}
                                            \umult{\xi} & \mbox{ if } K = \es \mbox{ and $\xi$ is any arbitrary  blind type} \\
                                             \umult{\sigk}_{\kK} & \mbox{  if } K \neq \es 
                                            \end{array} \right.
\end{array}  $
       \end{center}       
  The choice operation for union type is defined so that (1) the empty
union cannot be assigned to $\mu$-abstractions (2) subject reduction
is guaranteed in system $\Hlmu$ for erasing steps $(\mu \al.c)u\Rew{\mu}
\mu \al.c$, where $\al \notin \fn{c}$.  We present concrete examples
in Section~\ref{sec:discussion} which illustrates the need of
non-empty union types and blind types to guarantee subject reduction.

The \deft{arity} of types and union multiset types is defined by
induction: for types $\sig$, if $\sig = \IM \ftype  \UM$, then
$\ar{\sig}:=\ar{\UM}+1$, otherwise, $\ar{\sig} := 0$; for  union
multiset types, $\ar{\umult{\sigk}_{\kK}}:= \Sigma_{\kK}\ \ar{\sigk}$.
Thus, the arity of a type  counts the number of its \textit{top-level} arrows. 
The \deft{cardinality of multisets} is defined by 
$| \mult{\UMk}_{\kK}| = | \umult{\sigk}_{\kK}|  := |K|$.

\deft{Variable assignments} (written $\Gam$), are, as before, total
functions from variables to intersection multiset types. Similarly,
\deft{name assignments} (written $\Del$), are total functions from
names to union multiset types.  \modifrefb{The \deft{domain of $\Del$} is given
by $\dom{\Del} := \cset{ \al \mid \Del(\al) \neq \eumul }$, where
$\eumul$ is the empty union multiset, so that when
$\al\notin\dom{\Del}$, $\Del(\al)$ stands for $\eumul$.}  We may write
$\emptyset$ to denote the name assignment that associates the empty
union type $\eumul$ to every name.
 We write $\Del \vee \Del'$ for $\alpha \mapsto \modifrefb{\Del(\alpha)\union \Del'(\alpha)}$, so that $\dom{\Del \union  \Del'}=\dom{\Del}\cup\dom{\Del'}$. \ignore{ We use the \deft{choice} operator $\choice{\Delta}{\al}$ 
which   denotes $\Del(\al)$ when $\al \in \dom{\Del}$,  and \emph{any} singleton
  union type of the form $\umult{\e}$, with $\e \in \Empty$, when  $\al
  \notin \dom{\Del}$. } 

 When $\dom{\Del}\cap \dom{\Del'}=\es$, we may write $\Del;\Del'$ for $\Del\union \Del'$. We write $\al:\UM$ (even when $\al=\eumul$) for the name assignment mapping $\al$ to $\UM$ and all $\beta\neq \al$ to $\eumul$.  
 We write $\cmin{\Del}{\al}$ for
the name assignment defined by $(\cmin{\Del}{\al})(\al) = \eumul$ and
$(\cmin{\Del}{\beta})(\al) = \Del(\beta)$ if $ \beta \neq \al$.

We now present our typing systems $\Hlmu$ and $\Slmu$,
both having   \deft{regular} (resp. \deft{auxiliary})
judgments  of the form $\Gam \vdash t:\UM \mid \Del$ 
(resp. $\Gam \Vdash t:\IM \mid \Del$), together with their
respective notions of regular and auxiliary derivations.
An important syntactical property they enjoy is that
both are \deft{syntax directed}, \ie\
for each (regular/auxiliary) typing judgment $j$ there is a {\it unique} 
typing rule whose conclusion matches the judgment $j$. This 
makes our proofs much simpler than those
arising with idempotent types, which are  
based on long generation lemmas (\eg\ ~\cite{Bernadet-Lengrand2013,vB11}).

\subsection{System $\Hlmu$ }
\label{s:head-lambdamu}

In this section we present a quantitative typing system for $\lmu$,
called $\Hlmu$, characterizing head $\lmu$-normalization. It can be
seen as a first intuitive step to understand the typing system
$\Slmu$, introduced later in Section~\ref{s:strong-lambdamu}, and
characterizing strong $\lmu$-normalization.  However, to avoid
  redundancy,  the properties of the two systems
are not described in the same way:
\begin{itemize}
\item For $\Hlmu$, we provide  informal discussions
  to explain  the main
  requirements needed to capture quantitative information in the presence of
  classical feature (names, $\mu$-redexes). We particularly 
  focus on the necessity of banning empty union types. We do not give the proofs of the
  properties of $\Hlmu$, because they are simpler than those of system
  $\Slmu$.
  \item For $\Slmu$, we provide a more compact presentation,
    since the main technical  key choices used for $\Hlmu$ are 
    still
    valid. However, full statements and proofs of the properties of
    $\Slmu$ are detailed.
\end{itemize}

The (syntax directed)  rules of the typing system $\Hlmu$ are presented in
Figure~\ref{fig:head-lambdamu}. 
\begin{figure}[h]
\begin{framed}
\begin{center}
{\footnotesize 
\[ \begin{array}{c}
\infer[(\ax)]{ \UM \neq \eumul  } 
      {x: \mult{\UM} \vdash  x: \UM \mid \emptyset }  \sep 
\infer[(\introarrow) ]{ \Gam \vdash t : \UM \mid \Del}  
      { \cmin{\Gam}{x} \vdash  \l x.t: \umult{\Gamma(x)\ftype  \UM} \mid \Del}\quad 
\infer[ (\muu)]{\Gam \vdash t: \UM \mid \Del 
      }
      {\Gam \vdash \co{\al} t: \TypCom \mid \Del \vee \{\al: \UM \}} \\ \\ 
\infer[ (\mud)]{\Gam \vdash \Com: \TypCom \mid \Del }
      {\Gam \vdash \mu \al. \Com: \choice{\Del(\al)} \mid \cmin{\Del}{\al} } \quad 
\infer[ (\many)]{(\Gamk \vdash t: \UMk \mid \Delk)_{\kK}}
          {\inter_{\kK}\Gamk \Vdash t: \mult{\UMk}_{\kK} \mid \union_{\kK} \Delk} \\ \\
\infer[(\app)]{ \Gamt \vdash t: \umult{\IMk \ftype  \UMk}_{\kK}  \mid \Delt \quad  
       \Gamu  \Vdash u :  \inter_{\kK}  \IMk  \mid \Delu } 
       {\Gamt \inter \Gamu \vdash tu: \vee_{\kK} \UMk \mid \Delt 
       \union \Delu}  \\ \\
  \end{array} \]    
}
\end{center}
\caption{System $\Hlmu$}
\label{fig:head-lambdamu}
\end{framed}
\end{figure}
Rule
   $(\app)$ is to be understood as a {\it logical admissible} rule: if
   union (resp. intersection) is interpreted as the $\orc$
   (resp. $\andc$) logical connective, then $ \orc_{\kK}\ (\IMk
   \ftype \UMk)$ and $(\andc_{\kK}\ \IMk)$ implies
   $(\orc_{\kK}\ \UMk)$.   As in the simply typed
   $\lmu$-calculus~\cite{Parigot92}, the $(\muu)$ rule saves a type $\UM$
   for the name $\al$, however, in our system, the corresponding name
   assignment $\Del \vee \cset{\al: \UM}$, specified  by means of 
   $\vee$, collects {\it all} the types that $\al$ has been
   assigned during the derivation. Notice that the $(\mud)$-rule is not
   deterministic since $\choice{\Del(\al)}$ denotes an arbitrary
   union type  when $\Del(\al)$ is $\eumul$, a
   technical requirement which   is discussed at the end
of the section. 

In the simply typed $\lmu$-calculus, $\mbox{\bf call-cc}= \l y.\mu
\al.\co{\al} y (\l x.\mu \beta. \co{\al} x)$ would be typed with
$((a\ftype b)\ftype a) \ftype a$ (Peirce's Law), so that the fact that
$\alpha$ is {\it used} twice in the type derivation would not be
explicitly materialized with simple types (same comment applies to
{\it idempotent} intersection/union types).  This makes a strong
contrast with the
derivation in Figure~\ref{fig:typing-callcc}, where $\Ua := \umult{\bta}$,
$\Ub := \umult{\btb}$,  $\Uy :=\umult{\mult{\umult{\mult{\Ua}\ftype  \Ub}} \ftype  \Ua}$, and
$\mbox{\bf call-cc}$ is typed  with the union type
$\umult{\mult{  \umult{\mult{\umult{\mult{\Ua}\ftype  \Ub}} \ftype  \Ua} } \ftype  (\Ua \union  \Ua)}$.
  
\begin{figure}[ht]
\begin{framed}
  \begin{center}
   $
\infer[(\introarrow)]{\infer[(\mud)]{\infer[(\muu)]{\infer[(\app)]{ \infer[(\ax)]{}{\muju{y:\mult{\Uy}}{y:\Uy}{\es}}  \sep
       \infer[(\many)]{\infer[(\introarrow)]{\infer[(\mud)]{\infer[(\muu)]{\infer[(\ax)]{} {\muju{x:\mult{\Ua}}{x:\Ua}{\es}}}
                                                { \muju{x:\mult{\Ua}}{\co{\al} x:\TypCom}{\al:\Ua} }}
                                         { \muju{x:\mult{\Ua}}{\mu \beta. \co{\al} x:\Ub}{\al:\Ua} }}
                                  { \muju{\phantom{.}}{\l x.\mu \beta. \co{\al} x:\umult{\mult{\Ua}\ftype  \Ub}}{\al:\Ua} }}
      {\muJu{}{\l x.\mu \beta. \co{\al} x:\mult{\umult{\mult{\Ua}\ftype \Ub}}}{\al:\Ua} }}
      {\muju{y:\mult{\Uy }}{y(\l x.\mu \beta. \co{\al} x):\Ua}{\al:\Ua} }}
      {\muju{y:\mult{\Uy }}{\co{\al} y(\l x.\mu \beta. \co{\al} x):\TypCom}{\al: \Ua \union \Ua} }}
      {\muju{y:\mult{\Uy }}{\mu \al. \co{\al} y(\l x.\mu \beta. \co{\al} x): \Ua \union  \Ua}{\es} }}
      {\muju{}{\l y. \mu \al. \co{\al} y(\l x.\mu \beta. \co{\al} x):
\umult{\mult{  \umult{\mult{\umult{\mult{\Ua}\ftype  \Ub}} \ftype  \Ua} } \ftype (\Ua \union \Ua)}}{\es} }  $ 
  \end{center}
\caption{Typing {\bf call-cc}}
\label{fig:typing-callcc}
\end{framed}
\end{figure}


This example suggests to distinguish two different uses of names:
\begin{itemize}
\item The name $\al$ is saved twice
  by a $(\muu)$ rule : once for $x$ and once for $ y (\l x.\mu
    \beta. \co{\al} x)$, both times with type $\Ua$. After that, the
  abstraction $\mu \al.\co{\al} y (\l x.\mu \beta. \co{\al} x)$ {\bf
    restores}  the union of the two types that 
 were previously stored by  $\al$ (by means of the two $(\muu)$-rules). A similar phenomenon occurs with rule $(\introarrow)$, which restores the types of the abstracted variables.
\item The name $\beta$ is not free in $\co{\al}x$, so that 
a  new union type $\Ub$ is introduced to type 
the abstraction
  $\mu \beta.\co{\al}x$. 
  From a logical point of view, this \modifref{corresponds to a 
  {\it   weakening} on the right hand-side of the sequent, which is necessary, for instance, to derive the non-idempotent counterpart of Peirce's Law (\cf\ Fig.~\ref{fig:typing-callcc}).}  Consequently,   $\l$ and $\mu$-abstractions are not treated symmetrically: when $x$ is not free in $t$,
  then $\l x. t$ will be typed with $\umult{\emul \ftype \sig}$ (where $\sig$ is
  the type of $t$), and  no new arbitrary intersection type
  is introduced for the abstracted variable $x$.  
      \end{itemize}

An interesting observation 
is about the  restriction of system $\Hlmu$ to the pure $\l$-calculus: union
types, name assignments and rules $(\mud)$ and $(\muu)$ are ruled out, 
  so that every union multiset takes the single form $\umult{\tau}$,
  which can be simply identified with $\tau$. Thus,  the
  restricted typing system $\Hlmu$ becomes system $\Hpl$ in
  Figure~\ref{fig:intermediate-lambda-head}.  

  \modifrefb{ Another observation is about the property of
    \textit{relevance} of assignments. Although there is a
    \emph{restricted} form of weakening in system
    $\Hlmu$\ (allowing for example to derive the non-idempotent
      counterpart of Peirce's Law, see the top $(\mud)$-rule in
      Fig.~\ref{fig:typing-callcc}), if $\muju{\Gam}{o:\Any}{\Del}$
    is derivable, then any $x\in \dom{\Gam}$ (resp. $\al\in
    \dom{\Del}$) has at least one free occurrence in $o$.}  Formally,
\begin{lem}[{\bf Relevance for System $\Hlmu$}]
\label{l:relevane-lambda-mu-weak}
Let $o \in \objects{\lmu}$. If  $\Phi \tri \Gam \vdash o: \Any \mid \Del$,
then $\dom{\Gam}\subseteq \fv{o}$ and $\dom{\Del} \subseteq \fn{o}$.
\end{lem}

\begin{proof} By induction on $\Phi$. \end{proof}

\modifref{
We define now our notion of \deft{derivation size}}, which is a natural number  representing the amount  of  information in  tree derivations.  For any type  derivation $\Phi$,  $\sz{\Phi}$ is  inductively defined by the following rules, where we use an abbreviated, self-understanding notation for the premises. \label{def:size-Hlmu}
{
$$
\begin{array}{lll}
   \sz{\infer[ (\ax)]{ }{x: \mult{\UM} \vdash  x: \UM \mid \emptyset} }  & := & 1 \\ \\
  \sz{\infer[ (\introarrow)]{ \Phi_t \rhd t}  
      { \cmin{\Gam}{x} \vdash  \l x.t: \umult{\Gamma(x)\ftype  \UM} \mid \Del}}  & := & \sz{\Phi_t} +1 \\ \\
  \sz{\infer[ (\muu)]{\Phi_t\rhd t }
    {\Gam \vdash \co{\al} t: \TypCom \mid \Del \vee \{\al: \UM \}}}  & := & \sz{\Phi_t} + \ar{\UM } \\ \\
  \sz{\infer[ (\mud)]{\Phi_{\Com}\rhd \Com }
      {\Gam \vdash \mu \al. \Com: \choice{\Del(\al)} \mid \cmin{\Del}{\al} } }  & := & \sz{\Phi_{\Com}} +1 \\ \\
  \sz{\infer[ (\many)]{(\Phik \rhd t)_{\kK}  }
          {\inter_{\kK}\Gamk\Vdash t: \mult{\UMk}_{\kK} \mid \union_{\kK} \Delk} }  & := & \Sigma_{\kK}\ \sz{\Phik} \\\\
  \sz{\infer[]{\Phi_t \rhd   t:\umult{\IMk\rew \UMk}_{\kK} \quad  \Phi_u\rhd u} 
      {\Gam  \vdash 
       tu: \vee_{\kK} \VMk \mid \Del }\ (\app) }  & := &  \sz{\Phi_t} + \sz{\Phi_u} + |K|   
\end{array}
$$}

System  $\Hlmu$ behaves as expected, in particular,
typing is stable by reduction (Subject Reduction)
and anti-reduction (Subject Expansion):

\begin{pty}[{\bf Weighted Subject Reduction for $\Hlmu$}]
\label{l:sr-Hlmu}
Let $\Phi \tri \tyj{o}{\Gam}{\Any\mid\Del}$. If $o \Rew{} o'$, then
there exists a derivation  $\Phi' \tri \tyj{o'}{\Gam}{\Any\mid\Del}$
such that $\sz{\Phi'} \leq \sz{\Phi}$. 
Moreover, if the reduced redex is typed, then  $\sz{\Phi'}<\sz{\Phi}$. 
\end{pty}

An important remark is that, if the arity of the types were not taken
  into account in the size of the rules $(\muu)$, then we would only have
  $\sz{\Phi'}=\sz{\Phi}$ \modifref{(and not $\sz{\Phi'} <  \sz{\Phi}$)} for
  the $\mu$-reduction steps. Intuitively, the $\mu$-reduction $(\mu
  \al.\Com)u\Rew{\mu} \mu \al.\Com\ire{\al}{u}$ dispatches the
  $(\app)$-rule typing the root of the $\mu$-redex $(\mu
  \al.\Com)u$ into several created
  $(\app)$-rules in the reduct, but \modifref{ neither an
  increase nor a decrease of the measure is ensured}. The solution to recover this
  key feature (\ie\ the decrease) is 
    suggested by the effect of $\mu$-reduction on the $(\muu)$-rules
  associated to $\al$ (see Figure~\ref{fig:effect-mu-red-naming}): indeed,
  $\mu$-reduction replaces every named term $\co{\al}v$ by
  $\co{\al}vu$, where $u$ is the argument of the $\mu$-redex, so
  that the saved types are smaller under the created $(\app)$-rules than
  the ones in the original derivation.  
 \begin{figure}
 $
 \infer[(\muu)]{ \muju{\ldots}{v:\umult{\IMk\rew \UMk}_\kK}{ \ldots} }
               { \muju{\ldots}{\co{\al}v:\TypCom}{\ldots \union \cset{\al:\umult{\IMk\rew \UMk}_{\kK}}}} $\\[0.4cm]
 \hspace*{4cm}
 {\large becomes}\\[0.2cm]
 
$ \hspace*{\fill}
 \infer[(\muu)]{\infer[(\app)]{ \muju{\ldots}{v:\umult{\IMk\rew \UMk}_\kK}{\ldots} \qquad
                                \muJu{\ldots}{u:\inter_{\kK} \IMk}{ \ldots}}
                              { \muju{\ldots}{v\,u:\union_\kK \UMk}{\ldots}}}
                              {\muju{\ldots}{\co{\al}v\,u: \TypCom}{\ldots \union\cset{\al:\union_\kK \UMk}}  }
 $
 \caption{Effect of $\mu$-reduction on rule $(\muu)$}
 \label{fig:effect-mu-red-naming}
 \end{figure}

 \modifref{As expected from an intersection (and union)
type system, head normal forms are typable in system $\Hlmu$.} 
 Moreover,  subject expansion holds
for $\Hlmu$, meaning that typing is stable under
\textit{anti}-reduction. Note that we do not state a \textit{weighted}
subject expansion property (although this would be possible) only
because this is not necessary to prove the
  final characterization property of  system $\Hlmu$
  (\cf\ Theorem~\ref{t:hn-lmu}).

\begin{pty}[{\bf Subject Expansion for $\Hlmu$}]
\label{l:se}
Let $\Phi' \rhd \muju{\Gam'}{o':\A}{\Del'}$. If $o\Rew{} o'$, then there is $\Phi \rhd
\muju{\Gam'}{o:\A}{\Del'}$.
\end{pty}

Note in particular that the head strategy only reduces typed redexes 
(the head redex of a \textit{head reducible} typed term is necessarily  typed),
so that finally,  we can  state our type-theoretic characterization of head normalization 
for the $\lmu$-calculus:

\begin{thm}
  \label{t:hn-lmu}
  Let $o \in \objects{\lmu}$. Then $o$ is $\Hlmu$-typable iff $o \in
  \HN{\lmu}$ iff the head-strategy terminates on $o$.  Moreover, if
  $o$ is $\Hlmu$-typable with tree derivation $\Pi$, then $\sz{\Pi}$
  gives an upper bound to the length of the head-reduction strategy
  starting at $o$.
\end{thm}
We do not provide the proofs of these properties and the last theorem, because it uses special
cases of the more general technology that we are going to develop
later to deal with strong normalization.  Notice however that
Theorem~\ref{t:hn-lmu} ensures that the head-strategy is complete for 
head-normalization in $\lmu$.

 \subsection{Discussion}
\label{sec:discussion}
 Now that we have stated the main result of this section, 
based on the key (weighted) subject reduction property, 
let us come back to the design choices of our development, in
particular: 

\begin{itemize}
\item The choice $\choice{\_}$ operator in rule  $(\mud)$ requires the types to be blind, and
\item The  union types are non-empty. 
\end{itemize}

\modifref{Blind types raised in rule $(\mud)$ guarantee that terms do
  not have empty union types, the necessity of which is a delicate
  constraint discussed later. In the
  meantime, we start by justifying the particular form they adopt
  (arrows with empty domains), which is easier to justify. Thus,
  consider the following example:} let $t_1=\mu \beta.\co{\gamma} x$
with $\beta \neq \gamma$ and $x\neq y$, so that $t_1\,y \Rew{\mu}
t_1$. A typing derivation of $t_1$ necessarily concludes with
$\muju{x:\mult{\VM_x}}{t_1:\choice{\eumul}}{\gamma:\VM_x}$, for some
union type $\VM_x$. Let us assume, only temporarily, that a non-blind
type can be chosen by the non-deterministic operator in rule $(\mud)$
\eg\ $\choice{\eumul}=\umult{\mult{\UM}\ftype \UM}$ for some union
type $\UM$. If we then assign $\UM$ to $y$, the judgment
$\muju{y:\mult{\UM}; x:\mult{\VM_x}}{t_1\;y:\UM}{\gamma:\VM_x}$ is
derivable by rule $(\app)$. However, by relevance
(Lemma\;\ref{l:relevane-lambda-mu-weak}), the judgment
$\muju{y:\mult{\UM}; x:\mult{\VM_x}}{t_1:\UM}{\gamma:\VM_x}$ cannot be
derivable since $y\notin \fv{t_1}$. Thus, subject reduction simply
fails.  Note that if $\choice{\eumul}$ chooses a blind type, for
example $\umult{\emul \ftype \umult{\bta}}$, then $y$ is untyped in
the derivation of $t_1 y$,
\ie\ $\muju{x:\mult{\VM_x}}{t_1\;y:\UM}{\gamma:\VM_x}$, so that
subject reduction holds.  \\

  We now explain why union types are non-empty.
  In particular, 
there are two different typing rules  requiring union
  types to be non-empty: rule $(\ax)$ and rule $(\mud)$.  To illustrate the necessity of non-empty union
types\label{disc:u-types-ne} for $\mu$-abstractions, \ie\ rule $(\mud)$, let us
assume, again temporarily, that an empty union type is introduced by
the rule $(\mud)$ when $\alpha \notin \dom{\Del}$, and call such an
instance $(\mudempty)$ when this happens.  \modifref{Similarly, let us also call the instance of an application rule to be $(\appempty)$ when the union type of its
left-hand side is empty.}
\begin{center}
$\infer[ (\mudempty)]{\Gam \vdash \Com: \TypCom \mid \Del\ \sep \alpha \notin \dom{\Del} }
      {\Gam \vdash \mu \al. \Com: \eumul \mid \Del } \quad \quad
\infer[ (\appempty)]{ \Gam \vdash t: \eumul  \mid \Del \quad  
       \es   \Vdash u : \emul    \mid \es  } 
       {\Gam  \vdash t u: \eumul \mid \Del}$
\end{center}

 For instance, given $t_1:=\mu\beta.\co{\gamma} x$, where $\beta \neq \gamma$,  every derivation typing $t_1$
concludes with rule $(\mudempty)$ typing  a judgment of the form
$\muju{x:\mult{\VM_x}}{t_1:\eumul}{\gamma:\VM_x }$.

Assume now that a derivation $\Pi_\Com$ typing a command $\Com$
contains $k$ empty rules $(\mudempty)$ w.r.t. the name $\al$, and that
$\Pi_\Com$ is a subderivation of some other derivation typing a
$\mu$-redex $(\mu \al.\Com)u$ for some term $u$. Let us give an
example with $k=2$ by setting $\Com := \co{\al} t_2$, where  
$t_2:=\mu \beta'. \co{\al}t_1$ and $t_1:= \mu\beta.\co{\gamma} x$ as before. 
Then, the (empty union) types of the terms $t_1$
and $t_2$ are saved by $\al$, simply because $\beta$ (resp. $\beta'$)
does not occur free in the body of $t_1$ (resp. $t_2$). Thus, while
typing command $\Com$, the name $\al$ is necessarily typed with
$\eumul \union \eumul= \eumul$.

Formally, let $\VM_x$ be any (non-empty) union type. Then, 
\begin{center}
$\infer[(\muu)]
       {\infer[(\mudempty \mbox{ for } \beta')]
              {\infer[(\muu)]
                     {\infer[(\mudempty \mbox{ for } \beta)]
                            {\infer[(\muu)]{\vdots}
                                   {x:\mult{\VM_x} \vdash \co{\gamma} x: \TypCom  \mid \gamma: \VM_x}}
                            {x:\mult{\VM_x} \vdash t_1:  \eumul \mid \gamma: \VM_x}}
                     {x:\mult{\VM_x} \vdash \co{\al} t_1: \TypCom  \mid \gamma: \VM_x}}
              {x:\mult{\VM_x} \vdash \mu \beta'. \co{\al} t_1:\eumul  \mid \gamma: \VM_x}}
       {x:\mult{\VM_x} \vdash  \co{\al} (\mu \beta'. \co{\al} t_1):\TypCom   \mid \gamma: \VM_x}
         $
\end{center}

Let us now use  $\Com = \co{\al} t_2$ inside the 
$\mu$-reduction: 
$t=(\mu \al.\co{\al} t_2)u\Rew{\mu} \mu \al.\co{\al} (t_3 u) = t'$, 
where  $t_3:= \mu
\beta'. \co{\al} (t_1 u)$. Then  $u$ is necessarily typed as follows:
$$\infer[(\many)]{} {\es \vdash u: \emul \mid \es }$$ Let us call
$\Pi$ (resp. $\Pi'$) the type derivation of $t$ (resp $t'$). Note that
the $\mu$-reduction transforms each rule $(\muu)$ of $\Pi$ into a rule
$(\appempty)$ followed by $(\muu)$ in $\Pi'$ (see
Figure~\ref{fig:effect-mu-red-on-empty-types}).  Thus, if $\Pi_\Com$
contains $k$ rules $(\muu)$ introducing $\al$, then the
derivation obtained by subject reduction typing $\Com\ire{\al}{u}$
contains $k$ {\it new} rules $(\appempty)$, each followed by a rule
$(\muu)$.  Indeed, one may check in the example above that the
derivation $\Pi'$ contains {\it two} rules $(\appempty)$ and {\it two}
rules $(\muu)$ introducing $\alpha$, whereas
$\Pi$ contains {\it two} rules $(\muu)$ introducing $\alpha$, and just one rule $(\appempty)$ (the one typing the root of the redex).

In general, one can check that whatever the size definition is for
rules $(\app)$, $(\mudempty)$ and $(\appempty)$, the derivation size
cannot decrease for any choice of $k\geqslant 0$: indeed, if rule
$(\mudempty)$ is used, the (possibly empty) type of a term $\mu
\al.\Com$ holds no information capturing the number of free
occurrences of $\al$ in the command $\Com$, so that there is no local
way to know how many times the argument $u$ should be typed in the
whole derivation of the term $(\mu \al.\Com)u$ (compare
Figure~\ref{fig:effect-mu-red-on-empty-types} to
Figure~\ref{fig:effect-mu-red-naming}).

The reader will notice that the same kind of phenomenon occurs, when
in the previous example, the term $t_1$ is replaced with a variable
$x$ of type $\eumul$: then $t= (\mu \al. \co{\al}(\mu \beta'. \co{\al}
x))u \Rew{\mu} \mu \al. \co{\al}((\mu \beta'. \co{\al} (xu))u) = t'$,
the type derivation of $t'$ contains {\it two} rules $(\appempty)$
whereas the type derivation of $t$ contains just one rule
$(\appempty)$. This explains why one cannot assign the empty
union  type in the  rule $(\ax)$.

All these arguments to forbid the empty union type in our system are not only valid for system $\Hlmu$, but  also apply to the system $\Slmu$, introduced later in Section.~\ref{s:strong-lambdamu}.

\ignore{As we saw in Figure\;\ref{fig:sr-Hlmu}, during a  $\mu$-reduction step $(\mu \al.\Com)u\Rew{\mu} \mu \al.\Com\ire{\al}{u}$, the total size of the $\app$-rules is preserved (it does not decrease), because the $\app$-rule (and its weight) at the root of the redex is dispatched in created $\app$-rules in the reduct. Concerning ``empty rules'', note that $\mu$-reduction transforms an empty naming rule (w.r.t. $\al$) into an empty $\app$-rule followed by an empty naming rule (see Figure~\ref{fig:effect-mu-red-on-empty-types}).}

 \begin{figure}
 \begin{center}
 $
 \infer[(\muu)]{\muju{\ldots}{v:\eumul}{\ldots}}{\muju{\ldots}{\co{\al}v:\TypCom}{\ldots}} \hspace{1cm}\text{becomes}\hspace{1cm}
 \infer[(\muu)]{\infer[(\app)]{\muju{\ldots}{v:\eumul}{\ldots} \sep  \muJu{\ldots}{u:\emul}{\ldots}}{\muju{\ldots}{vu:\eumul}{\ldots}}}{\muju{\ldots}{\co{\al}v\,u:\TypCom}{\ldots}}
 $
 \caption{Effect of $\mu$-reduction on empty types}
 \label{fig:effect-mu-red-on-empty-types}
 \end{center}
 \end{figure}

\subsection{System $\Slmu$ }
\label{s:strong-lambdamu}

This section presents a quantitative typing system  
characterizing strongly normalizing $\lmu$-terms. 
The  (syntax directed)  typing rules of the  system  $\Slmu$ 
appear in  Figure~\ref{fig:strong-lambdamu}. 
\begin{figure}[h]
\begin{framed}
\begin{center}
{\footnotesize 
\[ \begin{array}{c}
\infer[(\ax)]{ \UM \neq \eumul } 
      {x: \mult{\UM} \vdash  x: \UM \mid \emptyset }  \sep 
\infer[(\introarrow)]{ \Gam \vdash t : \UM \mid \Del}  
      { \cmin{\Gam}{x} \vdash  \l x.t: \umult{\Gam(x)\ftype  \UM} \mid \Del}  \quad 
\infer[ (\muu)]{\Gam \vdash t: \UM \mid \Del 
      }
      {\Gam \vdash \co{\al} t: \TypCom \mid \Del \vee \{\al: \UM \}} \\ \\  
\infer[ (\mud)]{\Gam \vdash \Com: \TypCom \mid \Del }
      {\Gam \vdash \mu \al. \Com: \choice{\Del(\al)} \mid \cmin{\Del}{\al} }  \quad 
\infer[ (\many) ]{(\Gamk \vdash t: \UMk \mid \Delk)_{\kK} 
       }
      {\inter_{\kK}\Gamk \Vdash t: \mult{\UMk}_{\kK} \mid \union_{\kK} \Delk}\\ \\
\infer[(\appet)]{ \Gamt \vdash t: \umult{\IMk \ftype  \UMk}_{\kK}  \mid \Delt \sep  
       \Gamu  \Vdash u : \inter_{\kK}\ \choice{\IMk}  \mid \Delu} 
       {\Gamt \inter \Gamu \vdash tu: \union_{\kK} \UMk \mid \Delt  \union  \Delu}
  \end{array} \]  
}
\end{center}
\caption{System $\Slmu$}
\label{fig:strong-lambdamu}
\end{framed}
\end{figure}

 As in system $\Spl$, 
the  non-deterministic choice operator $\choice{\_}$ is used to choose arbitrary types for
erasable terms, so that no subterm
is untyped, thus ensuring strong $\lmu$-normalization.  The use of $\choice{\_}$ in  the $(\mud)$-rule (raising a non-empty \textit{union} type) can 
be seen as a {\it weakening} on the name $\al$ (on the right hand-side of
  the sequent) followed by a $\mu$-abstraction,
 while its use in rule  $(\appet)$ (raising a non-empty \textit{intersection} type) should be seen
{\it  weakening} on the left-hand side (domain) of an arrow type.   We still consider 
the definition of size given before, as 
the choice operator does not play any particular new role here. 

As in system $\Hlmu$, every term is typed with a non-empty union type:
\begin{lem}
\label{l:non-empty-union} 
     If $\tri \muju{\Gam}{t:\UM}{\Del}$, then $\UM \neq \eumul$.
\end{lem}

As well as in the case of $\Hlmu$, system $\Slmu$
can be restricted to the pure $\l$-calculus.  By the
same observations made at the end of Section~\ref{s:head-lambdamu},  we
get 
the typing system $\Spl$ in
Figure~\ref{fig:intermediate-lambda-strong} that characterizes
\modifref{strong  $\beta$-normalization}.

Relevance in  system $\Slmu$ is  stated as follows:  
\begin{lem}[{\bf Relevance for System $\Slmu$}]
\label{l:relevance}
Let $o \in \objects{\lmu}$. If $\Phi \tri \Gam \vdash o: \Any \mid \Del$,
then $\dom{\Gam}=\fv{o}$ and $\dom{\Del}=\fn{o}$.
\end{lem} 

\begin{proof}
By induction on $\Phi$.
\end{proof}

Note the difference between Lemma~\ref{l:relevane-lambda-mu-weak}
  (inclusion) and Lemma~\ref{l:relevance} (equality): this is because in
  $\Slmu$ we use a choice operator in the $(\appet)$-rule to prevent
  any subterm of a typed term to be left untyped.

  The definition of $\sz{}$ is extended to system $\Slmu$ as expected. 
\modifref{Hence,   
 $\sz{\Phi} \geq 1$ holds}
for any {\it regular} derivation $\Phi$, whereas, by definition, the
derivation of the empty auxiliary judgment $\es \Vdash t:
\emul \mid \es$ has size $0$.


\section{Typing Properties}
\label{s:properties}

This section shows two fundamental properties of reduction (\ie\ 
{\it forward}) and anti-reduction (\ie\ {\it backward}) of
system $\Slmu$. In
Section~\ref{s:forward}, we analyze the {\it subject reduction (SR)} property, and we
prove that reduction preserves typing and decreases the
size of type derivations (that is why we call it weighted SR). The
proof of this property makes use of two fundamental properties
(Lemmas~\ref{l:substitution} and~\ref{l:replacement}) guaranteeing
well-typedness of the meta-operations of substitution and replacement.
Section~\ref{s:backward} is devoted to {\it subject expansion (SE)}, which
states that  non-erasing anti-reduction preserves types. The proof 
uses the fact that reverse substitution (Lemma~\ref{l:reverse-substitution})
and reverse replacement (Lemma~\ref{l:reverse-replacement}) preserve types.

We start by stating  an interesting property, to be used in our forthcoming lemmas,
 that allows us to split and merge auxiliary derivations typing the same term: 
  \begin{lem}
    \label{l:decomposition}
  Let $\IM = \inter_{\kK} \IMk $.
    Then $\Phi \tri \Gam \Vdash t:\IM \mid \Del$ iff 
    there exist $(\Gamk)_{\kK}, (\Delk)_{\kK}$ s.t. 
    $(\Phik \tri \Gamk \Vdash t:\IMk \mid \Delk)_{\kK}$,
    $\Gam = \inter_{\kK} \Gamk $ and  $\Del = \union_{\kK} \Delk$.
    Moreover, $\sz{\Phi} = \Sigma_{\kK} \sz{\Phik}$.
  \end{lem}

\subsection{Forward Properties}
\label{s:forward}

We first state the substitution lemma, which guarantees that typing is stable
by substitution. The lemma also establishes the size of the derivation tree
of a substituted object from the sizes of the derivations trees of its components.

\begin{lem}[{\bf Substitution}]
\label{l:substitution}
Let 
$\Theu \tri \Gamu  \Vdash u: \IM  \mid \Delu$.
If $\Phi_o \tri \Gam; x:\IM\vdash o:\Any \mid  \Del$, then 
   there is $\Phi_{o\isubs{x/u}}$ such that 
\begin{itemize}
\item  $\Phi_{o\isubs{x/u}}\rhd\Gam \inter \Gamu \vdash o\isubs{x/u}: \Any \mid \Del\union  \Delu $.
\item  $\sz{\Phi_{o\isubs{x/u}}}=\sz{\Phi_o} +   \sz{\Theu} - |\IM|$.
\end{itemize}
\end{lem}

\begin{proof}
By induction on $\Phi_o$ using Lemmas~\ref{l:relevance} and~\ref{l:decomposition}.
See the Appendix for details. 
\end{proof}


Typing is also stable by replacement. Moreover, we can specify the  exact size of the derivation tree of the replaced object from the sizes of its  components.

\begin{lem}[{\bf Replacement}]
\label{l:replacement}
Let $\Theu \tri \Gam_u \Vdash u : \inter_{\kK}\ (\choice{\IMk}) \mid \Del_u $ where 
$\al \notin \fn{u}$. If $\Phi_o \tri \tyj{o}{\Gam}{\Any \mid \al: \umult{ \IMk \ftype 
      \VMk}_{\kK} ; \Del}$, then there is $\Phi_{o\ire{\al}{u}}$ such that :
\begin{itemize}
\item $\Phi_{o\ire{\al}{u}} \tri \tyj{o\ire{\al}{u}}{\Gam \inter \Gamu}
                           {\Any \mid \al: \union_{\kK} \VMk; \Del \union \Delu }$. 
\item $\sz{\Phi_{o\ire{\al}{u}}} =  \sz{\Phi_o} + \sz{\Theu}$.  
\end{itemize}
\end{lem}

\begin{proof}
By induction on $\Phi$ using Lemmas~\ref{l:relevance} and~\ref{l:decomposition}.
See the Appendix for details.
\end{proof}

Notice that the type of $\al$ in the conclusion of the 
derivation $\Phi_{o\ire{\al}{u} }$ (which is $\union_{\kK} \VMk$) is strictly  smaller than
that of the conclusion of the derivation $\Phi_o$ (which is $\umult{ \IMk \ftype 
      \VMk}_{\kK}$) if and only if $K \neq \emptyset$. 


Lemmas~\ref{l:substitution} and \ref{l:replacement} are used 
in the proof of the following key property.

\begin{pty}[{\bf Weighted Subject Reduction for $\Slmu$}]
\label{l:sr}
Let $\Phi \tri \tyj{o}{\Gam}{\Any\mid\Del}$. If $o \Rew{} o'$
is a non-erasing step, then
there exists a derivation  $\Phi' \tri \tyj{o'}{\Gam}{\Any\mid\Del}$ 
such that  $\sz{\Phi} > \sz{\Phi'}$. 
\end{pty}

\begin{proof}
By induction on $o \Rew{} o'$ using Lemmas~\ref{l:relevance}, ~\ref{l:substitution}
and~\ref{l:replacement}. 
See the Appendix for details.
\end{proof}

\technicalreport{
\begin{proof}
By induction on the  relation $\Rew{}$. We only show the main cases of 
reduction at the root, the other ones being straightforward. 

\begin{itemize}
\item If $o = (\lambda x.t)u$, then $o' = t\isubs{x/u}$ and $x\in \fv{t}$.
The derivation $\Phi$ has the following form:

$${\small   \Phi = \infer{
     \infer*{\Phi_t\rhd \Gam_t ;  x:\IM  \vdash t:\UM~|~ \Del_t}
            {\Gam_t \vdash \lambda x.t: \umult{\IM \ftype \UM}  \mid  \Del_t } \\
             {\Theu \rhd  \Gam_u \Vdash u: \choice{\IM } \mid    \Del_u}  }
   {\Gam \vdash o:\UM  \mid  \Del } }
$$  where $\Gam = \Gam_t \inter \Gam_u$, 
       $\Del=\Del_t\union  \Del_u$. Indeed,  $x \in \fv{t}$ implies by
   Lemma~\ref{l:relevance} that $\IM \neq \emul$ so that $\choice{\IM} =\IM = \mult{\UMk}_{\kK}$ 
for some  $K \neq \es$ and some $(\UMk)_{\kK}$.

   Lemma~\ref{l:substitution} yields a
derivation $\Phi'_{t\isubs{x/u}} \rhd \Gam_t \inter  \Gam_u  \vdash
t\isubs{x/u}: \UM \mid \Del_t \union \Del_u$ with
$\sz{\Phi'_{t\isubs{x/u}}}=\sz{\Phi_t}+  \sz{\Theu}-|K|$ ($| \IM |  = |K|$). We set $\Phi'
= \Phi'_{t\isubs{x/u}}$ so that $\sz{\Phi}
= \sz{\Phi_t} +1+ \sz{\Theu}+1 > \sz{\Phi'}$. \\


\item
If $o = (\mu \al. \Com)u$, then $o' = \mu \al . \Com\ire{\al}{u}$
and $\al \in \fn{\Com}$.

The derivation  $\Phi$ has the following form: 
$$
{\small   \Phi = 
  \infer{\infer*{\tingD{\PhiCom}
                  {\tyj{\Com}
                       {\GamCom}
                       {\TypCom\mid\al: \VMC; \DelCom}}}
                  {\tyj{\mu\al.\Com}
                       {\GamCom}
                       {\VMC \mid \DelCom}} \\
                  \tingD{\Theu }
                         {\Gamu \Vdash u: \IMu  \mid \Delu}}
         {\tyj{(\mu\al.\Com)u}{\GamCom \inter \Gamu}
                              {\UM \mid \DelCom \union \Delu}}}
$$ where $ \VMC= \umult{\IMk \ftype \VMk}_{\kK}$, 
$\IMu= \inter_{\kK}   \choice{\IMk}$, 
$\UM= \vee_{\kK} \VMk$,
$\Gam = \GamCom \inter \Gamu$ and $\Del = \DelCom \union \Delu$. 
Lemma~\ref{l:replacement} then gives the derivation
$\tingD{\Phi_{\Com\ire{\al}{u}}}
                        {\tyj{\Com\ire{\al}{u}}
                             {\GamCom \inter  \Gamu}
                             {\TypCom\mid\al:\union_{\kK} \VMk ;\DelCom  \union  \Delu }}$.
Since $\al \in \fn{\Com}$ by hypothesis, then $K \neq \es$ by Lemma~\ref{l:relevance} so that we construct the following derivation:
    $$
    \tingD{\phi'}
          {\infer{\Phi_{\Com\ire{\al}{u}}
                        }
                 {\tyj{\mu\al.\Com\ire{\al}{u}}
                      {\GamCom \inter \Gamu }
                      {\union_{\kK} \VMk 
                             \mid \DelCom  \union \Delu  }}}
    $$
    We conclude since 
    \[ \begin{array}{l} 
      \sz{\Phi'} = \sz{\phi_{\Com\ire{\al}{u}}} + 1\\
      =_{L.~\ref{l:replacement}} 
      \sz{\Phi_\Com} + \sz{\Theu}  + 1                  \\
  <  \sz{\Phi_{\Com}} + 1  + \sz{\Theu}+ |K|\\
  = \sz{\Phi_{\mu \al. \Com}} + \sz{\Theu}+ |K|
  = \sz{\Phi}
    \end{array} \] 

The step $<$ is justified by $K \neq \es$. \\

\end{itemize}
\end{proof}

} 

It is now worth discussing the erasing cases.  Note that variable and
name assignments are not necessarily preserved by erasing
reductions. \modifref{For example, consider} $t=(\l y. x)z \Rew{} x=t'$.
The term $t$ is typed with a variable assignment whose domain is
$\cset{x,z}$, while $t'$ can only be typed with an assignment whose
domain is $\cset{x}$. Concretely, starting from a derivation of
$x:\mult{\umult{a}}, z:\mult{\umult{b}} \vdash (\l y. x)z: \umult{a}$
\modifrefb{(see Example~\ref{ex:derivation} on page~\pageref{ex:derivation})}, we can derive
$x:\mult{\umult{a}} \vdash x: \umult{a}$ but not $x:\mult{\umult{a}},
z:\mult{\umult{b}} \vdash x: \umult{a}$, so that the type is preserved
while the variable assignment is not.

  \newcommand{\trans}[3]{ 
   \begin{scope}[xshift=#1cm,yshift=#2cm] #3 \end{scope}}

  \newcommand{\bigderiv}[3]{
  \draw (#1,#2) --++ (1,0.7) --++ (0.3,1.4) --++ (-2.6,0) --++ (0.3,-1.4) -- cycle;
\trans{#1}{#2}{  
  \draw (0,0.6) node {#3};
}
  }

    \newcommand{\smallderiv}[3]{
  \draw (#1,#2) --++ (0.5,0.5) --++ (0.2,1) --++ (-1.4,0) --++ (0.2,-1) -- cycle;
  \trans{#1}{#2}{
  \draw (0,0.6) node {#3};
}
    }

\begin{figure}    
  \begin{tikzpicture}

\bigderiv{-0.65}{1.8}{\Large$\Pi_r$}
\draw (-0.65,1.8) node [below left] {$r:\UM$};
\draw (-0.65,1.8) --++ (0,-0.55);
\draw (-0.65,1) node {$\lx$};
\draw (-0.7,1) node [below left] {$\lx.r:\umult{\emul \rew\UM}$};
\draw (-0.65,1) circle (0.25);
\draw (-0.48,0.82) -- (-0.135,0.395);
\draw (0,0.2) node {$\arob$};
\draw (0,0.2) circle (0.25);
\draw (0,0.) node [below] {$\muju{\Gam_r+\Gam_s}{(\lx.r)s:\UM}{\Del_r+\Del_s}$};

\draw (0.175,0.355) -- (1.5,1.2);
\smallderiv{1.5}{1.2}{\large $\Pi_s$}
\draw (1.5,1.2) node [below right] {$s:\VM$};

\draw [->,>=stealth,ultra thick] (2.4,1.4) --++ (2.7,0);
\draw (1.5,3.6) node {$x\notin \fv{r}$};

\bigderiv{6.5}{0.7}{\Large$\Pi_r$}
\draw (6.5,0.7) node [below ] {$\muju{\Gam_r}{r:\UM}{\Del_r}$};

  \end{tikzpicture}
\caption{Proof reduction (erasing $\beta$-redex)}
  \label{fig:erasing-proof-reduction-beta}
\end{figure}

Type systems lacking (full) subject reduction are unusual, but (1) our
restricted form of subject reduction, for non-erasing steps only, is
sufficient for our purpose (see how we deal with the erasing steps in
the proof of Lemma~\ref{l:typable-isn}), (2) strong normalization
differs from head normalization in that, whereas ``$(\lx.r)s$ is HN''
is equivalent to ``$r\isubs{x/r}$ is HN'', ``$(\lx.r)s$ is SN'' not
equivalent to ``$r\isubs{x/s}$ is SN'' (except in the non-erasing
cases). In other words, contrary to the HN case, erasing reductions lose
information about the fact that terms are or are not SN. This loss of information (occurring in the
erasing case) is what makes subject reduction  in
system $\Slmu$   fail in general. This  is illustrated in
Fig.~\ref{fig:erasing-proof-reduction-beta}: there are two derivations
$\Pi_r \tri \muju{\Gam_r}{r:\UM}$ and $\Pi_s\tri
\muju{\Gam_s}{s:\VM}{\Del_s}$ with $x\notin \fv{r}$, which are subderivations of
$\Pi$ typing the redex $(\lx.r)s$, which concludes with $\muju{\Gam_r+\Gam_s}{(\lx.s)r}{\Del_r+\Del_s}$ (this is represented on the left-hand side of
Fig.~\ref{fig:erasing-proof-reduction-beta}). The type $\VM$ can be
seen as an instance of the choice operator $\choice{\emul}$. The
derivation typing the reduct $r$ is just $\Pi_r$ (on the right-hand
side). All typing information pertaining to $s$ has disappeared (\ie\ has
been lost), which explains why the variable/name assignments are
modified ($\Gam_r/\Del_r$ instead of $\Gam_r+\Gam_s/\Del_r+\Del_s$).
We will come back to the semantics of strong normalization in Section~\ref{ss:SN-vs-HN}.


\subsection{Backward Properties}
\label{s:backward}

Subject expansion is based on two technical properties: the first one, called
reverse substitution, allows us to extract type information for an object $o$ and a term $u$
from the type derivation of $o\isubs{x/u}$; similarly, the second one, called
reverse replacement, gives type information for a command $\Com$ and a term $u$
from the type derivation of $\Com\ire{\al}{u}$. Both of them are proved 
by induction on derivations using Lemmas~\ref{l:relevance} and~\ref{l:decomposition}. Formally,


\begin{lem}[{\bf Reverse Substitution}]
\label{l:reverse-substitution}
Let $\Phi'\tri \muju{\Gam'}{ o\subxu:\A}{\Del'}$
Then there exist $\Gam, \Del, \IM,  \Gamu, \Delu$
       such that:
\begin{itemize}
     \item $\Gam' = \Gam \inter \Gamu$, 
     \item $\Del' = \Del \union \Delu$, 
     \item $\tri \muju{\Gam;x:\IM}{o:\A}{\Del}$
     \item $\tri \muJu{\Gamu}{u: \IM}{\Delu}$.
   \end{itemize}
\end{lem}

\begin{proof}
By induction on $\Phi'$ using Lemmas~\ref{l:relevance} and~\ref{l:decomposition}.
See the Appendix for details.
\end{proof}

\begin{lem}[{\bf Reverse Replacement}]
\label{l:reverse-replacement}
Let $\Phi'\rhd \muju{\Gam'}{o\alu:\A}{\al:\VM;\Del'} $, where
$\al \notin \fn{u}$. 
Then there exist $\Gam, \Del, \Gamu,  \Delu, (\IMk)_{\kK}, (\VMk)_{\kK}$
       such that:
\begin{itemize}
     \item $\Gam' = \Gam \inter \Gamu$, 
     \item $\Del' =\Del \union \Delu$, 
     \item $\VM = \uVMk$,
     \item $\tri \muju{\Gam}{o:\A}{\al:\uIVMk;\Del}$, and
     \item $\tri \muJu{\Gamu}{u:\iIMsk}{\Delu}$
   \end{itemize}
\end{lem}

\begin{proof} By induction on $\Phi'$ using Lemmas~\ref{l:relevance} and~\ref{l:decomposition}. 
See the Appendix for details. 
\end{proof}

\noindent The following property will be used in Section~\ref{s:sn} to 
show that normalization implies typability.


\begin{pty}[{\bf Subject Expansion for $\Slmu$}]
\label{l:se-s}
Assume $\Phi' \rhd \muju{\Gam'}{o':\A}{\Del'}$. If $o\Rew{} o'$
is a non-erasing step, then there is $\Phi \rhd
\muju{\Gam'}{o:\A}{\Del'}$.
\end{pty}

\begin{proof} 
By induction on $\Rew{}$ using Lemma~\ref{l:relevance},
~\ref{l:reverse-substitution}
and~\ref{l:reverse-replacement}. See the Appendix for
  details. 
\end{proof}


\section{Strongly Normalizing $\lmu$-Objects}
\label{s:sn}

\newcommand{\nfof}[1]{\mathtt{nf}(#1)}

\subsection{Type-theoretic characterization, with quantitative bounds}
\label{ss:charac-sn}

In this section, we show the characterization of strongly-normalizing
terms of the $\l\mu$-calculus by means of the typing system introduced
in Section~\ref{s:types-lambda-mu}, \ie\ we show that an $\lmu$-object $o$ is
strongly-normalizing iff $t$ is $\Slmu$-typable. 

As defined in Section~\ref{s:calculus}, for any $\lmu$-object $o$, we
write $\mrl{o}$ for the length of the maximal reduction sequence
starting at $o$.  Remark that $\mrl{o}<\infty$ iff $o\in
\SN{\lmu}$. The following equations will play a key role in our proof
 of Lemma~\ref{l:typable-isn}. 

\begin{lem}
  \label{l:equalities}
\[ \begin{array}{llll}
  \mrl{x\,t_1 \ldots t_n} & = & +_{i=1 \ldots n} \mrl{t_i}\\
  \mrl{\l x.t } & = & \mrl{t}\\
  \mrl{\mu \al.\Com} & = &  \mrl{\Com}\\
  \mrl{\co{\al}t } & = & \mrl{t}\\
  \mrl{(\l x.t)u\, \vec{v}} & = & \mrl{t\isubs{x/u}\vec{v}} + 1 & \mbox{ if }  x \in \fv{t}\\
  \mrl{(\l x.t)u\, \vec{v}} & = & \mrl{u} + \mrl{t\vec{v}} + 1 & \mbox{ if }  x \notin \fv{t}\\
  \mrl{(\mu \al.\Com)u \vec{v} } & = & \mrl{(\mu \al. \Com\ire{\al}{u})\vec{v}} + 1 & \mbox{ if }  \al  \in \fv{\Com}\\
  \mrl{(\mu \al.\Com)u \vec{v} } & = &  \mrl{u} + \mrl{(\mu \al. \Com)\vec{v}} +
 1 & \mbox{ if }  \al  \notin \fv{\Com}\\
  \end{array} \]
\end{lem}

\begin{proof} The proof follows the same lines of that
  of the $\l$-calculus (see~\cite{RaamsdonkSSX99}, Fundamental Lemma of Maximality 3.18).  We only briefly discuss  the two last cases. 
  \begin{itemize}
  \item Let $o = (\mu \al.\Com)u\, \vec{v}$, where
    $\al \in \fv{\Com}$.  Take any reduction sequence $\rho$
    to normal
  form starting at $o$, it is straightforward to see that $\rho$ must
  reduce the head redex $ (\mu \al.\Com)u $, so that $\rho$ has
  necessarily the form $o \Rewn{\lmu} (\mu \al.\Com')u'\, \vec{v'}
  \Rew{\lmu} (\mu \al. \Com'\ire{\al}{u'})\vec{v'} \Rew{\lmu} \ldots$,
  where $\Com \Rewn{\lmu} \Com'$, $u \Rewn{\lmu} u$ and $\vec{v}
  \Rewn{\lmu} \vec{v'}$. Then $\rho$ can be transformed into another
  (potentially longer) reduction
    sequence $\rho'$ of the form $o \Rew{\lmu} (\mu
    \al. \Com\ire{\al}{u})\vec{v} \Rewn{\lmu}
    \mu  \al. \Com'\ire{\al}{u'})\vec{v'}\Rew{\lmu} \ldots$.
    Thus  $\mrl{(\mu \al.\Com)u\, \vec{v} }  \leq  \mrl{(\mu \al. \Com\ire{\al}{u})\vec{v}} + 1$  holds.  The converse inequality is easy.
  \item Let $o = (\mu \al.\Com)u\, \vec{v}$, where $\al \notin \fv{\Com}$.
    Again, take any reduction sequence $\rho$ to normal form starting at $o$,
    it is again straightforward to see that $\rho$ must reduce the head
    redex $ (\mu \al.\Com)u $, so that $\rho$ has necessarily the form
    $o \Rewn{\lmu} (\mu \al.\Com')u'\, \vec{v'} \Rew{\lmu} (\mu
    \al. \Com')\vec{v'} \Rew{\lmu} \ldots$, where $\Com
    \Rewn{\lmu} \Com'$, $u \Rewn{\lmu} u$ and $\vec{v} \Rewn{\lmu}
    \vec{v'}$.  Then $\rho$ can be transformed into another (potentially longer) reduction
    sequence $\rho'$ of the form $o \Rewn{\lmu} (\mu
    \al. \Com) \nfof{u}\, \vec{v} \Rew{\lmu}  (\mu
    \al. \Com)  \vec{v} \Rewn{\lmu}
   (\mu  \al. \Com')\vec{v'}\Rew{\lmu} \ldots$,
    where $\nfof{u}$ denotes the normal form of $u$.
    Thus,  $\mrl{(\mu \al.\Com)u\, \vec{v} } \leq \mrl{u}+ \mrl{(\mu \al. \Com)\vec{v}} + 1$. The converse inequality is easy.
  \end{itemize}
 \end{proof}

\noindent The proof of our main result (Theorem~\ref{t:final}) relies on the following two ingredients:
\begin{itemize}  
\item Every $\Slmu$-typable object is in $\SN{\lmu}$ (Lemma~\ref{l:typable-isn}).  
\item Every object in $\SN{\lmu}$ is $\Slmu$-typable (Lemma~\ref{l:isn-typable}).  
\end{itemize}

We first show that any typable object $o$ belongs to $\SN{\lmu}$. 

\begin{lem} 
\label{l:typable-isn}
If  $o$ is $\Slmu$-typable, \ie\ $\Phi \tri\Gam \vdash o:\Any \mid \Del$, then  $\mrl{o} \leq \sz{\Phi}$. Thus $o  \in  \SN{\lmu}$. 
\end{lem}

\begin{proof} We show $\mrl{o} \leq \sz{\Phi}$
  by induction on $\sz{\Phi}$, where $\Phi \tri\Gam \vdash o:\Any \mid
  \Del$.  When $\Phi$ does not end with the rule $(\app)$ the proof
  holds straightforwardly by the \ih\ so we consider that $\Phi$ ends
  with $(\app)$, where $\Any = \UM$ and $o = x\,t_1 \ldots t_n$ or $o =
  (\mu \al. \Com)t_1 \ldots t_n$ or $o = (\l x.  u)t_1 \ldots t_n$,
  with  $n \geq 1$. \\

In all the three cases for $o$, there are subderivations
$(\Phi_{i})_{i \in \cset{1 \ldots n}}$ such that $\sz{\Phi_{i}} <
\sz{\Phi}$ so that the \ih\ gives 
$\mrl{t_i} \leq \sz{\Phi_{i}}$. Now, there are three different cases
to consider: \\

\begin{itemize}
\item[(1)]
 If $o=x\,t_1 \ldots t_n$, then there are non-empty subderivations $\Phi_1\ldots \Phi_n$ of $\Phi$ typing $t_1\ldots t_n$ respectively. Since $+_{i=1\ldots n} \sz{\Phi_i} + n+1 \leq \sz{\Phi}$, the \ih\ gives $\mrl{t_i}\leq \sz{\Phi_i}$ for $1\leq i \leq n$. We conclude since $ \mrl{x\,t_1 \ldots t_n}  =  +_{i=1 \ldots n} \mrl{t_i}$.    \\
  
\item[(2)] If $o = (\mu \al. \Com)t_1 \ldots t_n$, there are two cases: 
\newcommand{\mujudots}[1]{\muju{\ldots}{#1}{\ldots}}
  \begin{itemize}
 \item $\al \in \fn{\Com}$.  Using  Property~\ref{l:sr}, we get  $\Phi' \tri
   \tyj{(\mu \al. \Com\ire{\al}{t_1}) t_2 \ldots t_n }{\Gam}{\UM \mid
     \Del}$ and $\sz{\Phi'} < \sz{\Phi}$. Then the \ih\ gives $\mrl{(\mu \al. \Com\ire{\al}{t_1}) t_2 \ldots t_n } \leq \sz{\Phi'}$. We conclude since
   $\mrl{o} = \mrl{(\mu \al. \Com\ire{\al}{t_1}) t_2 \ldots t_n } + 1
   \leq   \sz{\Phi'} +1 \leq  \sz{\Phi}$.
   
 \item $\al  \notin \fn{\Com}$. 
Then $\Phi$ is of the form:
   $$
\infer{\infer{
  \infer{\Phi_{\mu\al.\Com}\tri \mujudots{\mu\al.\Com:\umult{\Blind_0}}\sep \infer{\Phi_{1}\tri \mujudots{t_1:\VM_1}}{ \mujudots{t_1:\mult{\VM_1}}}}{ \mujudots{(\mu \al.\Com)t_1:\umult{\Blind_1}}} }{\vdots}\sep \infer{\Phi_{n}\tri \mujudots{t_n:\VM_n}}{ \mujudots{t_n:\mult{\VM_n}}}}{ \mujudots{o:\umult{\Blind_n}}}
$$
with $\UM=\umult{\Blind_n}$.
The type $\umult{\Blind_0}$ is obtained by choice because Lemma~\ref{l:relevance} guarantees that no $\alpha$ is typed in the name assignments of the left derivations. Note that $\sz{\Phi}=+_{i=1\ldots n}\sz{\Phi_i}+\sz{\Phi_{\mu \al.c}}+n$ since the $\Blind_i$ are all blind.
We then build the following derivation:
 $$\Phi_{o'}=\infer{
   \infer{\Phi'_{\mu\al.\Com}\tri \mujudots{\mu\al.\Com:\umult{\Blind_1}}\sep \infer{\Phi_{2}\tri \mujudots{t_2:\VM_2}}{ \mujudots{t_2:\mult{\VM_2}}} }{\vdots}\sep \infer{\Phi_{n}\tri \mujudots{t_n:\VM_n}}{ \mujudots{t_n:\mult{\VM_n}} } }
 { \mujudots{o':\umult{\Blind_n}}}
 $$
 with $o'=(\mu\al.\Com)t_2\ldots t_n$ and $\Phi'_{\mu\al.\Com}$ the
 derivation $\Phi_{\mu\al.\Com}$ where the choice operator
 raises $\Blind_1$ instead of $\Blind_0$ (actually, any blind type of arity $\geq n-1$ would do). In particular,
 $\sz{\Phi'_{\mu\al.\Com}}=\sz{\Phi_{\mu\al.\Com}}$, so that
 $\sz{\Phi_{o'}}=+_{i=2\ldots
   n}\sz{\Phi_i}+\sz{\Phi_{\mu\al.\Com}}+n-1< \sz{\Phi}$.
 We also have $\sz{\Phi_1} < \sz{\Phi}$. The \ih\ then gives $\mrl{o'} \leq \sz{\Phi_{o'}}$ and  $\mrl{t_1} \leq \sz{\Phi_{1}}$.  We conclude since:
 $$
\begin{array}[t]{c@{\;}l@{\;}l}
  \mrl{(\mu\al.\Com)t_1\ldots t_n} & = & \mrl{(\mu\al.\Com)t_2\ldots t_n}+\mrl{t_1}+1 \\
  & \leq_{\ih} & \sz{\Phi_{o'}}+\sz{\Phi_{1}}+1\\
  & = &  +_{i=2\ldots n}\sz{\Phi_i}+\sz{\Phi_{\mu\al.\Com}}+n-1 +\sz{\Phi_{1}}+1 \\
  & = & \sz{\Phi}  
\end{array}$$

 \end{itemize}
\item[(3)] If $o=(\l x. u)t_1 \ldots t_n$, we reason similarly to the
  previous case. \qedhere
\end{itemize}
\end{proof}

\begin{lem} 
  \label{l:isn-typable}
  If $o\in \SN{\lmu}$, then $o$ is $\Slmu$-typable.
\end{lem}

\begin{proof}
  By induction on $\lexnl{o}$. If $o$ is not an application $t\,u$, the inductive step is straightforward. We only detail the case when $o$ is an application.
  \begin{itemize}
\item If $o=x\,t_1\ldots t_n$, the property is straightforward.
  \item If $o = (\mu \al. \Com)t_1 \ldots t_n$, we set $o'=(\mu\al.\Com\ire{\al}{t_1})t_2\ldots t_n$. There are two cases: 
\begin{itemize}
 \item $\al \in \fn{\Com}$. Then $\mrl{o}=\mrl{o'}+1$.  By the \ih,  there is $\Phi_{o'}\tri \muju{\Gam'}{o':\UM}{\Del'}$. By Property~\ref{l:se-s}, there is $\Phi\tri \muju{\Gam'}{o:\UM}{\Del'}$ and we are done.
 \item  $\al \notin \fn{\Com}$. Then $o'=(\mu\al.\Com)t_2\ldots t_n $ and $\mrl{o}=\mrl{o'}+\mrl{t_1}+1$. By the \ih, there are typing derivations  $\Phi_{o'}\tri \muju{\Gam'}{o':\UM}{\Del'}$ and $\Phi_1\tri \muju{\Gam_1}{t_1:\VM_1}{\Del_1}$. Since $\al \notin \fn{\Com}$, Lemma~\ref{l:relevance} entails that $\Phi_{o'}$ is of the following form, where $\Blind_1$ is a blind type:
    $$
\infer{
  \infer{\Phi'_{\mu\al.\Com}\tri \mu\al.\Com:\umult{\Blind_1}\sep \infer{\Phi_2\tri t_2:\VM_2}{t_2:\mult{\VM_2}}
  }{\vdots}\sep \infer{\Phi_{n}\tri t_n:\VM_n}{ t_n:\mult{\VM_n}}}{\muju{\Gam'}{o:\umult{\Blind_n}}{\Del'}}
$$
\end{itemize}
We then set $\Phi$ as follows:
$$
\infer{
  \infer{\infer{\Phi_{\mu\al.\Com}\tri \mu\al.\Com:\umult{\Blind_0}\sep
\infer{\Phi_1\tri \muju{\Gam_1}{t_1:\VM_1}{\Del_1}}{t_1:\mult{\VM_1}}
    }{
(\mu\al.\Com)t_1:\umult{\Blind_1}
}      \sep \infer{\Phi_2\tri t_2:\VM_2}{t_2:\mult{\VM_2}}
  }{\vdots}\sep \infer{\Phi_{n}\tri t_n:\VM_n}{ t_n:\mult{\VM_n}}}{\muju{\Gam'+\Gam_1}{o:\umult{\Blind_n}}{\Del'+\Del_1}}
$$
where $\Phi_{\mu\al.\Com}$ is exactly as $\Phi'_{\mu\al.\Com}$ except
that we raise the blind type $\Blind_0=\emul \rew \umult{\Blind_1}$
instead of $\Blind_1$ (actually, any blind type of arity
$\geq n$ would do).

\item If $o=(\lx.t)u\,\vec{v}$, we reason similarly to the previous case. \qedhere

  \end{itemize}
\end{proof}

Lemmas~\ref{l:typable-isn} and~\ref{l:isn-typable}  allow us to conclude 
with the main result of this paper which is  the equivalence between typability and 
strong-normalization for the $\lmu$-calculus.  Notice that no 
reducibility argument was used in the  whole proof. 

\begin{thm}
  \label{t:final}
Let $o \in \objects{\lmu}$. Then $o$ is typable in system $\Slmu$ iff
$o \in \SN{\lmu}$.  Moreover, if $o$ is $\Slmu$-typable with tree
derivation $\Pi$, then $\sz{\Pi}$  gives an upper bound to the maximal
length of a reduction sequence starting at $o$. 
\end{thm}

\ignore{
\delia{CETTE PARTIE s'en va!!}
\modifref{ To prove the second part of the statement, \ie\ that
  $\sz{\Pi}$ bounds the length of every reduction sequence starting at
  $o$, it is sufficient to endow the system with \emph{non-relevant}
  axioms for variables and names, \ie\ to replace the rule $(\ax)$
  with the following weakened form:
$$
\infer{\UM \neq \eumul}{\muju{\Gam \inter x:\mult{\UM}}{x:\UM}{\Del}}\ (\mathtt{axw})  
$$ This extension of system $\Slmu$, called $\Slmu'$,  which does not satisfy subject expansion, is however sufficient to guarantee
\textit{weighted} subject reduction for erasing as well as for
non-erasing steps. The definition of $\sz{\_}$ in this new extended system
counts an $(\axw)$-rule for 1.  Thus, if $o$ is $\Slmu$-typable, let say by means of a typing derivation $\Pi$, then $o$ is also typable by means of the same derivation $\Pi$, with the same size $\sz{\Pi}$,  in the
extended system $\Slmu'$, and this is
  because rule $(\mathtt{axw})$ subsumes $(\ax)$.  By \textit{weighted} subject reduction in system
  $\Slmu'$, $\sz{\Pi}$ gives the expected upper bound.  
}
}

\modifrefb{
  \subsection{Discussion}
\label{ss:SN-vs-HN}
As we have observed in the end of Section~\ref{s:forward}, subject
reduction  for erasing steps fails in system $\Slmu$. The same is true
for subject expansion. This naturally rise the question: why subject
reduction and expansion should hold in system $\Hlmu$ (related to HN)
and not in $\Slmu$ (related to SN)? The difference of treatment lies
in the semantics of what these two systems are designed to capture:
  \begin{itemize}
  \item $\Hlmu$ captures head normalization (Theorem~\ref{t:hn-lmu}) \ie\ a $\lmu$-object $o$ reduces to a HNF iff $o$ is
      $\Hlmu$-typable.   The
      first crucial observation about these first results is that no
      "semantic information" with respect to head normalization is
      lost during (arbitrary) reduction. We make this notion more
      concrete below.  Indeed, if $o \rew o'$ (whether this step is
      erasing or not), then, by confluence, $o$ is HN iff $o'$ is HN,
      and this is why we want both subject reduction and subject
      expansion for $\Hlmu$ to hold. Moreover, $o$ and $o'$ have head
      normal forms of the same shape. Intuitively, it is no more
      difficult to prove that $o'$ is HN than to prove that $o$ is HN, 
      and vice-versa. Of course, the step $o\rew o'$ may erase some
      reduction paths that exist in $o$, but these (erased) reduction paths are irrelevant
      to the fact that $o$ is HN or not: no semantic information
      pertaining to HN has been lost. This situation is very different
    in system $\Slmu$. 
    
\item $\Slmu$ captures strong normalization (Theorem~\ref{t:final}), which is different from weak and head normalization because it is
  a property about the finiteness of \emph{all} reduction paths, and not about the existence of \emph{at least one} reduction path
  to a normal form. This difference is materialized by
  the following three  observations.
    \begin{itemize}
    \item If $o\rew o'$, then ``$o$ is SN'' and ``$o'$ is SN'' are not equivalent propositions, \eg $(\lx.y)\Om\rew y$, $(\lx.y)\Om$ is not SN whereas $y$ is. In particular, no type system characterizing SN satisfies subject expansion.
    
    \item If $o \rew o'$ and ``$o$ is SN'', then ``$o'$ is always
      SN'', but some important semantic information may be lost with
      respect to strong normalization!  For instance, if $u$ is SN but
      the normal form of $u$ cannot be reached in less than 1000
      reductions steps, then $(\lx.y)u$ is also SN, but the (erasing)
      reduction step $(\lx.y)u\rew y$ loses information regarding the
      reduction paths starting at $(\lx.y)u$: the reduction
      obliterates reduction paths in $u$ (intuitively, on may not know
      if $y$ originates from $(\lx.y)u$, or $(\lx.y)x$, or
      $(\lx.y)\Om$\ldots).  This is a loss, which explains why full
      subject reduction does not hold for system $\Slmu$.  Moreover,
      this also suggests that full subject reduction would arguably be
      less \emph{faithful} to the semantics of strong normalization.
      Note that (full) subject reduction could be obtained by just
      allowing weakening in the typing system, while
      preserving the characterization theorem. But
      weakening does not restore  subject expansion, since SN is not
      stable under expansion.
  \item If $o\rew o'$ and ``$o'$ is SN implies $o$ is SN'',
    then the reduction step is called \emph{perpetual}: perpetual
    strategies  are precisely
    those that are used to study strong normalization.
      \end{itemize} Again, it is interesting to note that, when $o \rew o'$
    is \emph{non-erasing}, then $o$ is SN iff $o'$ is SN, and
    any reduction path in $o$ has residuals in $o'$.  This explains why a typing system for strong normalization should satisfy
      subject reduction and subject expansion for
      \emph{non-erasing} steps (which is sufficient to prove the characterization theorem),  while this is not necessary for erasing steps.
      \end{itemize}
  These  observations summarize why our typing systems enjoy
  full subject reduction and subject expansion in one case ($\Hlmu$) and only
  subject reduction and subject expansion for non-erasing steps in
  the other ($\Slmu$). 
}

 
\section{The $\lmuex$-calculus}
\label{s:lmuex}

This section introduces the syntax (Section~\ref{s:syntax-ex}) and the
operational semantics (Section~\ref{s:operational-ex}) of the
$\lmuex$-calculus, a \modifref{small-step refinement of} $\lmu$, \modifrefb{ for which
      the typing system $\Slmu$ naturally extends}.  The restriction of the
$\lmuex$-calculus to intuitionistic logic is known as the
\textit{linear substitution calculus}~\cite{ABKL14}, \modifref{deeply
  studied in} rewriting theory and  complexity analysis.

\subsection{Syntax}
\label{s:syntax-ex}

The set of 
\deft{objects} ($\objects{\lmuex}$), \deft{terms} ($\terms{\lmuex}$) and 
\deft{commands} ($\commands{\lmuex}$) of the $\lmuex$-calculus
are  given by the following grammars
\[ \begin{array}{llll}
   (\textbf{objects})       & o & ::= & t \mid \Com \\
   (\textbf{terms})       & t,u & ::= & x \mid \l x. t \mid t u  \mid \mu \al. \Com \mid t[x/u]  \\
   (\textbf{commands})      & \Com & ::= & \co{\al} t \mid \Com \rempl{\al}{\beta}{u}  \\
   \end{array} \]
The construction $[x/u]$ (resp. $\rempl{\al}{\beta}{u}$) is called an
\deft{explicit substitution (ES)} (resp. \deft{explicit replacement
  (ER)}). Remark that ES do not apply to commands and ER do not apply
to terms.  An ES $[x/u]$ implements the \textit{meta-substitution} operator 
$\isubs{x/u}$ while an ER $\rempl{\al}{\beta}{u}$ implements the \textit{fresh}
 replacement meta-operator $\ire{\al}{\beta.u}$ introduced in Section~\ref{s:operational}, \ie\ the  small step computation of $\Com \rempl{\al}{\beta}{u}$
replaces  only one occurrence of 
$\co{\al} t$ inside $\Com$ by $\co{\beta} t\rempl{\al}{\beta}{u}u$.
As in Section~\ref{s:syntax}, the \deft{size of an object} $o$
is denoted by $|o|$.  

The notions of \deft{free} and \deft{bound  variables} and \deft{names} are extended as expected, in
particular $\fv{t[x/u]} := (\fv{t} \sm \cset{x}) \cup \fv{u}$
and $\fn{\Com \rempl{\al}{\beta}{u}} := (\fn{\Com} \sm \cset{\al}) \cup \cset{\beta} \cup \fn{u}$.
The derived  notion of $\alpha$-conversion (\ie\
renaming of bound variables and names) will be assumed in the rest of the paper.
 Thus \eg\ $(\co{\gamma} x[x/y])\rempl{\gamma}{\beta}{z} =_{\al} 
(\co{\gamma'} x'[x'/y])\rempl{\gamma'}{\beta}{z}$. 
The \deft{number of free occurrences} of the variable $x$ (resp. the name $\al$) in 
$o$ is denoted by $|o|_x$ (resp. $|o|_\al$). 

\deft{List} ($\slist$), \deft{term} ($\cxtt,\cxct,\cxot$), and \deft{command} ($\cxtc,\cxcc,\cxoc$) \deft{contexts} are respectively defined by the following grammars:
\[ \begin{array}{llll}
 \slist     ::=  \Box \mid \slist [x/u] \\
 \cxtt     ::=  \Box \mid  \l x. \cxtt \mid \cxtt\,  t  \mid t\, \cxtt 
                                            \mid \mu \al. \cxct \mid \cxtt [x/t] \mid  t[x/\cxtt] \\
 \cxct     ::=    \co{\al} \cxtt  \mid \cxct \rempl{\al}{\beta}{u}  \mid 
                                            \Com \rempl{\al}{\beta}{\cxtt} \\
 \cxot    ::= \cxtt \mid \cxct  \\
 \cxtc     ::=  \l x. \cxtc \mid \cxtc\,  t  \mid t\, \cxtc 
                                            \mid \mu \al. \cxcc \mid \cxtc [x/t] \mid  t[x/\cxtc] \\
 \cxcc     ::=  \boxdot \mid \co{\al} \cxtc  \mid \cxcc \rempl{\al}{\beta}{u}  \mid 
                                            \Com \rempl{\al}{\beta}{\cxtc} \\
 \cxoc   ::= \cxtc \mid \cxcc  \\
\end{array} \]
 The hole $\Box$ (resp. $\boxdot$) can be replaced by a term (resp. a
 command)\modifref{. Indeed}, $\slist[t]$ denotes the replacement of $\Box$ in
 $\slist$ by the term $t$ (similarly for $\cxtt[t]$, $\cxct[t]$ and
 $\cxot$), while $\cxcc[\Com]$ denotes the replacement of $\boxdot$ in
 $\cxcc$ by the command $\Com$ (similarly for $\cxtc$ and $\cxoc$).
   Every meta-expression
   $\mathtt{X}\mathtt{Y}$  with $\mathtt{X} \in
   \cset{\mathtt{T},\mathtt{C}}$ and
   $\mathtt{Y} \in
   \cset{\mathtt{T},\mathtt{C}}$ 
   must be interpreted as a context \modifref{taking an object $\mathtt{Y}$}
   and yielding an object $\mathtt{X}$: \eg\ $ \cxtc$
   denotes a context that takes a command ($\mathtt{C}$ on the right) and outputs a term ($\mathtt{T}$ on the left).
  
We write $\cxot^{\sset}$ for a term context $\cxot$
which does not capture the free variables and names in the set $\sset$, \ie\ there are no abstractions and substitutions in the context
that bind the symbols  in $\sset$. For instance $\cxtt  = \l
y. \Box$ can be specified as $\cxtt^x$ while $\cxtt  = \l x. \Box$
cannot. In order to emphasize this particular property we may write
$\cxtt^{\sset}\cwc{t}$ instead of $\cxtt^{\sset}[t]$,
and we may omit $\sset$ when it is clear from the context. Same concepts
apply to command contexts, \ie\ $\cxoc^{\sset}$ does not capture the variables and names in $\sset$ and the notation used for that is $\cxoc^{\sset}\cwc{\Com}$. 

\subsection{Operational Semantics}
\label{s:operational-ex}

The reduction rules of the $\lmuex$-calculus aim to give a
 \modifref{small-step semantics} to the $\lmu$-calculus, based on the {\it
  substitution/replacement at a distance} paradigm~\cite{AK10,ABKL14}.
 The reduction relation $\lmuex$ of the calculus is given by the
context closure  of the following rewriting rules.
\begin{center}
  $ \begin{array}{l*{4}{@{\hspace{.2cm}}l}}
   \slist[\l x. t]\, u
    & \rrule{\B} & \slist[t [x/u]]   \\
   \cxtt\cwc{x}[x/u] &  \rrule{\cntrs} &   \cxtt\cwc{u}[x/u] & \mbox{if } |\cxtt\cwc{x}|_x >1  \\
   \cxtt\cwc{x}[x/u] &  \rrule{\ders}&   \cxtt\cwc{u} & \mbox{if } |\cxtt\cwc{x}|_x =1 \\
   t[x/u]  &  \rrule{\Gcs}&   t & \mbox{if } x \notin \fv{t} \\ 
   \slist[\mu \al. \Com ]\, u 
     &   \rrule{\Mu}&  \slist[\mu \gamma. \Com\rempl{\al}{ \gamma}{u}] &  \mbox{if } \gamma \mbox{ is fresh } \\
\cxcc\cwc{\co{\al} t} \rempl{\al}{ \gamma}{u} &  \rrule{\cntrr} &   \cxcc\cwc{\co{\gamma}t u} \rempl{\al}{ \gamma}{u} & \modifrefb{\mbox{if } |
\cxcc\cwc{\co{\al}t]}|_\al > 1} \\
\cxcc\cwc{\co{\al}t} \rempl{\al}{ \gamma}{u} &  \rrule{\derr} &   \cxcc\cwc{\co{\gamma}t u} & \mbox{if } |\cxcc\cwc{\co{\al}t}|_\al =1  \\
\Com  \rempl{\al}{ \gamma}{u} &  \rrule{\Gcr} & \Com &  \mbox{if } \al  \notin \fn{\Com} \\
  \end{array}$
  \end{center}
where $\cxtt$ is to be understood as $\cxtt^x$ and
$\cxcc$ as $\cxcc^{\al,\gamma}$.

We use  $\Rew{\Gc}$ for the reduction relation generated by 
the set of rules $\{ \rrule{\Gcs}, \rrule{\Gcr}\}$ 
and  $\Rew{\nonelmuex}$ for  the \textbf{non-erasing reduction relation} $\Rew{\lmuex} \sm \Rew{\Gc}$. For instance, the big step reduction 
\begin{center}
  $(\mu
\al.\co{\al}x(\mu\beta.\co{\al}\l x.xx))u\rew_\mu \mu
\gamma.\co{\gamma}x(\mu\beta.\co{\gamma}(\l x. xx)u))u$
\end{center}
where $\al$ has
been alpha-renamed to  $\gamma$,  can be now emulated by 3 small steps :
$$\begin{array}{l}
(\mu \al.\co{\al}x(\mu\beta.\co{\al}\l x.xx))u \\
\rew_{\Mu} \mu \gamma.(\co{\al}x(\mu\beta.\co{\al}\l x.xx))\rempl{\al}{\gamma}{u} \\
\rew_{\cntrr} \mu \gamma.(\co{\al}x(\mu\beta.\co{\gamma}(\l x.xx)u))\rempl{\al}{\gamma}{u} \\
\rew_{\derr} \mu \gamma.\co{\gamma}x(\mu\beta.\co{\gamma}(\l x. xx)u))u
\end{array}$$
Notice that the occurrences of $\al$ are (arbitrarily) replaced by
$\gamma$ one  after another, thus  replacement is \textit{linearly} processed. When
there is just one occurrence of $\al$ left, the small reduction step
$\derr$ performs the last replacement  and erases the remaining ER $\rempl{\al}{\gamma}{u}$ to complete the operation.

More generally, not only the syntax of the  $\lmuex$-calculus can be seen as a refinement of the $\lmu$-calculus,
but also its operational semantics. Formally, \medskip

\begin{lem}
If $o \in \objects{\lmu}$, then $o \Rew{\lmu} o'$ implies $o \Rewplus{\lmuex} o'$.  
\end{lem}

\begin{proof} By induction on the reduction relation $\Rew{\lmu}$. \end{proof}

Moreover, we can project $\lmuex$-reduction sequences into $\lmu$-reduction sequences.
Indeed, consider the projection function $\proj{\_}$ computing all the explicit
substitutions and replacements of an object, thus  in particular
$\proj{t[x/u]}:= \proj{t}\isubs{x/\proj{u}}$ and 
$\proj{\Com  \rempl{\al}{ \al'}{u}} := \proj{\Com}  \ire{\al}{ \al'.\proj{u}}$. 
Then,  \medskip 

\begin{lem}
If $o \in \objects{\lmuex}$, then $o \Rew{\lmuex} o'$ implies $\proj{o} \Rewn{\lmu} \proj{o'}$.  
\end{lem}

\begin{proof} By induction on the reduction relation $\Rew{\lmuex}$. \end{proof}

\subsection{Typing System}

In this section we extend the (quantitative) typing system $\Slmu$ in order to 
capture the $\lmuex$-calculus, the aim being to characterize the set
of strongly $\lmuex$-normalizing objects by using quantitative arguments. 

More precisely, system $\Slmu$ is  enriched with the two typing rules in Figure~\ref{fig:strong-lambdamuex}.
\begin{figure}[h]
\begin{framed}
\begin{center}
${
\begin{array}{c}
 \infer[ (\subs)]{\Gamt; x:  \IM \vdash t: \UM  \mid \Delt \sep 
       \Gamu  \Vdash u :  \choice{\IM}   \mid \Delu } 
       {\Gamt \inter \Gamu \vdash t[x/u]: \UM \mid \Delt 
         \union \Delu}  \\ \\
{ \infer[ (\repl)]{\GamCom \vdash \Com: \TypCom \mid \DelCom; \al: \umult{\IMk \rew \VMk}_{\kK }\!\!\!\!\sep 
       \Gamu  \Vdash u :  \choice{({\inter_{\kK} \choice{\IMk}})}   \mid \Delu } 
       {\GamCom \inter \Gamu \vdash  \Com\rempl{\al}{\al'}{u}: \TypCom  \mid \DelCom 
       \union \Delu \vee  \al': \union_{\kK} \VMk}}  \\
\end{array}}$
\end{center}
\caption{Additional Rules for System $\Slmuex$}
\label{fig:strong-lambdamuex}
\end{framed}
\end{figure}
Rule $(\subs)$ is inspired by the derivation tree typing the term $(\l
x. t)u$: indeed, any derivation $\tri \Gam \vdash (\l x.t)u:\VM \mid
\Del$ induces two derivations $\tri \Gamt, x:\IM \vdash t: \VM\mid
\Delt$ and $\tri \Gamu \Vdash u:\choice{\IM} \mid \Delu$, from which
we can type $t[x/u]$.  Likewise, the rule $(\repl)$ is motivated by
the derivation tree typing a $\mu$-redex.  In particular, when $K=\es$
(\ie\ when $\al \notin\fn{\Com}$), then $\choice{({\inter_{k\in \es}
    \choice{\IMk}})} = \choice{\emul}$, so that the outer star in
$\choice{({\inter_{\kK} \choice{\IMk}})} $ gives an arbitrary multiset
$\mult{\sig}$ ensuring the typing (and thus the $SN$ property) of the
replacement argument $u$. Notice that Lemma~\ref{l:non-empty-union}
still holds for $\Slmuex$ \ie\ if $\tri_{\Slmuex} \muju{\Gam}{t:\UM}{\Del}$, then $\UM\neq \eumul$.

As one may expect, system $\Slmuex$ encodes 
a non-idempotent and relevant system for intuitionistic logic 
with ES~\cite{KV14}. More precisely, restricting rule $(\subs)$ 
to $\l$-terms with ES gives the following rule:
$$
\infer{\Gam; x:\IM \vdash t:  \sig \sep 
       \Gam' \Vdash u: \choice{\IM} }
      {\Gam \inter \Gam' \vdash t[x/u]: \sig}$$  

Relevance also holds for $\lmuex$: \medskip
\begin{lem}[\textbf{Relevance}]
\label{l:relevance-bis}
Let $o \in \objects{\lmuex}$. If $\Phi\tri \Gam \vdash o: \Any \mid \Del$
(resp. $\Phi\tri \Gam \Vdash t: \IM \mid \Del$ with $\IM\neq \emul$),
then $\fv{o}=\dom{\Gam}$ and $\fn{o} = \dom{\Del}$ 
(resp. $\fv{t}=\dom{\Gam}$ and $\fn{t} = \dom{\Del}$).
\end{lem}
\begin{proof} By induction on $\Phi$. \end{proof}

We now extend the function $\sz{\_}$ introduced in Section~\ref{s:strong-lambdamu} by adding the 
following cases: 
\begin{center}
$ 
{ \begin{array}{lll}
 \sz{\infer[ (\subs)]{\Phit \rhd   t \sep  \Phiu \rhd u }
      {\Gamt \inter  \Gamu   \vdash 
       t[x/u]:\UM  \mid \Del \union \Delu} }  & := &  \sz{\Phi_t} +  \sz{\Phi_u} \\ \\
 \sz{\infer[ (\repl)]{\Phi_\Com \rhd  \modifrefb{\muju{\Gam}{\Com:\TypCom}{\Del,\al:\umult{\IM_k \rew \UM}_{\kK}}} \sep  \Phi_u \rhd u }
      {\GamCom \inter  \Gamu   \vdash 
        \Com \rempl{\al}{\al'}{u}:\TypCom  \mid \DelCom \union \Delu} } & := &  \sz{\Phi_\Com} +  \sz{\Phi_u} +|K| - \frac{1}{2}\\
\end{array} } $
\end{center}
Notice that $\sz{\Phi} \geqslant 1$ still holds for any \textit{regular} derivation $\Phi$. \\


As explained in Section~\ref{s:forward},
  weighted subject reduction 
  holds for $\mu$-reduction steps like $t=(\mu \al.\Com)u \Rew{\mu} \mu
  \gamma.\Com \ire{\al}{\gamma.u}=t'$ 
   because $\gamma$ is typed in $t'$ with smaller
  arity than that of  $\al$ in $t$. The (big) step above is emulated
  in the  $\lmuex$-calculus by 
  the (small) steps $t \Rew{\Mu} \mu
  \gamma.\Com\rempl{\al}{\gamma}{u} \Rewplus{\cntrr, \derr, \Gcr} t'$, where  $\cntrr$ and $\derr$ perform 
  linear replacements, so they are also naturally expected to decrease
  the size of type derivations.  However, 
  for the first step $t=(\mu \al.\Com)u\Rew{\Mu} \mu
  \gamma.\Com\rempl{\al}{\gamma}{u} =t'$, even if no real replacement has taken place
  yet,  we should still have a quantifiable  decrease of the form  $\sz{\Phit} >\sz{\Phi_{t'}}$. 
  This is the reason we use  "$-\frac{1}{2}$" when defining the size of explicit replacements,
  which does not compromise the forthcoming weighted subject reduction property. 

  
One may naively think that the "$-\frac{1}{2}$" component in the size
definition of an ER can compromise the decrease of the size for a step
$t=\mu \gam.\Com\rempl{\al}{\gam}{u} \rew_{\derr} \mu \al.\Com
\ire{\al}{\gam.u}=t'$, when $\Com$ holds exactly one occurrence of
$\al$ : indeed, removing the ER $\rempl{\al}{\gam}{u}$ induces an \textit{increase} of
the measure equal to $\frac{1}{2}$. 
However, the arity contribution of (the unique occurrence of) $\al$ in $t$ is greater than that of the new occurrence of $\gam$ in $t'$: the replacement
operation then induces a decrease of the measure which is equal to some $k\geqslant 1$; and  thus the overall decrease of the measure is in the worst case
$k-\frac{1}{2}>0$, which still grants $\sz{\Phi} > \sz{\Phi'}$.  
The decrease of the measure for a  $\Gcr$-step is 
  more evident.  Last, but not least, the fact that
$\sz{\Phi}$ is a half-integer greater or equal to one ensures that the measure is still well-founded.


\section{Typing Properties}
\label{s:properties-ex}

As in the case of the $\lmu$-calculus, we show that the refined
$\lmuex$-calculus is well-behaved w.r.t. the extended typing system
$\Slmuex$. This is done by   means of forward (Section~\ref{s:forward-ex}) and backward
(Section~\ref{s:backward-ex}) properties.

\subsection{Forward Properties}
\label{s:forward-ex}

\technicalreport{Weighted Subject reduction for the $\lmuex$-calculus
(Lemma~\ref{l:psr}) is based on the fact that linear
substitution (Lemma~\ref{l:partial-substitution}) and linear replacement
 (Lemma~\ref{l:partial-replacement}) preserve types. }

Weighted Subject Reduction for the $\lmuex$-calculus
(Property~\ref{l:psr}) is based on two key properties, called
respectively the {\bf Linear Substitution} and the \textbf{Linear
  Replacement} Lemmas.  These properties may simply be understood as a
refinement of the Substitution Lemma~\ref{l:substitution} and the
Replacement Lemma~\ref{l:replacement} to the case of \modifref{the small-step $\lmuex$-calculus}.  Their
  precise statements and proofs can be found in the Appendix
(Lemmas~\ref{l:partial-substitution} and~\ref{l:partial-replacement}).

\technicalreport{
\begin{lemma}[\textbf{Linear  Substitution}]
\label{l:partial-substitution}
Let $\Theu \tri \Gamu  \Vdash u: \IM \mid \Delu$. If 
$\Phi_{\cxot\cwc{x}} \tri \Gam; x: \IM \vdash \cxot\cwc{x}: \Any \mid
\Del$, then \delia{there exist} $\IM_1, \IM_2, \Gamuu, \Gamud, \Deluu, \Delud$  s.t. 
\begin{itemize}
\item $\IM = \IM_1 \inter \IM_2$, where $\IM_1 \neq \emul$, 
\item $\Gamu = \Gamuu \inter \Gamud$ and $\Delu = \Deluu \union \Delud$, 
\item $\Theuu \tri \Gamuu \Vdash u: \IM_1 \mid \Deluu$,
\item $\Theud \tri \Gamud \Vdash u: \IM_2 \mid \Delud$,
\item $\Phi_{\ctx\cwc{u}} \tri \Gam \inter \Gamuu; x: \IM_2 \vdash
\cxot\cwc{u}: \Any \mid \Del \union \Deluu$, and 
\item $\sz{\Phi_{\cxot\cwc{u}}} = \sz{\Phi_{\cxot\cwc{x}}} + \sz{\Theuu}  - |  \IM_1|$.
\end{itemize}
\end{lemma}}
\ignore{
\begin{proof}
By induction on the context $\cxot$ using Lemma~\ref{l:relevance-bis}.
\end{proof}
}
\technicalreport{ 
  \begin{proof}
    The proof is by induction on the context $\cxot$
so we need to prove the statement of the lemma 
for regular derivations simultaneously with the following one
for  \textit{non-empty} auxiliary derivations:
if $\Phi_{\cxtt\cwc{x}} \tri \Gam; x: \IM \Vdash \cxtt\cwc{x}: \JM \mid
\Del$ and $\JM \neq \emul$, then \delia{there exist} $\IM_1, \IM_2, \Gamuu, \Gamud, \Deluu, \Delud$  s.t. 
\begin{itemize}
\item $\IM = \IM_1 \inter \IM_2$, where $\IM_1 \neq \emul$, 
\item $\Gamu = \Gamuu \inter \Gamud$ and $\Delu = \Deluu \union \Delud$, 
\item $\Theuu \tri \Gamuu \Vdash u: \IM_1 \mid \Deluu$,
\item $\Theud \tri \Gamud \Vdash u: \IM_2 \mid \Delud$,
\item $\Phi_{\cxtt\cwc{u}} \tri \Gam \inter \Gamuu; x: \IM_2 \Vdash
\cxtt\cwc{u}: \JM \mid \Del \union \Deluu$, and 
\item $\sz{\Phi_{\cxtt\cwc{u}}} = \sz{\Phi_{\cxtt \cwc{x}}} + \sz{\Theuu}  - | \IM_1|$.
\end{itemize}

Notice that 
$\IM\neq\emul$ by Lemma~\ref{l:relevance-bis}, since 
$x \in \fv{\cxot\cwc{x}} $ (resp. $x \in \fv{\cxtt\cwc{x}})$.
We only show the case $\cxot = \Box$ since all the other ones are straightforward. 
So assume $\cxot=\Box$. Then $\IM=\mult{\UM}$ for some $\UM$
and  the derivation $\Phi_{x}$ has the following form :
\[ \Phi_x = \infer{ }
   {\muju{x:\mult{\UM}}{x:\UM}{\es} } \]

   Thus, $\sz {\Phi_{x }} = 1$. We set then $\Theuu=\Theu$ and $\Theud=\infer{}{\muJu{\phdot }{u:\emul}{\phdot}}$.
   We have 
$\sz{\Phi_{u}} =    \sz{\Theuu} = \sz{\Phi_x} + \sz{\Theu} - 
| \IM_1|  $
since $|\IM_1|=1$.
  \end{proof}
}\medskip


\technicalreport{

\begin{lemma}[\textbf{Linear Replacement}]
\label{l:partial-replacement}
Let $\Theu \tri \muJu{\Gamu}{u:\inter_{\lL} \choice{\IMl}}{\Delu}$ s.t. $\al \notin \fv{u}$. 
If $\Phi_{\cxoc{\cwc{\coal t}}} \tri \muju{\Gam}{\cxoc\cwc{\co{\al}t}:\Any}{\al: \umult{ \IMl \rew \VMl}_{\lL};\Del}$, 
then \delia{there exist} $L_1, L_2, \Gamuu, \Gamud, \Deluu, \Delud, \Phi_{\cxoc\cwc{\co{\al'} tu}} $  s.t.
\begin{itemize}
\item $L=L_1\uplus L_2$, where $L_1\neq \es$.
\item $\Gamu = \Gamuu \inter \Gamud$ and $\Delu = \Deluu \union \Delud$, 
\item $\Theuu \tri \muJu{\Gamuu}{u: \inter_{\lL_1} \choice{\IMl}}{\Deluu}$,
\item $\Theud \tri \muJu{\Gamud}{u: \inter_{\lL_2} \choice{\IMl} }{\Delud}$,
\item $\Phi_{\cxoc\cwc{\co{\al'} tu}} \tri \muju{\Gam \inter \Gamuu}{\cxoc{\cwc{\co{\al'} tu}}:\Any}{\al: \umult{ \IMl \rew \VMl}_{\lL_2}; \al':\union_{\lL_1} \VMl \union \Del \union \Deluu}$, and 
\item $\sz{\Phi_{\cxoc\cwc{\co{\al'} tu}}} = \sz{\Phi_{\cxoc\cwc{\coal t}}} + \sz{\Theuu}$.
\end{itemize}
\end{lemma}
}

\ignore{
\begin{proof} The proof is by induction on the context $\cxoc$
  using Lemmas~\ref{l:relevance-bis},~\ref{l:non-empty-union}
  and~\ref{l:decomposition}.  
\end{proof}
}

\technicalreport{

\begin{proof}
The proof is by induction on the context $\cxoc$ 
so we need to prove the statement of the lemma 
for regular derivations simultaneously with the following one
for  \textit{non-empty} auxiliary derivations:
 if $\Phi_{\cxtc{\cwc{\coal t}}} \tri \muJu{\Gam}{\cxtc\cwc{\co{\al}t}:\JM}{\al: \umult{ \IMl \rew \VMl}_{\lL};\Del}$
and $\JM \neq \emul$, then \delia{ there exist } $L_1, L_2, \Gamuu, \Gamud, \Deluu, \Delud, \Phi_{\cxtc\cwc{\co{al'} tu}} $  s.t.
\begin{itemize}
\item $L=L_1\uplus L_2$, where $L_1\neq \es$.
\item $\Gamu = \Gamuu \inter \Gamud$ and $\Delu = \Deluu \union \Delud$, 
\item $\Theuu \tri \muJu{\Gamuu}{u: \inter_{\lL_1} \choice{\IMl}}{\Deluu}$,
\item $\Theud \tri \muJu{\Gamud}{u: \inter_{\lL_2} \choice{\IMl} }{\Delud}$,
\item $\Phi_{\cxtc\cwc{\co{\al'} tu}} \tri \muJu{\Gam \inter \Gamuu}{\cxtc{\cwc{\co{\al'} tu}}:\JM}{\al: \umult{ \IMl \rew \VMl}_{\lL_2}; \al':\union_{\lL_1} \VMl \union \Del \union \Deluu}$, and 
\item $\sz{\Phi_{\cxtc\cwc{\co{\al'} tu}}} = \sz{\Phi_{\cxtc\cwc{\coal t}}} + \sz{\Theuu}$.
\end{itemize}
Notice that $L \neq \es$ by Lemma~\ref{l:relevance-bis}, since $\alpha \in \fn{\cxoc\cwc{\co{\al}t}} $ (resp. $\alpha \in \fn{\cxtc\cwc{\co{\al}t}}$). 
We only show the case $\cxoc = \boxdot$ since all the other ones are straightforward.

So assume $\cxoc=\boxdot$. Then the derivation $\Phi_{\co{\al}t}$ has the
following form, where $K \neq \es$ holds by Lemma~\ref{l:non-empty-union}: 
\[ \infer{\Phi_t \tri \tyj{t}{\Gam}{ \umult{ \IMk \rew \VMk}_{\kK} \mid \al: \umult{ \IMl \rew \VMl}_{\lL \sm K} ; \Del}}
   {\tyj{\co{\al}t}{\Gam}{\TypCom \mid \al: \umult{ \IMl \rew \VMl}_{\lL} ; \Del}}\]

Thus, $\sz {\Phi_{\co{\al}t }} =\sz{\Phi_t} + \ar{\umult{ \IMk \rew \VMk}_{\kK} } = \sz{\Phi_t} + |K| + \ar{\union_{\kK} \VMk}$.
We set $L_1=K$ and $L_2=L\setminus K$ and we write
$\inter_{\lL} \choice{\IMl}$ as $(\inter_{\lL_1} \choice{\IMl})\inter (\inter_{\lL_2} \choice{\IMl})$. Then by Lemma~\ref{l:decomposition}, 
 there are $\Theuu \tri \muJu{\Gamuu}{u:\inter_{\lL_1} \choice{\IMl}}{\Deluu},~ \Theud \tri \muJu{\Gamud}{u:\inter_{\lL_2} \choice{\IMl}}{\Delud}$ s.t. $\Gamuu\inter \Gamud =\Gamu,~ \Deluu \union \Delud =\Delu$. 
We set $\VM=\union_{\lL_1} \VMl$
and then construct the following derivation $\Phi_{\co{\al'}tu}$: 
\[ \infer{\infer{\Phi_t \\ \Theuu}
              { \muju{\Gam \inter \Gamuu}{ tu: \union_{\lL_1} \VMl}{\al: \umult{ \IMl \rew \VMl}_{\lL_2} ; \Del \union  \Deluu}}}
       { \muju{\Gam \inter \Gamuu}{\co{\al'}tu: \TypCom}{\al: \umult{ \IMl \rew \VMl}_{\lL_2}; \al': \VM \vee  \Del\union  \Deluu}}\]

We have: 
\[ \begin{array}{l}
   \sz{\Phi_{\co{\al'}tu}} =    \sz{\Phi_{tu}} + \ar{\union_{\kK} \VMk}  \\
   = \sz{\Phi_{t}} + \sz{\Theuu} + |K| + \ar{\union_{\kK} \VMk}    \\
   = \sz{\Phi_{t}} + \sz{\Theuu} +  \ar{ \umult{ \IMk \rew \VMk}_{\kK} }    \\
   =  \sz{\Phi_{\co{\al}t}} + \sz{\Theuu}     \\
   \end{array} \]
\end{proof}
}


\begin{property}[\textbf{Weighted Subject Reduction for $\lmuex$}]
\label{l:psr}
Let $\Phi \tri \tyj{o}{\Gam}{\Any\mid\Del}$. If $o \Rew{} o'$ is a non-erasing step, then $\Phi' \tri \tyj{o'}{\Gam}{\Any\mid\Del}$ and  $\sz{\Phi} > \sz{\Phi'}$. 
\end{property}

\begin{proof}
By induction on the relation $\rew$ using
Lemma~\ref{l:non-empty-union}, Lemma~\ref{l:relevance-bis} and the
Linear Substitution and Replacement Lemmas mentioned above.  See the
Appendix for details.
\end{proof}

\technicalreport{
\begin{proof}
By induction on the reduction relation $\rew$. We only show the main cases of reduction at the root, the other ones being straightforward.

\begin{itemize}
\item If $o = (\slist[(\lambda x.t)]) u \Rew{} \slist[t[x/u]]= o'$: we proceed by induction on $\slist$,
by detailing only the case $\slist = \Box$ as the other one is straightforward. 

The derivation $\Phi$ has the following form:
$$\Phi = \infer{ \infer*{\Phi_t\rhd \Gam_t ; x:\IM \vdash t:\UM~|~
    \Del_t} {\Gam_t \vdash \lambda x.t: \umult{\IM \rew \UM} \mid
    \Del_t } \\ \Theu \rhd \Gamu
    \Vdash u:  \choice{\IM} \mid  \Del}  {\Gam \vdash
  (\l x. t)u:\UM \mid \Del } $$ where $\Gam = \Gam_t \inter
\Gamu$, $\Del=\Del_t\union  \Delu$ and $\Any = \UM$. 
We then construct the following derivation $\Phi'$:
$$
\infer{\Phit\rhd \Gam_t ;  x:\IM  \vdash t:\UM \mid \Delt \\ 
       \Theu \rhd \Gam_u \Vdash u:\choice{\IM}   \mid  \Delu}
      {\Gamt \inter \Gamu  \vdash t[x/u]:\UM  \mid \Del_t\union \Delu  }$$
We  have: 
\[ \begin{array}{lll}
   \sz{\Phi} =    \sz{\Phi_{t}} +  \sz{\Theu} + 2  \\
   > \sz{\Phi_t} +  \sz{\Theu}  =    \sz{\Phi'}
   \end{array} \] 

\item If $o = (\slist [\mu \al. \Com]) u \Rew{}  \slist[\mu \al'. \Com\rempl{\al}{\al'}{u}] = o'$: we proceed by induction on $\slist$,
by detailing only the case $\slist = \Box$ as the other one is straightforward.
The  derivation $\Phi$ has the following form:
$${\small \Phi= 
 {\infer{ \infer*{\PhiCom \tri 
       \muju{\GamCom}{\Com:\TypCom}{\al:\VMC;\DelCom
     }}
                  {\muju{\GamCom}{\mu\al.\Com:\VMC}{\DelCom}} \\
                   \Theu \tri \muJu{\Gamu}{u:\IMu}{\Delu}
                          }
         { \muju{\GamCom\inter \Gamu}{(\mu\al.\Com)u:\UM}{\Delu}}
  }}
$$
 where $ \VMC= \umult{\IMl \rew \VMl}_{\lL}$, $\IMu= \inter_{\lL}   \choice{\IMl}$, $\UM= \vee_{\lL} \VMl$, $\Gam = \GamCom \inter \Gamu$ and $\Del = \DelCom \union \Delu$.
 
Moreover, Lemma~\ref{l:non-empty-union} gives $L \neq \es$,
so that $\inter_{\lL} \choice{\IMl} = \choice{(\inter_{\lL} \choice{\IMl})}$.

We then construct the following derivation $\Phi'$:
$$
\infer{\infer{\PhiCom \\ \Theu }
             {\Gam' \inter \Gamu \vdash \Com\rempl{\al}{\al'}{u}: \TypCom \mid \Del' \union  \Delu; \al': \UM }}
      {\Gam' \inter \Gamu  \vdash  \mu \al'. \Com\rempl{\al}{\al'}{u}:\UM  \mid \Del' \union  \Delu  }
$$
We conclude since $|L|\geq 1$ in the following equation:
\[ \begin{array}{l} 
   \sz{\Phi'}  
   = \sz{\Phi_{\Com\rempl{\al}{\al'}{u}} } + 1 \\
   = \sz{\Phi_{\Com}} + \sz{\Theu} +  |L| - \frac{1}{2} + 1  \\ 
   = \sz{\Phi_{\mu \al. \Com}} + \sz{\Theu} + |L| - \frac{1}{2}  \\
   < \sz{\Phi_{\mu \al. \Com}} + \sz{\Theu } + |L|    = \sz{\Phi}
         \end{array} \]

\item If $o = \cxtt\cwc{x}[x/u] \Rew{} \cxtt\cwc{u}[x/u] = o'$, with $|\cxtt\cwc{x}|_x>1$.
The derivation $\Phi$ has the following form:
{\small 
  $$\infer{\Phi'_{\cxtt\cwc{x}}
    \!\tri\! \Gam'; x:  \IM \vdash \cxtt\cwc{x}: \UM  \mid \Del'  \\ 
         \Theu  \! \tri \! \Gamu  \vdash u :  \choice{\IM}  \mid \Delu } 
        {\Gam' \inter \Gamu \vdash \cxtt\cwc{x}[x/u]: \UM\mid \Del' 
       \union  \Delu}$$}
Moreover, $|\cxtt\cwc{x}|_x>1$ so that 
Lemma~\ref{l:relevance-bis} applied to $\Phi_{\cxtt\cwc{x}}$ gives $\IM \neq \emul$ and thus   $\choice{\IM} = \IM$.
We can then apply Lemma~\ref{l:partial-substitution} which gives  a derivation
\[\Phi_{\cxtt\cwc{u}} \tri \Gam' \inter \Gamuu; x: \IM_2   \vdash
\cxtt\cwc{u}: \UM \mid \Del' \union \Deluu \] 
where $\IM = \IM_1 \inter \IM_2$ and $\IM_1 \neq \emul$
and $\Gamu = \Gamuu \inter \Gamud$ and $\Delu = \Deluu \union \Delud$. 
Moreover $\Theuu \tri \Gamuu \Vdash u: \IM_1 \mid \Deluu$,
$\Theud \tri \Gamud \Vdash u: \IM_2 \mid \Delud$, 
and $\sz{\Phi_{\cxtt\cwc{u}}} = \sz{\Phi_{\cxtt\cwc{x}}} + \sz{\Theuu}  - | \IM_1 | $.

The  hypothesis $|\ctx\cwc{x}|_x>1$ implies  $|\ctx\cwc{u}|_x>0$, then 
$ \IM_2 \neq \emul$ by Lemma~\ref{l:relevance-bis} applied to $\Phi_{\cxtt\cwc{u}}$ so that $\choice{\IM_2} = \IM_2 $.
We can then  construct the derivation $\Phi'$ as follows:
$$\infer{\Phi_{\cxtt\cwc{u}} \\
         \Theud}
        { \Gam' \inter \Gam \vdash \cxtt\cwc{u}[x/u]: \UM \mid \Del' \inter \Del }\ (\subs)
$$
We conclude since
$\sz{\Phi'} = \sz{\Phi_{\cxtt\cwc{u}}} +   \sz{\Theud} =_{Lemma~\ref{l:partial-substitution}}
\sz{\Phi_{\cxtt\cwc{x}}} + \sz{\Theuu}  - |  \IM_1 |  +   \sz{\Theud} = \sz{\Phi_{\cxtt\cwc{x}}} + 
\sz{\Theu} - | \IM_1 |  < \sz{\Phi}$. 

The step $<$ is justified by $\IM_1 \neq \emul$.
\item If $o = \ctx\cwc{x}[x/u] \Rew{} \ctx\cwc{u} = o'$, with $|\ctx\cwc{x}|_x = 1$.
The derivation $\Phi$ has the following form:
{\small
$$\infer{\Phi_{\cxtt\cwc{x}}\! \tri \!\Gam'; x:  \IM \vdash \cxtt\cwc{x}: \UM  \mid \Del'  \\ 
        \Theu  \!\tri\! \Gamu   \Vdash u : \choice{\IM}  \mid \Delu } 
        {\Gam' \inter \Gamu \vdash \cxtt\cwc{x}[x/u]: \UM\mid \Del'   \union \Delu}$$}
Lemma~\ref{l:relevance-bis} applied to $\Phi_{\cxtt\cwc{x}}$ gives $\IM \neq \emul$ and thus   $\choice{\IM} = \IM$.
We can then apply Lemma~\ref{l:partial-substitution} which gives  a derivation
\[\Phi_{\cxtt\cwc{u}} \tri \Gam' \inter \Gamuu; x: \IM_2   \vdash
\cxtt\cwc{u}: \UM \mid \Del' \union \Deluu \] 
where $\IM = \IM_1 \inter \IM_2$ and $\IM_1 \neq \emul$
and $\Gamu = \Gamuu \inter \Gamud$ and $\Delu = \Deluu \union \Delud$. 
Moreover $\Theuu \tri \Gamuu \Vdash u: \IM_1 \mid \Deluu$,
$\Theud \tri \Gamud \Vdash u: \IM_2 \mid \Delud$, 
and $\sz{\Phi_{\cxtt\cwc{u}}} = \sz{\Phi_{\cxtt\cwc{x}}} + \sz{\Theuu}  - |  \IM_1|$.
By hypothesis $|\cxtt\cwc{x}|_x=1$ so that $|\cxtt\cwc{u}|_x=0$, then 
$\IM_2= \es$ by Lemma~\ref{l:relevance-bis} applied to $\Phi_{\cxtt\cwc{u}}$. Thus $\IM = \IM_1$. 
We then set $\Phi' = \Phi_{\cxtt\cwc{u}}$ and conclude since
\[ \begin{array}{l}
   \sz{\Phi'} = \sz{\Phi_{\cxtt\cwc{u}}}  \\
   =_{Lemma~\ref{l:partial-substitution}} \sz{\Phi_{\cxtt\cwc{x}}} +  \sz{\Theuu}  - | \IM_1| \\
   = \sz{\Phi_{\cxtt\cwc{x}}} +  \sz{\Theu}    - |  \IM|   < \\
   = \sz{\Phi_{\cxtt\cwc{x}}} +  \sz{\Theu}      = \sz{\Phi}
   \end{array} \]

The step $<$ is justified by $\IM = \IM_1 \neq \emul$.

\item If $o = \cxcc\cwc{\co{\al}t} \rempl{\al}{ \al'}{u}   \Rew{} \cxcc\cwc{\co{\al'}t u} \rempl{\al}{ \al'}{u} = o'$, with $|\cxcc\cwc{\co{\al}t}_\al > 1$. Then $\Phi$ has the following form

{\small
$$
  \infer{\Phi_{\Com}
    \!\tri\! \GamCom \! \vdash \! \cxcc\cwc{\co{\al}t}: \TypCom \mid \! \DelCom; \al: \VM'\!\!\! \\ \Theu \!\tri\! \Gamu \! \Vdash \! u: \IMu \! \mid \!\Delu}    {\GamCom \inter \Gamu \vdash \cxcc\cwc{\co{\al}t} \rempl{\al}{ \al'}{u}: \TypCom \mid \DelCom \union \Delu \vee \al': \union_{\lL} \VMl}  $$}
where $\Com = \cxcc\cwc{\co{\al}t} $, $ \VM'= \umult{\IMl \rew
  \VMl}_{\lL}$, $\IMu= \choice{ (\inter_{\lL} \choice{\IMl})}$,
$\AM = \TypCom$, $\Gam = \GamCom \inter \Gamu$ and
$\Del = \DelCom \union \Delu\vee \al': \union_{\lL} \VMl$.
Since $|\cxcc\cwc{\co{\al}t}|_\al > 1$ implies $L \neq
\es$ by Lemma~\ref{l:relevance-bis}, we have that  $\IMu=\inter_{\lL}
\choice{\IMl}$.  By Lemma~\ref{l:partial-replacement} there are $L_1,\,
L_2,\ \Gamuu,\, \Gamud,\ \Deluu,\, \Delud,\ \Phi_{\cxcc\cwc{\co{\al'}
    tu}} $ s.t.
\begin{itemize}
\item $L=L_1\uplus L_2$, where $L_1\neq \es$.
\item $\Gamu = \Gamuu \inter \Gamud$ and $\Delu = \Deluu \union \Delud$, 
\item $\Theuu \tri \muJu{\Gamuu}{u: \inter_{\lL_1} \choice{\IMl}}{\Deluu}$,
\item $\Theud \tri \muJu{\Gamud}{u: \inter_{\lL_2} \choice{\IMl} }{\Delud}$,
\item $\Phi_{\cxcc\cwc{\co{\al'} tu}} \tri \muju{\GamCom \inter \Gamuu}{\cxcc{\cwc{\co{\al'} tu}}:\AM}{\al: \umult{ \IMl \rew \VMl}_{\lL_2}; \al':\union_{\lL_1} \VMl \union \DelCom \union \Deluu}$, and 
\item $\sz{\Phi_{\cxcc\cwc{\co{\al'} tu}}} = \sz{\Phi_{\cxcc\cwc{\coal t}}} + \sz{\Theuu}$.
\end{itemize}

   Moreover, $|\cxcc\cwc{\co{\al}t}|_\al > 1$ implies $|\cxcc\cwc{\co{\al'}tu}|_\al > 0$ so that
   $L_2 \neq \es$ holds by Lemma~\ref{l:relevance-bis} and thus
   $\inter_{\lL_2} \choice{\IMl}= \choice{(\inter_{\lL_2} \choice{\IMl})}$. Then we can build the
   following derivation $\Phi'$:

   $${ \infer{\Phi_{\cxcc\cwc{\co{\al'}tu} } \\
            \Theud }
     {\Gam' \vdash \cxcc\cwc{\co{\al'}tu} \rempl{\al}{ \al'}{u}: \TypCom \mid \Del' }
   }$$
   where $\Gam'=(\GamCom \inter \Gamuu)\inter \Gamud =\Gam$, $\Del'=(\al':\union_{\lL_1} \VMl \union \DelCom \union \Deluu) \union \Delud \union (\al':\union_{\lL_2} \VMl)=\Del$.

   We conclude since
 \[ \begin{array}{l}
       \sz{\Phi'} = \sz{\Phi_{\cxcc\cwc{\co{\al'}tu} }} + \sz{\Theud} + | L_2| - \frac{1}{2}  \\
       =_{Lemma~\ref{l:partial-replacement}}    \sz{\Phi_{\cxcc\cwc{\co{\al}t}}} + \sz{\Theuu} + \sz{\Theud} + |L_2| - \frac{1}{2}   \\
       = \sz{\Phi_{\cxcc\cwc{\co{\al}t}}} + \sz{\Theu} + |L_2| - \frac{1}{2}    \\
       < \sz{\Phi_{\cxcc\cwc{\co{\al}t}}} + \sz{\Theu} + | L | - \frac{1}{2}         = \sz{\Phi}      
 \end{array} \]

The step $<$ is justified because $L_1 \neq \es$ and thus $|L_2 | < |L|$.

\item If $o = \cxcc\cwc{\co{\al}t} \rempl{\al}{ \al'}{u}   \Rew{} 
 \cxcc\cwc{\co{\al'}t u}  = o'$, with 
$|\cxcc\cwc{\co{\al}t}|_\al =  1$. The derivation $\Phi$ has the following form

$${\small 
  \infer{\Phi_{\Com}
    \!\tri\! \GamCom \! \vdash \! \cxcc\cwc{\co{\al}t}: \TypCom \mid \! \DelCom; \al: \VM'\!\!\! \\ \Theu \!\tri\! \Gamu \! \Vdash \! u: \IMu \! \mid \!\Delu}    {\GamCom \inter \Gamu \vdash \cxcc\cwc{\co{\al}t} \rempl{\al}{ \al'}{u}: \TypCom \mid \DelCom \union \Delu \vee \al': \union_{\lL} \VMl}  }$$
 
 where $\Com = \cxcc\cwc{\co{\al}t} $, $ \VM'= \umult{\IMl \rew \VMl}_{\lL}$, $\IMu= \choice{ (\inter_{\lL}   \choice{\IMl})}$,
 $\AM = \TypCom$, $\Gam = \GamCom \inter \Gamu$ and $\Del = \DelCom \union \Delu\union \al': \union_{\lL} \VMl$. 
Since $|\cxcc\cwc{\co{\al}t}|_\al = 1$ implies $L \neq \es$ by Lemma~\ref{l:relevance-bis}, we have that $\IMu=\inter_{\lL}   \choice{\IMl}$.
By Lemma~\ref{l:partial-replacement} there are
$L_1,\, L_2,\ \Gamuu,\, \Gamud,\ \Deluu,\, \Delud,\ \Phi_{\cxcc\cwc{\co{\al'} tu}} $  s.t.
\begin{itemize}
\item $L=L_1\uplus L_2$, where $L_1\neq \es$.
\item $\Gamu = \Gamuu \inter \Gamud$ and $\Delu = \Deluu \union \Delud$, 
\item $\Theuu \tri \muJu{\Gamuu}{u: \inter_{\lL_1} \choice{\IMl}}{\Deluu}$,
\item $\Theud \tri \muJu{\Gamud}{u: \inter_{\lL_2} \choice{\IMl} }{\Delud}$,
\item $\Phi_{\cxcc\cwc{\co{\al'} tu}} \tri \muju{\GamCom \inter \Gamuu}{\cxcc{\cwc{\co{\al'} tu}}:\AM}{\al: \umult{ \IMl \rew \VMl}_{\lL_2}; \al':\union_{\lL_1} \VMl \union \DelCom \union \Deluu}$, and 
\item $\sz{\Phi_{\cxcc\cwc{\co{\al'} tu}}} = \sz{\Phi_{\cxcc\cwc{\coal t}}} + \sz{\Theuu}$.
\end{itemize}

   Moreover, $|\otx\cwc{\co{\al}t}|_\al = 1$ implies $|\otx\cwc{\co{\al'}tu}|_\al = 0$ so that
   $L_2 = \es$ and $L=L_1$ holds by Lemma~\ref{l:relevance-bis}. Thus, $\Theuu=\Theu$ and so on. We then set $\Phi' = \Phi_{\otx\cwc{\co{\al'}tu} }$
   and  conclude since

 \[ \begin{array}{l}
       \sz{\Phi'} 
       = \sz{\Phi_{\otx\cwc{\co{\al'}tu} }}  \\
       =_{Lemma~\ref{l:partial-replacement}}    \sz{\Phi_{\otx\cwc{\co{\al}t}}} + \sz{\Theu}   \\
       <  \sz{\Phi_{\otx\cwc{\co{\al}t}}} + \sz{\Theu}  + | L | - \frac{1}{2}   
       = \sz{\Phi} \\ 
   \end{array} \] 

 The step $<$ is justified because $L\neq \es$, so that $|L|\geqslant 1$
 implies  $|L|- \frac{1}{2}>0$.

\ignore{samedi 29 octobre: en ecrivant les cas effacants, j'ai l'impression que notre enonce est faux tout compte fait (pareil pour la subj red implicite) : dans le cas effacant, on peut passer au context seulement dans le cas ou on a $\Box \, t_1\ldots t_n$, pas quand $\Box$ est lui-meme en argument (en fait, des qu'on fait des abstractions, on peut avoir des problemes...). Je continue la reecriture sur partial-reverse...}

\ignore{
\item \pierre{If $o=t[x/u] \rew t=o'$ with $x\notin t$. 
  The derivation $\Phi$ has the following form
  $$
  \infer{\Phit \tri \muju{\Gamt}{t:\UM}{\Delt} \\  \Theu\tri \muJu{\Gamu}{u:\choice{\emul}}{\Delu}
  }{\muju{\Gamt\inter \Gamu}{t[x/u]:\UM}{\Delt \union \Delu}}
  $$
  where $x\notin \dom{\Gamt}$ by  Lemma~\ref{l:relevance-bis}. We set then $\Phi'=\Phit$. We conclude since $\sz{\Phi'}=\sz{\Phit}<\sz{\Phit}+\sz{\Theu}\odelia{+1}=\sz{\Phi}$.}
\item \pierre{If $o = \Com \rempl{\al}{ \al'}{u}   \rew \Com$, with 
$\al \notin \fn{\Com} $. The derivation $\Phi$ has the following form
 $$ \infer{\PhiCom\tri \muju{\GamCom}{\Com:\TypCom}{\DelCom} \\ \Theu\tri \muJu{\Gamu}{u:\choice{\emul}}{\Delu} }{
    \muju{\GamCom}{\Com \rempl{\al}{\al'}{u}:\TypCom}{\DelCom\union \Delu}}
  $$
  where $\al \notin \dom{\DelCom}$ by Lemma~\ref{l:relevance-bis}.
  We set then $\Phi'=\PhiCom$. We conclude since $\sz{\Phi'}=\sz{\PhiCom}<\sz{\PhiCom}+\sz{\Theu}-1/2$. Indeed, $\Theu$ is non empty, so that
  $\sz{\Theu}\geqslant 1$ implies $\sz{\Theu}-1/2>0$.}}
\end{itemize}
\end{proof}
}


\subsection{Backward Properties}
\label{s:backward-ex}

As in the implicit case (Section~\ref{s:backward}), subject expansion
for non-erasing $\lmuex$-step relies on \textbf{(Linear) Reverse
  Substitution} and \textbf{(Linear) Reverse Replacement} Lemmas: if
$\Phi'\tri \muju{\Gam}{o':\Any}{\Del}$ and $o'$ has been obtained from
$o$ by substituting one occurrence of $x$ by $u$ (or one subcommand
$\coal t$ by $\co{\al'}tu$), then, informally speaking, it is possible
to decompose $\Phi'$ into a regular derivation $\Phi_0$ typing $o$ and
an auxiliary derivation $\Theu$ typing $u$. The precise statements and proofs can be found in the Appendix (Lemma~\ref{l:reverse-partial-substitution}
and~\ref{l:reverse-partial-replacement}).

\technicalreport{
\begin{lemma}[\textbf{Reverse Partial Substitution}]
\label{l:reverse-partial-substitution}
Let $\Phi \tri \Gam \vdash \cxot\cwc{u}: \Any \mid \Del$, where $x\notin \fv{u}$.
Then, \delia{there exists} $\Gam_0, \Del_0, \IM_0\neq \emul, \Gamu, \Delu$ such that 
\begin{itemize}
   \item $\Gam = \Gam_0 \inter \Delu$, 
     \item $\Del = \Del_0 \union \Delu$, 
     \item $\Phi_{\cxot\cwc{x}} \tri \muju{\Gam_0 \inter x:\IM_0}{\cxot\cwc{x}:\Any }{\Del_0}$       
     \item $\tri \muJu{\Gamu}{u:\IM_0}{\Delu}$.
\end{itemize}
\end{lemma}

\begin{proof}  The proof is by induction on the context $\cxot$. 
For this induction to work, we need as usual to adapt the statement for auxiliary derivations.
We only  show the  case $\cxot =  \Box$ since  all the other  ones are
straightforward and rely on suitable partitions of the contexts in the
premises. So assume $\cxot=\Box$,  then $\Any=\UM$ for some $\UM$.  We
set $\Gam_0 = \Del_0=\es$,  $\IM_0=\mult{\UM}$, $\Gamu = \Gam$, $\Delu
= \Del$  (so that  $\tri \muJu{\Gamu}{u:\IM_0}{\Delu}$ holds  by using
the $(\many)$ rule), and
\[ \Phi_{x} = \infer{ }{\muju{x:\mult{\UM}}{x:\UM}{\es} } \] 
  
The claimed set and context equalities trivially hold.   
\end{proof}

}

\ignore{ 
\begin{lemma}[\textbf{Reverse Partial Substitution}]
\label{l:reverse-partial-substitution}
Let $\Phi \tri \Gam \vdash \ctx\cwc{u}: \UM \mid \Del$, where $x\notin \fv{u}$.
Then, \delia{there exist} $\Gam_0, \Del_0, K\neq \es, \UMk)_{\kK}, (\Gamk)_{\kK},(\Delk)_{\kK}$ such that 
\begin{itemize}
\item $\Gam = \Gam_0 \inter_{\kK} \Gamk$, 
     \item $\Del = \Del_0 \union_{\kK} \Delk$, 
     \item $\Phi_{\ctx\cwc{x}} \tri \tyj{\ctx\cwc{x}}{x:\mult{\UMk}_{\kK} \inter  \Gam_0}{\UM  \mid \Del_0}$,  and 
     \item $(\Phiuk \tri \Gamk \vdash u: \UMk  \mid \Delk)_{\kK}$.
\end{itemize}

\end{lemma}

\begin{proof}  The proof is by induction on the context $\ctx$. 
We only show the case $\ctx = \Box$ since all the other ones are
straightforward and rely on suitable partitions of the contexts in the
premises. So assume $\ctx=\Box$. We set $\Gam_0 = \Del_0=\es$,  $K=\cset{ k_1}$, $\UM_{k_1}=\UM$,
$\Gam_{k_1} = \Gam$, $\Del_{k_1} = \Del$, $ \Phi_u^{k_1}= \Phi$, and 
\[ \Phi_{x} = \infer{ }{\muju{x:\mult{\UM}}{x:\UM}{\es} } \] 
  
The claimed set and context equalities trivially hold.   
\end{proof}
}   

\technicalreport{
\begin{lemma} [\textbf{Reverse Partial Replacement}]
  \label{l:reverse-partial-replacement}
  Let $\muju{\Gam}{\cxoc\cwc{\co{\al'}tu}:\Any}{\al':\VM;\Del}$, where  $\al,\al' \notin \fn{u}$.  Then \delia{there exist} $\Gam_0, \Del_0,  \VM_0, K \neq \es, (\IMk)_{\kK}, (\VMk)_{\kK}, \Gamu, \Delu$ such that 
\begin{itemize}
\item $\Gam = \Gam_0 \inter \Gamu$,
\item $\Del=\Del_0 \union \Delu$,
\item $\VM = \VM_0 \union_{\kK} \VMk$, 
\item $\tri \muju{\Gam_0}{\cxoc\cwc{\co{\al}t}: \Any}{\al':\VM_0;\al:\umult{\IMk\rew \VMk}_{\kK} \union   \Del_0}$,  and 
\item     $\tri \muJu{\Gamu}{u :  \inter_{\kK}\choice{\IMk} }{\Delu}$
\end{itemize}
\end{lemma}

\begin{proof} 
  The proof is by induction on the context $\cxoc$. For this
    induction to work, we need as usual to adapt the statement for auxiliary
    derivations.  Notice that $\VM \neq \umult{ \,}$ by
  Lemma~\ref{l:relevance-bis}, since $\al' \in
  \fn{\cxoc\cwc{\co{\al'}tu}} $.  We only show the case $\cxoc =
  \boxdot$ since all the other ones are straightforward.  So assume
  $\cxoc=\boxdot$. Then the derivation of $\co{\al'}tu$ has the
  following form, where $K\neq \es$: {\small
$$
\infer{\infer{\Phit\!  \tri \! \Gam_0 \!\vdash t\!:\!\!\umult{\IMk\rew \VMk}_{\kK} \!\!\mid
         \!\al'\!\!:\!\VM_0; \Del_0\!\!
    \\
   \tri \Gamu \Vdash \! u: \inter_{\kK}\choice{\IMk}\! \mid \!\Delu  }
  {\Gam_0 \inter \Gamu \vdash tu: \union_{\kK} \VMk \mid \al':\VM_0; \Del_0 \union \Delu }}
      {\Gam_0 \inter \Gamu \vdash \co{\al'}tu: \TypCom \mid \al': \VM_0\union_{\kK}\VMk; \Del_0 \union \Delu}
      $$}
where $\Gam =\Gam_0 \inter \Gamu$,
and $\Del=\Del_0 \union \Delu$ and $\VM = \VM_0 \union_{\kK}\VMk $. 

We then construct the following derivation :
\[ \infer{\Phit }{
 \muju{\Gam}{\co{\al}t:\TypCom}{\al':\VM_0;\al:\umult{\IMk\rew \VMk}_{\kK}\union \Del_0}}
\]
Thus, we have all the claimed set and context equalities.
\end{proof}
}


\begin{property}[Subject Expansion for $\lmuex$]
\label{l:pse}
Let $\tingD{\Phi'}{\tyj{o'}{\Gam}{\Any\mid\Del}}$. If $o \Rew{\nonelmuex} o'$
(\ie\  a non-erasing $\lmuex$-step), 
then 
$\tingD{\Phi}{\tyj{o}{\Gam}{\Any\mid\Del}}$.
\end{property}

\begin{proof} By induction on $\Rew{\nonelmuex}$ 
using the Linear Reverse Lemmas mentioned above. See the Appendix for details.
\end{proof}

\technicalreport{
\begin{proof}
By induction on the non erasing reduction relation $\Rew{\nonelmuex}$. We only show the main cases of 
non-erasing reduction at the root, the other ones being straightforward. 

\begin{itemize}
\item If $o = (\slist [\lambda x.t]) u \Rew{} \slist[t[x/u]] = o'$, we proceed by induction on $\slist$, by detailing only the case $\slist = \Box$ as the other one is straightforward. 

The derivation $\Phi'$ has the following form :
$$\infer{\Phit \tri  \muju{\Gamt;x:\IM}{t:\UM}{\Delt}
   \\ \Theu \tri \muJu{\Gamu}{u:\choice{\IM}}{\Delu}   }{\muju{\Gam}{u:\UM}{\Del}}$$

We then construct the following derivation $\Phi$: 
$$\infer{ \infer*{\Phit \tri \muju{\Gamt;x:\IM}{t:\UM}{\Delt}}{
    \muju{\Gamt}{\l x. t:\IM\rew \UM}{\Delt}}
\\ \Theu \tri \muJu{\Gamu}{u:\choice{\IM}}{\Delu}
}{\muju{\Gam}{(\l x.t)u:\UM}{\Del}}
$$

\item If $o = (\slist[\mu \al. \Com]) u \Rew{} \slist[\mu
  \al'. \Com\rempl{\al}{\al'}{u}] = o'$, where $\al'$ is fresh,
  then we proceed by induction on $\slist$, by detailing only the
  case $\slist = \Box$ as the other one is straightforward.  Then
  $\Phi'$ has the following form :
$${\scriptsize
\infer{\infer{\PhiCom \tri \muju{\GamCom}{\Com:\TypCom }{\al:\VMal;\DelCom}\!\!\!\! \\ \Theu \tri \muJu{\Gamu}{u:\IMu}{\Delu}
    }{\muju{\GamCom \inter \Gamu}{\Com\rempl{\al}{\al'}{u}: \TypCom}{\DelCom \union \Delu; \al': \VMalp}
  }}{\muju{\GamCom \inter \Gamu}{\mu \al'. \Com\rempl{\al}{\al'}{u}:\choice{(\VMalp)}}{ \DelCom \union \Delu}}
}
$$
where $ \VMal= \umult{\IMl \rew \VMl}_{\lL}$, $\IMu= \choice{(\inter_{\lL}   \choice{\IMl})}$, $\VMalp= \vee_{\lL} \VMl$,
$\Any=\choice{(\VMalp)} = \choice{(\vee_{\lL} \VMl)}$, $\Gam = \GamCom \inter \Gamu$ and $\Del = \DelCom \union \Delu$.
Notice that the  name assignment of the judgment typing
$\Com\rempl{\al}{\al'}{u}$ has the form  $\DelCom \union \Delu; \al': \VMalp$
since $\al'$ is a fresh name  by hypothesis,  so that $\al' \notin \dom{\DelCom \union \Delu}$ 
holds by Lemma~\ref{l:relevance-bis}.
We now consider two cases:

If $L \neq \es$, then  $\choice{\umult{\IMl\rew \VMl}_{\lL}} = \umult{\IMl\rew \VMl}_{\lL}$, 
$\choice{(\inter_{\lL}   \choice{\IMl})} = \inter_{\lL}   \choice{\IMl}$, 
$\Any = \choice{(\vee_{\lL} \VMl)} = \vee_{\lL} \VMl$, so that 
we construct the following derivation $\Phi$:
 $$
  {\scriptsize   
   \infer{\infer*{\PhiCom}
                 {\muju{\Gam_\Com }
                       {\mu \al.\Com:\umult{\IMl\rew \VMl}_{\lL}}{\Del_\Com}  }\!\! \\
          \Theu \!\tri \muJu{\Gamu}{u:\inter_{\lL}   \choice{\IMl} }{\Delu}}
         {\muju{\GamCom \inter \Gamu}
               { (\mu \al.\Com)u : \vee_{\lL} \VMl}
               {\DelCom\union \Delu}}
  } $$

If $L = \es$, then let $\choice{(\inter_{\lL}   \choice{\IMl})}$
(resp. $ \choice{(\vee_{\lL} \VMl)} $) 
be of the form $\mult{\UM}$ (resp. $\umult{\sig}$) 
for some arbitrary $\UM$ (resp. $\sig$).
Then we choose  $\choice{\umult{\IMl\rew \VMl}_{\lL}}$ to be
$ \umult{\mult{\UM} \rew \umult{\sig}}$. We then construct the following derivation $\Phi$:
 $$
  {\small 
   \infer{\infer*{\PhiCom}
                 {\muju{\Gam_\Com }
                       {\mu \al.\Com:\umult{\mult{\UM} \rew \umult{\sig}}}{\Del_\Com}  }\!\! \\
          \Theu \!\tri \muJu{\Gamu}{u: \mult{\UM}}{\Delu}}
         {\muju{\GamCom \inter \Gamu}
               { (\mu \al.\Com)u : \umult{\sig}}
               {\DelCom\union \Delu}}
  } $$
We conclude since $\Any = \umult{\sig}$.

\item If $o = \cxot\cwc{x}[x/u] \Rew{} \cxot\cwc{u}[x/u] = o'$, with $|\cxot\cwc{x}|_x>1$.
The derivation $\Phi'$ has the following form:

$$
\infer{\Phi_{\cxot\cwc{u}}\tri \Gams;x:\IM \!\vdash\! \cxot\cwc{u}\!:\!\Any \!\mid \!\Dels\!\!\!\! \\
         \Theu\!\tri \!\Gamu\! \Vdash \!u:\!\choice{\IM}\! \!\mid \!\Delu}
        { \muju{\Gams \inter \Gamu}{\cxot\cwc{u}[x/u]: \Any}{\Dels \union \Delu} }
$$
where $x \in \fv{\cxot\cwc{u}}$ implies  $\IM\neq \emul$ by Lemma~\ref{l:relevance-bis}, so that $\choice{\IM}=\IM$.

        By Lemma~\ref{l:reverse-partial-substitution} applied to
        $\Phi_{\cxot\cwc{u}}$, we have
$\Gam'_0,\, \Del_0,\, \IM_0\neq \emul,\, \Gamu', \, \Delu'$ such that 
\begin{itemize}
   \item $\Gams;x:\IM = \Gam'_0 \inter \Gamu'$, 
     \item $\Dels = \Del_0 \union \Delu'$, 
     \item $\Phi_{\cxot\cwc{x}} \tri \muju{\Gam'_0 \inter x:\IM_0}{\cxot\cwc{x}:\Any }{\Del_0}$       
     \item $\tri \muJu{\Gamu'}{u:\IM_0}{\Delu'}$.
\end{itemize}
We set $\IM^+=\IM\inter \IM_0,\, \Gamu^+=\Gamu \inter \Gamu',\, \Delu^+=\Delu \union \Delu'$. Thus in particular $\choice{({\IM^+})} = \IM^+$.  By Lemma~\ref{l:relevance-bis}, $x\notin \dom{\Gamu'}$, so that $\Gamo'=\Gamo;x:\IM$ for some $\Gamo$
and thus  $\Gamo'\inter x:\IM_0=\Gamo;x:\IM^+$.
By Lemma~\ref{l:decomposition} there is a derivation
$\Theu^+\tri \muJu{\Gamu^+}{u:\IM^+}{\Delu^+}$
We then construct the following derivation $\Phi$ : 
$$\infer{\Phi_{\cxot\cwc{x}}   \\ 
         \Theu^+ \tri \Gamu^+  \vdash u : \IM^+  \mid \Delu^+ }
        {\Gam_0 \inter  \Gamu^+ \vdash \cxot\cwc{x}[x/u]: \Any\mid \Del_0
          \union \Delu^+}\ (\subs)$$
        We conclude since $\Gam_0\inter \Gamu^+=\Gam_0\inter \Gamu' \inter \Gamu=\Gams\inter \Gamu=\Gam$ and $\Del_0 \union \Delu^+=\Del_0\union \Delu' \union \Delu=\Dels\union \Delu=\Del$.

\item If $o = \cxot\cwc{x}[x/u] \Rew{} \cxot\cwc{u} = o'$, with $|\cxot\cwc{x}|_x = 1$.
The derivation $\Phi'$ ends with $\muju{\Gam}{\cxot\cwc{u}:\Any}{\Del}$ where  $x\notin \dom{\Gam}$ by Lemma~\ref{l:relevance-bis}.
        By Lemma~\ref{l:reverse-partial-substitution} applied to $\Phi'$, we have
$\Gam_0,\, \Del_0,\, \IM_0\neq \emul,\, \Gamu, \, \Delu$ such that 
\begin{itemize}
   \item $\Gam = \Gam_0 \inter \Gamu$, 
     \item $\Del = \Del_0 \union \Delu$, 
     \item $\Phi_{\cxot\cwc{x}} \tri \muju{\Gam_0 \inter x:\IM_0}{\cxot\cwc{x}:\Any }{\Del_0}$       
     \item $\tri \muJu{\Gamu}{u:\IM_0}{\Delu}$.
\end{itemize}
Thus in particular $\choice{\IM_0}=\IM_0$. Since $x\notin \dom{\Gam}$, $x\notin \dom{\Gam_0}$,
so that $\Gam_0 \inter x:\IM_0=\Gam_0; x:\IM_0$.
We then construct the following derivation $\Phi$ :
$$\infer{\Phi_{\cxot\cwc{x}}   \\ 
         \tri \muJu{\Gamu}{u:\IM_0}{\Delu} }
{ \muju{\Gamo\inter \Gamu}{\cxot\cwc{x}[x/u]: \UM}{\Del_0 \union \Delu}}$$
We conclude since $\Gam = \Gam_0 \inter \Gamu$ and $\Del = \Del_0 \union \Delu$.\\

\item If $o = \cxoc\cwc{\co{\al}t} \rempl{\al}{ \al'}{u}   \Rew{} 
 \cxoc\cwc{\co{\al'}t u} \rempl{\al}{ \al'}{u} = o'$, with 
 $|\cxoc\cwc{\co{\al}t]}_\al > 1$.
Then  $\Phi'$ has the following form :
   $${ \small \infer{
             \Phio'\tri \Gams\!\! \vdash \!\! \cxoc\cwc{\co{\al'}tu} \!\!:\!\Any \!\!\mid \!\!
                \Dels ;  \! \al'\!\!\!:\! \VM_{\!\al'}; \al\!:\!\!\VM_{\al} \\
              \Theu \!\tri \!\muJu{\Gamu\!}{\!u\!:\!\IMu \!}{\!\Delu}}
             {\muju{\Gams\! \inter\! \Gamu\!}
                   {\!\cxoc\cwc{\co{\al'}tu} \rempl{\al}{ \al'}{u}:\!\Any}
                   {\! (\Dels; \al' : \VM_{\!\al'} )\! \union \Delu\!\union\! \al'\!:\!\VM}}
 } $$
where $\VMal= \umult{\IMl \rew \VMl}_{\lL}$, $\IMu= \choice{(\inter_{\lL}   \choice{\IMl})}$,
$\VM=  \vee_{\lL} \VMl$,
$\Gam = \Gams \inter \Gamu$ and
$\Del = (\Dels; \al' : \VM_{\!\al'} ) \union \Delu \union \al':\VM =
(\Dels \union \Delu; \al' : \VM_{\!\al'} \union \VM)$
since $\al'\notin \fn{u}$ implies $\al' \notin \dom{\Delu}$. Since $\al\in \fn{\cxoc\cwc{\co{\al'}tu}}$, then $L\neq \es$ by Lemma~\ref{l:relevance-bis}, so that $\IMu= \inter_{\lL}   \choice{\IMl}$.

By Lemma~\ref{l:reverse-partial-replacement} applied to $\Phio'$, we have
$ \Gamo,\, \Delo'$,$\, \Phio,\, \VM_0,\, K \neq \es, \, (\IMk)_{\kK},\, (\VMk)_{\kK},\, \Gamu',\Delu'$, and $\Theu'$ such that 
\begin{itemize}
\item $\Gams = \Gam_0 \inter \Gamu'$,
\item $\Dels;\al:\VM_\al=\Delo' \union \Delu'$,
\item $\VMalp = \VM_0 \union_{\kK} \VMk$, 
\item $\Phio\tri \muju{\Gam_0}{\cxoc\cwc{\co{\al}t}: \Any}{\al':\VM_0;\al:\umult{\IMk\rew \VMk}_{\kK} \union  \Delo'}$  and 
\item     $\tri \muJu{\Gamu'}{u :  \inter_{\kK}\choice{\IMk} }{\Delu'}$
\end{itemize}
We set $L^+=L \uplus K,\, \Gamu^+=\Gamu \inter \Gamu',\, \Delu^+=\Delu \union \Delu'$ and $\IMu^+=\inter_{\lL^+} \choice{\IMl}$. By Lemma~\ref{l:relevance-bis}, $\al\notin \dom{\Delu'}$, so that $\Delo'=\Delo;\al:\VM_{\al}$ for some $\Delo$ and $\al':\VM_0;\al:\umult{\IMk\rew \VMk}_{\kK} \union   \Delo'=\al':\VM_0;\al:\umult{\IMl\rew \VMl}_{\lL^+};\Delo$ since $\al'\notin \dom{\Delo'}$.
By Lemma~\ref{l:decomposition}, there is $\Theu^+\tri \muJu{\Gamu^+}{u:\IMu^+}{\Delu^+}$.
We then construct the following derivation $\Phi$: 
$${\small 
 \infer{\Phio \\ \Theu^+ }{\muju{\Gamo \inter  \Gamu^+\!}
  {\cxoc\cwc{\co{\al}t}\rempl{\al}{ \al'}{u}:\!\Any }{
    \al'\!:\!\VM_0\union_{\lL^+}\!\VMl; \Delo\union \Delu^+}  }
}$$
We conclude since
$\Gamo \inter \Gamu^+=\Gamo \inter\Gamu'  \inter  \Gamu =\Gams \inter \Gamu =\Gam$,
$\Delo\union \Delu^+=\Delo \union \Delu' \union \Delu=\Dels \union \Delu$ and
$\VM_0\union_{\lL^+} \VMl=\VM_0 \union_{\kK} \VMk \union_{\lL} \VMl=
\VMalp \union_{\lL} \VMl = \VMalp \union \VM$.

\item If $o = \cxoc\cwc{\co{\al}t} \rempl{\al}{ \al'}{u}   \Rew{} 
 \cxoc\cwc{\co{\al'}t u}  = o'$, with 
$|\cxoc\cwc{\co{\al}t}|_\al =  1$, then the derivation $\Phi'$ necessarily ends 
with the  judgment
 $\muju{\Gam}{\cxoc\cwc{\co{\al'}tu}:\Any}{\Dels;\al':\VM}$, where $\Del=\Dels;\al':\VM$.

By Lemma~\ref{l:reverse-partial-replacement} applied to  $\Phi'$,  
 we have $ \Gamo,\, \Delo,\, \VM_0,\, K \neq \es, \, (\IMk)_{\kK},\, (\VMk)_{\kK},\, \Gamu,\, \Delu$, and $\Theu$ such that 
\begin{itemize}
\item $\Gam = \Gamo \inter \Gamu$,
\item $\Dels=\Delo \union \Delu$,
\item $\VM= \VMo \union_{\kK} \VMk$, 
\item $\Phio\tri \muju{\Gamo}{\cxoc\cwc{\co{\al}t}: \Any}{\al':\VM_0;\al:\umult{\IMk\rew \VMk}_{\kK} \union  \Delo} $  and 
\item $\Theu \tri \muJu{\Gamu}{u :  \inter_{\kK}\choice{\IMk} }{\Delu}$
\end{itemize}
Notice that  $K \neq \es$ implies $\choice{(  \inter_{\kK}\choice{\IMk})} =  \inter_{\kK}\choice{\IMk} $. 
Moreover, by Lemma~\ref{l:relevance-bis}, since $\al \notin  \fn{\cxoc\cwc{\co{\al'}t u}}$, then $\al \notin \dom{\Del}$, thus $\al \notin \dom{\Del_0}$
and $\al:\umult{\IMk\rew \VMk}_{\kK} \union \Delo =\al:\umult{\IMk\rew \VMk}_{\kK} ; \Delo$. 
We then construct $\Phi$ :
{\small   
$$\infer{
  \Phio\\ \Theu}{
  \muju{\Gam}{ \cxoc\cwc{\co{\al}t} \rempl{\al}{ \al'}{u}:\Any}{(\Delo;\al':\VMo)\!\union \Delu \union \al'\!:\!\union_{\kK}\VMk } }
$$
}

We conclude since $\al' \notin \fn{u}$ implies $\al' \notin \dom{\Delu}$
by Lemma~\ref{l:relevance-bis} so that  $(\Delo;\al':\VMo )\union \Delu\union \al': \union_{\kK} \VMk = \Delo\union \Delu; \al':\VMo\union_{\kK} \VMk=\Dels;\al':\VM=\Del$
  as desired.
\end{itemize}
\end{proof}
}

\section{Strongly Normalizing $\lmuex$-Objects}
\label{s:sn-ex}

In this section we show a characterization of the set of strongly
$\lmuex$-normalizing terms by means of typability. The proof is done
in several steps.  The first key point is the characterization of the
set of strongly $\nonelmuex$-normalizing terms (instead of
\modifrefb{strongly normalizing $\lmuex$-terms}). For that, SR and SE
lemmas for the type system are used.
The second key point is the equivalence between strongly $\nonelmuex$
and $\lmuex$-normalizing terms.  While the inclusion $\SN{\lmuex}
\subseteq \SN{\nonelmuex}$ is straightforward, the fact that every
$\Gc$-reduction step can be \textit{postponed} w.r.t. any
$\nonelmuex$-step (Lemma~\ref{l:postponement}) turns out to be crucial
to show $\SN{\nonelmuex} \subseteq \SN{\lmuex}$.

These technical tools are now used to prove that
 $\SN{\nonelmuex}$ coincides  exactly  with the set of typable
terms.  To close the picture, \ie\ to show that also
$\SN{\lmuex}$ coincides with the set of typable terms, we
establish an equivalence between $\SN{\nonelmuex}$ and
$\SN{\lmuex}$.

As defined in Section~\ref{s:calculus}, for any $\lmuex$-object $o$, we
write now $\mrl{o}$ for the length of the maximal reduction sequence
starting at $o$.  The following equations will play a key role in our proof
of Theorem~\ref{th:typable-sn}.

\begin{lem}
\label{l:equalities-lmus}
\modifrefb{    \[ \begin{array}{llll}
      \mrl{x\,t_1 \ldots t_n} & = & +_{i=1 \ldots n} \mrl{t_i}\\
      \mrl{\l x.t } & = & \mrl{t}\\
  \mrl{\mu \al.\Com} & = &  \mrl{\Com}\\
  \mrl{\co{\al}t } & = & \mrl{t}\\
\mrl{(\l x. u)v \vec{t}  } & = & \mrl{u[x/v] \vec{t} } \\
\mrl{(\mu \al. \Com)v \vec{t}} & = & \mrl{(\mu \al'. \Com \rempl{\al}{ \al'}{v}) \vec{t}} \\
  \mrl{t[x/s]} & = & \mrl{t} + \mrl{s} +1 & \mbox{ if } |t|_x=0 \\
 \mrl{\Com \rempl{\al}{ \al'}{s}} & = &  \mrl{\Com} + \mrl{s} + 1  & \mbox{ if } |\Com|_\al=0\\
\mrl\cxtt \cwc{x}[x/u] {} & = & \mrl{\cxtt \cwc{u}} & \mbox{ if } |\cxtt \cwc{x}|_x = 1\\
\mrl{\cxcc  \cwc{[\al]t}\rempl{\al}{ \al'}{v}} & = & \mrl{\cxcc  \cwc{[\al']tv}} & \mbox{ if }|\cxcc  \cwc{[\al]t}|_\al = 1 \\
\mrl{\cxtt \cwc{x}[x/u]} & = & \mrl{\cxtt \cwc{u}[x/u]} & \mbox{ if } |\cxtt \cwc{x}|_x > 1\\
\mrl{\cxcc \cwc{[\al]t}\rempl{\al}{ \al'}{v}} & = & \mrl{\cxcc \cwc{[\al']tv}\rempl{\al}{ \al'}{v}} & \mbox{ if } |\cxcc \cwc{[\al]t}|_\al > 1\\
\mrl{(tu)[x/s]} & = & \mrl{t[x/s]u } & \mbox{ if } |u|_x =0\\
    \end{array} \]}
\end{lem}

In order to infer $\SN{\nonelmuex} \subseteq \SN{\lmuex}$, the following postponement property is crucial.

\begin{lem}[Postponement]
\label{l:postponement}
Let $o \in \objects{\lmuex}$. 
If $o {\Rewplus{\Gc}\Rew{\nonelmuex}} o'$
then $o {\Rew{\nonelmuex}\Rewplus{\Gc}} o'$.
\end{lem}

\begin{proof} We first show by cases  $o \Rew{\Gc} \Rew{\nonelmuex} o'$
implies $o \Rew{\nonelmuex}  \Rewplus{\Gc} o'$.
Then, the statement holds  by induction on the
number of  $\Gc$-steps from $o$.
\end{proof}

\begin{lem}[From $\nonelmuex$ to $\lmuex$]
\label{l:sn-mk-sn-m}
Let $o \in \objects{\lmuex}$. If $o \in \SN{\nonelmuex}$, then $o \in \SN{\lmuex}$.
\end{lem}

\begin{proof} 
We show that any reduction sequence $\rho: o
\Rew{\lmuex} \ldots$ is finite  by induction on the pair
$\pair{o}{n}$, where $n$ is the maximal integer 
such that  $\rho$ can be decomposed as
$\rho: o \Rew{\Gc}^n o' \Rew{\nonelmuex} o'' \Rew{} \ldots$  (this is well-defined since $\Rew{\Gc}$ is trivially terminating).  We compare the pair $\pair{o}{n}$ using $\Rew{\nonelmuex}$ for the first component (this is well-founded since $o\in \SN{\nonelmuex}$ by hypothesis) and the standard order on natural numbers for the second one. When the reduction sequence starts with at least one $\Gc$-step we conclude by Lemma~\ref{l:postponement}.
All the other cases are straightforward.
\ignore{
\begin{itemize}
\item The base case is $\pair{t}{0}$ where $t$ is in  $\nonelmuex$-nf. 
  Then $\rho$ is empty so that it is finite. 
\item If $\rho$  only contains $\nonelmuex$-steps then it is finite by the hypothesis
      $t \in \SN{\nonelmuex}$.
\item If $\rho$ only contains $\Gc$-steps then it is trivially finite since
 $\Gc$-reduction is trivially a terminating.
\item If $\rho$ starts with at least one $\Gc$-step, \ie\  $t \Rewplus{\Gc}  t'   \Rew{\nonelmuex} u
  \ldots $, then  Lemma~\ref{l:weak-postponement} gives
   a sequence $t   \Rew{\nonelmuex} s \Rewplus{\Gc}  u \ldots$.  The 
   subsequence $\rho': s\Rewplus{\Gc}  u \ldots$ is
   then finite by the \ih\, so that the one starting at $u$ is finite too. We can thus conclude.
\item If $\rho$ starts with at least one $\nonelmuex$, \ie\ 
$t \Rew{\nonelmuex} s \Rew{}\ldots $, then
the subsequence $\rho': s \Rew{} \ldots \Rew{}$ is finite by the \ih\  We can thus conclude.
\end{itemize}
}
\end{proof}

We conclude with the main theorem of this section: 

\begin{thm}
\label{th:typable-sn}
Let $o \in \objects{\lmuex}$. Then $o \in \SN{\lmuex}$ iff $o$ is  typable.
\end{thm}

\begin{proof}
  Let $\tingD{\Phi}{\tyj{o}{\Gam}{\tau} \mid \Del}$. Assume 
  $o \notin \SN{\nonelmuex}$ so that \modifrefb{there exists an infinite} sequence 
$o = o_0 \Rew{\nonelmuex} o_1 \Rew{\nonelmuex} o_2 \Rew{\nonelmuex}\cdots$.
By Lemma~\ref{l:psr} $\tingD{\Phi_i}{\tyj{o_i}{\Gam}{\tau} \mid \Del}$ for every $i$, 
and there exists an infinite sequence $\sz{\Phi_{0}} > \sz{\Phi_{1}} > \sz{\Phi_{2}} > \ldots$, 
which leads to a contradiction because $\sz{\_}$ is a half-integer $\geqslant 1$. Therefore, $o \in \SN{\nonelmuex} \subseteq_{\mbox{Lemma}~\ref{l:sn-mk-sn-m}} \SN{\lmuex}$.

For the converse,  $o \in \SN{\lmuex} \subseteq \SN{\nonelmuex}$ because $\Rew{\nonelmuex} \subseteq \Rew{\lmuex}$. 
We then show that $o \in \SN{\nonelmuex}$ implies $o$ is typable.
For that, we use the equalities in Lemma~\ref{l:equalities-lmus} to reason
by induction on $\mrl{t}$. The cases (1)-(6) and (13) are straightforward
while the cases (7)-(12) use Lemma~\ref{l:pse} (Partial Subject Expansion). 
\end{proof}

It is worth noticing that the proof of Theorem~\ref{th:typable-sn} is
self-contained: we do not use at all the previous characterization of
strongly normalizing objects in the $\lmu$-calculus that we have
developed in Section~\ref{s:sn}. We remark however that an  alternative proof of this theorem can be given
in terms of the projection function defined in Section~\ref{s:operational-ex}, 
an appropriate \modifrefb{preservation of strong normalization}-like property~\cite{Kes09}, and Theorem~\ref{t:final}.


\section{Conclusion}

This paper provides non-idempotent type assignment systems $\Hlmu$  and $\Slmu$ for the $\lmu$-calculus, characterizing, respectively, head and strongly normalizing
terms. These systems feature intersection and union types and can be used to get quantitative
information of $\lmu$-reduction sequences in the following sense:
\begin{itemize}
\item Whenever  $o$ is typable in system $\Hlmu$, then
   its type derivation
   gives a measure providing an  upper bound to the length of  the head-reduction  strategy starting at $o$.
\item The same happens with system $\Slmu$ with respect to the maximal
  length of a reduction sequence starting at $o$.
\item Systems  $\Hlmu$ and $\Slmu$ have suggested the definition
of the calculus $\lmuex$, which implements   a  \modifref{small-step operational semantics} for classical natural deduction that is  an extension of the {\it substitution at a distance paradigm}
to the classical case.
\item  The calculus $\lmuex$ was endowed  with 
an extension of the typing system $\Slmu$ presented for the
$\lmu$-calculus. The resulting system does not only characterize
strong-normalization of small-step reduction but also gives quantitative information about
it.  
\end{itemize}

Following Chapter 3 of~\cite{Krivine93}
(resp.~\cite{BucciarelliKesnerVentura}) in the framework of idempotent
(resp. non-idempotent) intersection types for the $\l$-calculus, it is
also possible to use system $\Hlmu$ to characterize \textit{weak}
normalization of $\lmu$-terms. This can be done by considering a
restricted class of judgments based on positive/negative occurrences
of the empty type $\emul$. This characterization also gives a
certification of the fact that the leftmost-outermost strategy
  is complete for weak normalization in the $\lmu$-calculus.

This work suggests many perspectives in the close future, including: 
\begin{itemize}
\item Quantitative types are  a powerful tool to provide \textit{relational models} for $\lambda$-calculus~\cite{Carvalho07,AEtlca15}. The construction of such models for $\lmu$  should be investigated, particularly to understand in the classical case the collapse relation between quantitative and qualitative models~\cite{Ehrhard12}.
\item  We expect to be able to transfer the ideas in this paper to a  \textit{classical sequent calculus} system, as was already done for focused intuitionistic logic~\cite{KV15}.  \modifref{In particular,  the relational model proposed for the $\bar{\lambda}\mu$-calculus~\cite{Vaux07} could be useful for this purpose}.
\item  The fact that idempotent types were already used to  show \textit{observational equivalence} between call-by-name and call-by-need~\cite{Kesner16} in  intuitionistic logic suggests that  our typing system $\Slmuex$ could be used in the future to provide a type-theoretical view of the  fact that classical call-by-name and classical call-by-need are \textit{not} observationally equivalent~\cite{PedrotSaurin16}. 
\item Moreover, as in~\cite{bernadetleng11},  it should be possible to obtain {\it exact} bounds (and
    not only \textit{upper} bounds) for the lengths of the
    head-reduction and the maximal reduction sequences.
    Although this result remains as future work,
    we remark that the
    difficult and conceptual part of the technique relies on a
    decreasing measure for $\lmu$-reduction, which is precisely one of the
    contributions of this paper.
\item 
The \textit{inhabitation problem} for $\l$-calculus is known to be undecidable for idempotent intersection types~\cite{Urzyczyn99}, but decidable for the non-idempotent ones~\cite{BKRDR14}. We may conjecture that inhabitation is also decidable for $\Hlmu$. 
\end{itemize}


{\bf Acknowledgment:}
We would like to thank Vincent Guisse, who  started  a reflexion on quantitative types for the $\lmu$-calculus \modifref{during his M1 internship in Univ. Paris-Diderot}.

\renewcommand{\em}{\it}
\bibliographystyle{abbrv}
\bibliography{paper}

\section*{Appendix}

\noindent {\bf Lemma~\ref{l:substitution} ({\bf Substitution}).}
Let $\Theu \tri \Gamu  \Vdash u: \IM  \mid \Delu$.
If $\Phi_o \tri \Gam; x:\IM\vdash o:\Any \mid  \Del$, then 
   there is $\Phi_{o\isubs{x/u}}$ such that 
\begin{itemize}
\item  $\Phi_{o\isubs{x/u}}\rhd\Gam \inter \Gamu \vdash o\isubs{x/u}: \Any \mid \Del\union  \Delu $.
\item  $\sz{\Phi_{o\isubs{x/u}}}=\sz{\Phi_o} +   \sz{\Theu} - |\IM|$.
\end{itemize}

  \begin{proof}
We prove  a more general statement, namely:

Let 
$\Theu \tri \Gamu  \Vdash u: \IM  \mid \Delu$.
\begin{itemize}
 \item If $\Phi_o \tri \Gam_o ; x:\IM\vdash o:\Any \mid  \Del_o $, then 
   there is $\Phi_{o\isubs{x/u}}$ such that
   \begin{center}
     $\Phi_{o\isubs{x/u}}\rhd\Gam_o  \inter \Gamu \vdash o\isubs{x/u}: \Any \mid \Del_o \union  \Delu $
   \end{center}
\item If $\Phi_o \tri \Gam_o ; x:\IM\Vdash t:\JM \mid  \Del_o $, then 
  there is $\Phi_{o\isubs{x/u}}$ such that  
\begin{center}
  $\Phi_{o\isubs{x/u}}\rhd\Gam_o  \inter \Gamu \Vdash t\isubs{x/u}: \JM  \mid \Del_o \union  \Delu $
  \end{center}
   \end{itemize}
In both cases  $\sz{\Phi_{o\isubs{x/u}}}=\sz{\Phi_o} +   \sz{\Theu} -|\IM|$.

  We proceed by induction on the structure of $\Phi_o$.
  \begin{itemize}
    \item $(\ax)$: 
      \begin{itemize}
      \item If $o=x$, then $\IM=\mult{\UM}$ is a singleton, $\Any=\UM$, $\Gam_o =\Del_o  =\es$ and   $o\isubs{x/u}=u$. 
      The derivation $\Theu$ is necessarily of the following form
      $$\infer[(\many)]{\Phi'_u \tri \Gam_u  \vdash u:\UM \mid \Del_u}
                       {\Gam_u  \Vdash  u:\mult{\UM} \mid \Del_u}$$
      We then set $\Phi_{x\isubs{x/u}} = \Phi'_u$.
      Then $\sz{\Phi_{x\isubs{x/u}}}=\sz{\Phi_x}+\sz{\Theu}-|\IM|$, since $\sz{\Phi_x}=1= |\IM|$
      and $\sz{\Theu} = \sz{\Phi'_u}$.
      \item If $o=y\neq x$, then $\IM=\emul$ and $o\isubs{x/u}=y$. 
      Moreover,  $\Theu$ is necessarily :
       $$\infer[ (\many)]{ }
              {\es \Vdash u:\emul \mid \es}$$
      We set $\Phi_{y\isubs{x/u}}=\Phi_y$. 
      Then      $\sz{\Phi_{y\isubs{x/u}}}=\sz{\Phi_y} + \sz{\Theu} - |\IM|$ since $| \IM| =0$
      and $\sz{\Theu} = 0$. 
      \end{itemize}  
      \item $(\introarrow)$ : then $o=\lambda x.t$ and 
      the derivation $\Phi_o$ has the following form 
\begin{center}
      $\infer[(\introarrow)]{\Phi_t \rhd \Gam_o; x:\IM;
          y:\JM \vdash t:\UM_t \mid \Del_o} {\Gam_o;
    x:\IM\vdash \lambda y.t: \umult{\JM\ftype  \UM_t} \mid \Del_o}$
  \end{center}
      By the \ih\ we have $\Phi_{t\isubs{x/u}}\rhd (\Gam_o;y:\JM)
       \inter \Gamu \vdash t\isubs{x/u}:\UM \mid \Del_o \vee \Delu$ with
      $\sz{\Phi_{t\isubs{x/u}}}=\sz{\Phi_t}+ \sz{\Theu} - |\IM|$. By $\alpha$-conversion $y \notin \fv{u}$ so that
        $y \notin \dom{\Gamu}$ by Lemma~\ref{l:relevance}, thus  $(\Gam_o;y:\JM) \inter \Gamu =
        (\Gam_o \inter \Gamu); y:\JM $.  We  then set 
       $\Phi_{(\l y.t)\isubs{x/u}}$ equal to 
       \begin{center}
         $\infer[(\introarrow)]{\Phi_{t\isubs{x/u}}}
         {\Gam_o \inter \Gamu \vdash \lambda y.t\isubs{x/u}: \umult{\JM \ftype  \UM_t} \mid \Del_o \vee \Delu}$
       \end{center}

 We have $\sz{\Phi_{(\l y. t)\isubs{x/u}}} = \sz{\Phi_{t\isubs{x/u}}} +1 =_{\ih} 
      \sz{\Phi_t} + \sz{\Theu} - |\IM| + 1 = \sz{\Phi}
      + \sz{\Theu} - |\IM|$.

    \item $(\many)$: then $o$ is a term $t$ and $\Phi_o$ has the following form
      $$\infer[(\many)]{(\Gamk;x:\IMk \vdash t: \UMk \mid \Delk)_{\kK}}
              {\Gam_o=;x:\IM \Vdash t: \mult{\UMk}_{\kK} \mid \Del_o}$$
              where $\IM=\inter_{\kK} \IMk$, $\Gam_o = \inter_{\kK} \Gamk$ and $\Del_o = \union_{\kK} \Delk$. By Lemma~\ref{l:decomposition} there are auxiliary derivations
              $(\tri \Gamu^k \Vdash u: \IMk \mid \Delu^k)_{\kK}$
              such that $\Gamu = \inter_{\kK} \Gamu^k$ and
              $\Delu = \union_{\kK} \Delu^k$. The \ih\ gives
              derivations $(\tri \Gamk \inter \Gamu^k \vdash t\isubs{x/u}: \UMk \mid \Delk \inter \Delu^k)_{\kK}$ and we construct the following auxiliary derivation to conclude
              $$\infer[(\many)]{(\Gamk \inter \Gamu^k \vdash t\isubs{x/u}: \UMk \mid \Delk \inter \Delu^k)_{\kK} }
                      {\inter_{\kK} \Gamk \inter \Gamu^k \Vdash t\isubs{x/u}: \mult{\UMk}_{\kK} \mid \union_{\kK} \Delk \inter \Delu^k }$$
                      We have $\inter_{\kK} \Gamk \inter \Gamu^k  = \Gam_o \inter \Gamu$ and
                      $\union_{\kK} \Delk \inter \Delu^k  = \Del_o \union \Delu$ as desired.
              The size statement trivially holds by the \ih\ 
    \item $(\appet)$:
      then $o=tv$ and the derivation $\Phi_o$ has the following form 
{$$
      \infer[(\appet)]{ \Phit\rhd \Gamt;x: \IMt \vdash t:\umult{\IMk\ftype \VMk}_{\kK} \mid \Delt \hspace{0.7cm}
                   \Phiv \rhd \Gamv;x: \IMv   \Vdash v: \inter_{\kK} \choice{\IMk} \mid \Delv }
            {\Gam_o;x: \IM   \vdash tv:\uVMk \mid \Del_o} $$}
                  where $\Gam_o = \Gamt \inter  \Gamv$, $\Del_o = \Delt \union \Delv$ and $\IM = \IMt \inter \IMv$.

     Moreover,  by Lemma~\ref{l:decomposition}  we can split $\Theu$  in 
$\Thetu \tri \muJu{\Gamtu}{u:\IMt}{\Deltu} $ and $\Thevu\tri \muJu{\Gamvu}{u:\IMv}{\Delvu} $ s.t. $\sz{\Theu}=\sz{\Thetu}+\sz{\Thevu}$.

By the \ih\ there is $\Phi_{t\isubs{x/u}} \rhd \Gamt' \vdash
t\isubs{x/u}:\umult{\IMk\ftype  \VMk}_{\kK} \mid \Delt'$, where
$\Gamt'= \Gamt \inter \Gamtu$ and $\Delt' = \Delt \union \Deltu$ and $\sz{\Phi_{t\isubs{x/u}}}=\sz{\Phi_t} + \sz{\Theu} - |\IMt|$.

Also by the \ih\ there is $\Phi_{v\isubs{x/u}} \rhd \Gamv' \Vdash
v\isubs{x/u}:\inter_{\kK} \choice{\IMk} \mid \Delv'$, where $\Gamv'
= \Gamv \inter  \Gamvu$ and $\Delv' = \Delv
\union \Delvu$ and $\sz{\Phi_{v\isubs{x/u}}}=\sz{\Phiv} +
  \sz{\Thevu} - |\IMv|$.
  
        We set then 
$$\Phi_{o\isubs{x/u}}=
      \infer[(\appet)]{\Phi_{t\isubs{x/u}}   \sep  
             \Phi_{v\isubs{x/u}}}
            {\Gam' \vdash (tv)\isubs{x/u}:\uVMk \mid \Del'} $$
            where $\Gam ' = (\Gamt \inter \Gamtu) \inter (\Gamv \inter
            \Gamvu) = \Gam_o \inter \Gamu$
and $\Del' =  (\Delt \union \Deltu) \union (\Delv \union \Delvu) = \Del_o \union \Delu $ as desired. 
We conclude since 
\[ \begin{array}{l}
\sz{\Phi_{o\isubs{x/u}}} = \sz{\Phi_{t\isubs{x/u}}} + \sz{\Phi^{v\isubs{x/u}}} + |K| \\
=_{\ih}(\sz{\Phit} + \sz{\Thetu} - |  \IMt| ) + 
    (\sz{\Phiv}+ \sz{\Thevu} - |  \IMv| ) + |K| \\
= \sz{\Phi}  + \sz{\Theu} - | \IM |  
\end{array} \] 

    \item All the other cases are straightforward. 
  \end{itemize}
\end{proof}


\noindent {\bf Lemma~\ref{l:replacement} ({\bf Replacement}). } 
Let $\Theu \tri \Gam_u \Vdash u : \inter_{\kK}\ (\choice{\IMk}) \mid \Del_u $ where 
$\al \notin \fn{u}$. If $\Phi_o \tri \tyj{o}{\Gam_o}{\Any \mid \al: \umult{ \IMk \ftype 
      \VMk}_{\kK} ; \Del_o}$, then there is $\Phi_{o\ire{\al}{u}}$ such that :
\begin{itemize}
\item $\Phi_{o\ire{\al}{u}} \tri \tyj{o\ire{\al}{u}}{\Gam_o \inter \Gamu}
                           {\Any \mid \al: \union_{\kK} \VMk; \Del_o \union \Delu }$. 
\item $\sz{\Phi_{o\ire{\al}{u}}} =  \sz{\Phi_o} + \sz{\Theu}$.  
\end{itemize}

\begin{proof}   
We prove  a more general statement, namely:\\
Let $\Theu \tri \Gam_u \Vdash u : \inter_{\kK} \choice{\IMk} \mid \Del_u $ where $\al \notin \fn{u}$.
\begin{itemize}
\item If $\Phi_o \tri \tyj{o}{\Gam_o}{\Any \mid \al: \umult{ \IMk \ftype 
      \VMk}_{\kK} ; \Del_o}$, then there is $\Phi_{o\ire{\al}{u}}$ such that \\ $\Phi_{o\ire{\al}{u}} \tri \tyj{o\ire{\al}{u}}{\Gam_o \inter \Gamu}
                           {\Any \mid \al: \union_{\kK} \VMk; \Del_o \union \Delu }$. 
\item If $\Phi_o \tri \muJu{\Gam_o}{t:\JM}{ \al: \umult{ \IMk \ftype   \VMk}_{\kK} ; \Del_o}$,
      then there is $\Phi_{o\ire{\al}{u}}$ such that \\ $\Phi_{o\ire{\al}{u}} \tri
\muJu{\Gam_o\inter \Gamu}{t\ire{\al}{u}:\JM}{\al:\uVMk;\Del_o\union \Delu}$
\end{itemize}
In both cases, $\sz{\Phi_{o\ire{\al}{u}}} =  \sz{\Phi_o} + \sz{\Theu}$.  \\

We reason by induction on $\Phi_o$. 
Let us call $\UM_\al = \umult{\IMk \ftype  \VMk}_{\kK}$
and $\UM'_\al = \union_{\kK} \VMk$. 
\begin{itemize}
  \item $(\ax)$:  $o=x$, thus we have by construction 
  $$ \tingD{\Phi_o}{\infer[(\ax)]{}{\tyj{x}
                        { x:\mult{\UM} }
                        {\UM \mid \es }}}$$ 
  so that $K= \es$.  Thus, $\inter_{k \in K} \choice{\IMk}  =  \emul$ 
  and $\Gamu = \Delu = \es$, then $\Theu$ is :
 $$\infer[ (\many)]{ }
              {\es \Vdash u:\emul \mid \es}$$
  Thus $\sz{\Theu} = 0$.

We  set $\Phi_{o\ire{\al}{u}} = \Phi_o$ and the first result holds because the derivation 
 has the desired form. We conclude since 
  $\sz{\Phi_{o\ire{\al}{u}}}   =
    \sz{\Phi_o}  + \sz{\Theu}$ as desired. 
  \item $(\introarrow)$: 
 then $o=\l x.t$,  $o\ire{\al}{u}=\l x. (t\ire{\al}{u})$ and 
  by construction we have
  $$ \Phi_{\l x.t}=
           \infer[(\introarrow)]{\tingD{\Phi_{t}}
                         {\tyj{t}
                              {x:\IM;\Gam_o}
                              {\UM \mid \al :\UM_\al ;\Del_o}}}
       {\tyj{\lambda{x}.t}{\Gam_o}
       {\umult{\IM \ftype \UM}\mid \al :\UM_\al;
       \Del_o}}  $$
  By \ih\ it follows that
  $$ \tingD{\Phi_{t\ire{\al}{u}}}
           {\tyj{t\ire{\al}{u}}
                {(x:\IM;\Gam_o) \inter \Gamu }
                {\UM \mid \al :\UM'_\al;\Del_o \union \Delu}} $$
  with $\sz{\Phi_{t\ire{\al}{u}}} = \sz{\Phi_{t}}+\sz{\Theu}$.
  By $\alpha$-conversion we can assume  that $x \notin  \fv{u}$, thus by Lemma~\ref{l:relevance} 
  $x  \notin \dom{\Gamu}$,  so that 
  $(x:\IM;\Gam_o) \inter \Gamu    = x:\IM;\Gam_o \inter \Gamu$.
  
  We thus obtain $\Phi_{\l x.t\ire{\al}{u}}$ of the form:
  $$ 
  \infer[(\introarrow)]{\Phi_{t\ire{\al}{u}}}
       {\tyj{\lambda{x}.t\ire{\al}{u}}{\Gam_o \inter \Gamu}
       {\umult{\IM\ftype  \UM } \mid \al :\UM'_\al;
       \Del_o \union \Delu}} $$
  We conclude since
  \[ \begin{array}{l}
    \sz{\Phi_{\l x.t\ire{\al}{u}}} = \sz{\Phi_{t\ire{\al}{u}}} + 1 
   =_{\ih} \sz{\Phi_{t}}+ \sz{\Theu} +1     =\sz{\Phi_{\l x.t}}+\sz{\Theu}
  \end{array} \]

  \item $(\appet)$: then $o=tv$,  $o\ire{\al}{u}=t\ire{\al}{u}v\ire{\al}{u}$ and 
  by construction we have $\Phi_o= $
{ 
$$ 
       \infer[(\appet)]{\Phit \tri\muju{\Gamt}{t:\UM_t}{\al:\umult{ \IMk \ftype 
      \VMk}_{\kK_t};\Delt} \hspace{0.5cm}  \Phiv \tri \muJu{\Gamv}{v:\JM_v}{\al:\umult{ \IMk \ftype  \VMk}_{\kK_v};\Delv} }{
      \muju{\Gam_o}{o:\UM}{\al:\umult{ \IMk \ftype 
      \VMk}_{\kK}; \Del_o}
}$$    }      
where $\UM_t=\umult{\JMl \ftype \UMl}_{\lL},~ \JM_v=\inter_{\lL} \choice{\JMl},~ \UM=\union_{\lL} \UMl$ (those types are of no matter here, except they satisfy the typing constraint of $\appet$), $\Gam_o=\Gamt\inter \Gamv$, $\Del_o=\Delt \union \Delv,~ K=K_t \uplus K_v$.

     Moreover, by Lemma~\ref{l:decomposition}, we can split $\Theu$ in
$\Thetu \tri \muJu{\Gamtu}{u:\inter_{\kK_t} \choice{\IMk}
}{\Deltu} $ and
$\Thevu \tri \muJu{\Gamvu}{u:\inter_{\kK_v} \choice{\IMk} }{\Delvu} $
s.t. $\sz{\Theu}=\sz{\Thetu}+\sz{\Thevu}$.

  By \ih\ we have
$\Phi_{t\ire{\al}{u}} \rhd \muju{\Gamt \inter \Gamtu}{t\ire{\al}{u}:\UM_t}{\al:\union_{\kK_t} \VMk; \Delt\union \Deltu}$
(since $\al \notin \fn{u}$) with $\sz{\Phi_{t\ire{\al}{u}}}
= \sz{\Phit}+\sz{\Thetu}$.

Also by \ih\  we have
$\Phi_{v\ire{\al}{u}} \rhd \muJu {\Gamv \inter \Gamvu}{v\ire{\al}{u}:\JM_v}{\al:\union_{\kK_v} \VMk,\Delv\union \Delvu}$ with  $\sz{\Phi_{v\ire{\al}{u}}} = \sz{\Phi_v}+\sz{\Thevu}$.

We can now  construct the following derivation            
$$  \infer[(\appet)]{\Phi_{t\ire{\al}{u}}\sep  \Phi_{v\ire{\al}{u}}}{ \muju{\Gam'}{o\ire{\al}{u}:\UM}{\al:\union_{\kK} \VMk;\Del'} }
$$
where $\Gam'=(\Gamt\inter \Gamtu) \inter (\Gamv \inter \Gamvu) = (\Gamt \inter \Gamv) \inter (\Gamtu \inter \Gamvu) = \Gam_o \inter \Gamu$ and likewise, $\Del'=\Del_o \union \Delu$ as desired.
Moreover, 
  \[ \begin{array}{l}
    \sz{\Phi_{(tv)\ire{\al}{u}}}= \sz{\Phi_{t\ire{\al}{u}}}+\sz{\Phi_{v\ire{\al}{u}}} +|L|\\
    =_{\ih} (\sz{\Phit} +\sz{\Thetu})  
     +  (\sz{\Phiv} + \sz{\Thevu})+|L| \\
     = (\sz{\Phit}+\sz{\Phiv}+|L|)+(\sz{\Thetu}+\sz{\Thevu}) = \sz{\Phi_{tv}}   +\sz{\Theu} 
  \end{array} \]

  \item If $o  = \co{\al} t$, then $o\ire{\al}{u}=\co{\al} t\ire{\al}{u}u $ and by construction
  we have a derivation $\Phi_{\co{\al}t}$ of the form:
  $$    \infer[(\muu)]{
         \Phit \tri \muju{\Gam_o}{t:\umult{\IMk \ftype \VMk}_{\kK_t}}{\al : \umult{\IMk \ftype \VMk}_{\kK_\al} ;\Del_o}}{
       \muju{\Gam_o}{\co{\al} t:\TypCom}{\al : \umult{\IMk \ftype \VMk}_{\kK}; \Del_o}}$$
  where $K= K_t \uplus K_\al$.

     Moreover,  by Lemma~\ref{l:decomposition},  we can split $\Theu$  in 
$\Thetu \tri \muJu{\Gamtu}{u:\inter_{\kK_t} \choice{\IMk}}{\Deltu} $ and $\Thealu \tri \muJu{\Gamalu}{u:\inter_{\kK_\al} \choice{\IMk}}{\Delalu} $ s.t. $\sz{\Theu}=\sz{\Thetu}+\sz{\Thealu}$. 

  By the \ih\ we have $\Phi_{t\ire{\al}{u}} \tri \muju{\Gam_o \inter \Gamalu}{t\ire{\al}{u}:\umult{\IMk \ftype  \VMk}_{\kK_t}}{\al:\union_{\kK_\al} \VMk;\Del_o \union \Delalu}$ 
    with  $\sz{\Phi_{t\ire{\al}{u}}} = \sz{\Phit}+\sz{\Thealu}$.

We can then construct the following derivation $\Phi_{\co{\al}t\ire{\al}{u} u}$: 
 {
$$ \infer[(\muu)]{
   \infer[(\appet)]{\Phi_{t\ire{\al}{u}}  \sep  \Thetu}{\muju{\Gam'}{t\ire{\al}{u}u:\union_{\kK_t} \VMk}{\al:\union_{\kK_\al}\VMk;\Del'} } }{\muju{\Gam'}{\co{\al}t\ire{\al}{u}u:\TypCom}{\al:\union_{\kK}\VMk;\Del'} } 
$$
}
with $\Gam'=\Gam_o \inter \Gamalu \inter \Gamtu =\Gam_o \inter \Gamu$ and likewise $\Del'=\Del_o \union \Delu$ (since $\al \notin \fn{u}$) as expected.
  We conclude since 
  \[ \begin{array}{l}
    \sz{\Phi_{\co{\al}t\ire{\al}{u} u}}     = \sz{\Phi_{t\ire{\al}{u} u}} + \ar{\union_{\kK_t}\VMk} \\
     =  \sz{\Phi_{t\ire{\al}{u}}} +\sz{\Thetu} + |K_t| + \ar{\union_{\kK_t}\VMk} \\
     =_{\ih} (\sz{\Phit}+\sz{\Thealu}) 
     + \sz{\Thetu} +  \ar{\umult{\IMk \ftype \VMk}_{\kK_t}}  \\
= \sz{\Phit} + \sz{\Theu}+ \ar{\umult{\IMk \ftype  \VMk}_{\kK_t}} = \sz{\Phi_{\co{\al}t}}+\sz{\Theu}   \\
  \end{array} \] 
  
  \item All the other cases  are straightforward.  \qedhere

\end{itemize}
\end{proof}


\noindent {\bf Property~\ref{l:sr} (Weighted Subject Reduction for $\Slmu$).}
Let $\Phi \tri \tyj{o}{\Gam}{\Any\mid\Del}$. If $o \Rew{} o'$
is a non-erasing step, then
there exists a derivation  $\Phi' \tri \tyj{o'}{\Gam}{\Any\mid\Del}$ 
such that  $\sz{\Phi} > \sz{\Phi'}$. 

\begin{proof}
By induction on the  relation $\Rew{}$. We only show the main cases of 
reduction at the root, the other ones being straightforward. 

\begin{itemize}
\item If $o = (\lambda x.t)u$, then $o' = t\isubs{x/u}$ and $x\in \fv{t}$.
The derivation $\Phi$ has the following form:

$${   \infer[(\appet)]{
     \infer[(\introarrow)]{\Phi_t\rhd \Gam_t ;  x:\IM  \vdash t:\UM~|~ \Del_t}
            {\Gam_t \vdash \lambda x.t: \umult{\IM \ftype \UM}  \mid  \Del_t } \sep 
             {\Theu \rhd  \Gam_u \Vdash u: \choice{\IM } \mid    \Del_u}  }
   {\Gam \vdash o:\UM  \mid  \Del } }
$$  where $\Gam = \Gam_t \inter \Gam_u$, 
       $\Del=\Del_t\union  \Del_u$. Indeed,  $x \in \fv{t}$ implies by
   Lemma~\ref{l:relevance} that $\IM \neq \emul$ so that $\choice{\IM} =\IM = \mult{\UMk}_{\kK}$ 
for some  $K \neq \es$ and some $(\UMk)_{\kK}$.

   Lemma~\ref{l:substitution} yields a
derivation $\Phi'_{t\isubs{x/u}} \rhd \Gam_t \inter  \Gam_u  \vdash
t\isubs{x/u}: \UM \mid \Del_t \union \Del_u$ with
$\sz{\Phi'_{t\isubs{x/u}}}=\sz{\Phi_t}+  \sz{\Theu}-|K|$ ($| \IM |  = |K|$). We set $\Phi'
= \Phi'_{t\isubs{x/u}}$ so that $\sz{\Phi}
= \sz{\Phi_t} +1+ \sz{\Theu}+1 > \sz{\Phi'}$. \\


\item
If $o = (\mu \al. \Com)u$, then $o' = \mu \al . \Com\ire{\al}{u}$
and $\al \in \fn{\Com}$.

The derivation  $\Phi$ has the following form: 
$$
  \infer[(\appet)]
        {\infer[(\mud)]
               {\tingD{\PhiCom}
                      {\tyj{\Com}
                           {\GamCom}
                           {\TypCom\mid\al: \VMC}; \DelCom}}
                      {\tyj{\mu\al.\Com}
                           {\GamCom}
                           {\choice{\VMC} \mid \DelCom}} \quad 
                 \tingD{\Theu }
                       {\Gamu \Vdash u: \choice{\IMu} \mid \Delu}}
         {\tyj{(\mu\al.\Com)u}{\GamCom \inter \Gamu}
                              {\UM \mid \DelCom \union \Delu}}$$
where $\choice{\VMC} = \VMC= \umult{\IMk \ftype \VMk}_{\kK}$, 
$\choice{\IMu} = \IMu= \inter_{\kK}   \choice{\IMk}$, 
$\UM= \vee_{\kK} \VMk$,
$\Gam = \GamCom \inter \Gamu$ and $\Del = \DelCom \union \Delu$.
 Indeed, the hypothesis $\al \in \fn{\Com}$ implies $K \neq \es$ by
Lemma~\ref{l:relevance}, and thus $\choice{\VMC} = \VMC$ and $\choice{\IMu} = \IMu$. 
Lemma~\ref{l:replacement} then gives the derivation
$\tingD{\Phi_{\Com\ire{\al}{u}}}
                        {\tyj{\Com\ire{\al}{u}}
                             {\GamCom \inter  \Gamu}
                             {\TypCom\mid\al:\union_{\kK} \VMk ;\DelCom  \union  \Delu }}$.
We can then construct the following derivation $\Phi'$:
    $$ \infer[(\mud)]{\Phi_{\Com\ire{\al}{u}}
                        }
                 {\tyj{\mu\al.\Com\ire{\al}{u}}
                      {\GamCom \inter \Gamu }
                      {\union_{\kK} \VMk 
                             \mid \DelCom  \union \Delu  }}
    $$
    We conclude since 
    \[ \begin{array}{l} 
      \sz{\Phi'} = \sz{\phi_{\Com\ire{\al}{u}}} + 1  =_{Lemma\ref{l:replacement}} 
      \sz{\Phi_\Com} + \sz{\Theu}  + 1  <                 \\
    \sz{\Phi_{\Com}} + 1  + \sz{\Theu}+ |K| = \sz{\Phi_{\mu \al. \Com}} + \sz{\Theu}+ |K| = \sz{\Phi}
    \end{array} \] 

The step $<$ is justified by $K \neq \es$. \qedhere

\end{itemize}
\end{proof}

The reader should notice that the fact that the choice operator
produces a \textit{blind type} for union types is not used in the
proof of Property~\ref{l:sr}. Indeed, by Lemma\;\ref{l:relevance}, the
variable (resp. name) of a $\beta$-redex (resp. $\mu$-redex) has an
empty intersection (resp. union) type in system $\Slmu$ only when this
redex is erasing, a case that is not in the scope of
Property\;\ref{l:sr}. However, note that blind types are involved in
the proof of the subject reduction property in system $\Hlmu$,
which can easily be adapted from that of Property~\ref{l:sr}.\\


\noindent {\bf Lemma~\ref{l:reverse-substitution} (Reverse Substitution).}
Let $\Phi'\tri \muju{\Gam'}{ o\subxu:\A}{\Del'}$ Then there exist $\Gam_o, \Del_o, \IM, \Gamu, \Delu$
such that:
\begin{itemize}
     \item $\Gam' = \Gam \inter \Gamu$, 
     \item $\Del' = \Del \union \Delu$, 
     \item $\tri \muju{\Gam;x:\IM}{o:\A}{\Del}$
     \item $\tri \muJu{\Gamu}{u: \IM}{\Delu}$.
   \end{itemize}

  \begin{proof}
We prove a more general statement, namely:
\begin{itemize}
\item If $\Phi'\tri \muju{\Gam'}{ o\subxu:\A}{\Del'}$, then $\tri \muju{\Gam_o;x:\IM}{o:\A}{\Del_o}$,
$\tri \muJu{\Gamu}{u: \IM}{\Delu}$, where $\Gam'=\Gam_o \inter \Gamu,~ \Del'=\Del_o\union \Delu$ for some $\IM,~ \Gam_o,~ \Gamu,~ \Del_o,~\Delu$.
\item If $\Phi'\tri \muJu{\Gam'}{t\subxu:\JM}{\Del'}$, then $\tri \muJu{\Gam_o;x:\IM}{t:\JM}{\Del_o}$,
$\tri \muJu{\Gamu}{u: \IM}{\Delu}$, where $\Gam'=\Gam_o \inter \Gamu,~ \Del'=\Del_o\union \Delu$ for some $\IM,~ \Gam_o,~ \Gamu,~ \Del_o,~\Delu$.
\end{itemize}

We proceed by induction on the structure of $\Phi'$. 
\begin{itemize}

\item $(\ax)$

\begin{itemize}

\item If $o = y \neq x$, then $y\isubs{x/u} = y$. By
  construction  one has that $\Gam' =
  y:\mult{\UM}$ and $\Any = \UM$. The
  result thus holds for $\IM=\emul$, $\Gam_o =\Gam',~ \Del_o=\Del'$,
  $\Gamu = \es$ and $\Delu = \es$ as $\es  \Vdash u: \emul \mid \es$ is derivable by the $(\many)$ rule.

\item  If $o = x$, then $x\isubs{x/u} = u$. By
  construction  one has that $\Any = \UM$. We type $x$ with the axiom rule: 
  $$\infer[(\ax)]{\es}{ \tyj{x}{x:\mult{\UM}}{\UM \mid \es }}$$
  so that the property holds for 
  $\Gam_o = \Del_o =\es$, $\IM = \mult{\UM}$, $\Gamu = \Gam'$, 
   $\Delu = \Del'$, where 
   $\tri \muJu{\Gamu}{u:\IM}{\Delu}$ is obtained by the rule 
   $(\many)$ from $\muju{\Gam'}{u:\UM}{\Del'}$.

\end{itemize}

\item $(\introarrow)$   $o = \l y.t$ and $(\l {y}.t)\isubs{x/u} =
  \l {y}.t\isubs{x/u}$. Then $\Phi'$ is of the form
$$\infer[(\introarrow)]{\Phit'\tri \muju{\Gam';y:\JM}{t\subxu:\VM}{\Del'} }
        {\muju{\Gam'}{\l y.t\subxu:\umult{\JM\ftype \VM}}{\Del'}}$$
  where $\UM = \umult{\JM \ftype \VM} $. 

By the \ih\ $\Gam'; y: \IM = \Gamt \inter \Gamu$ and $\Del' = \Del_t
\union \Delu$, $\tri \muju{\Gamt;x:\IM}{t:\VM}{\Del_t}$ and
$\tri \muJu{\Gamu}{u:\IM}{\Delu}$.  By $\alpha$-conversion we can
assume that $y \notin \fv{u}$, so that $y \notin \dom{\Gamu}$ by Lemma~\ref{l:relevance} and thus
$\Gamt =\Gam'_t;y:\JM$ and $\Gam'= \Gam'_t \inter \Gamu$. Hence, we obtain
$\tyj{\l {y}.t}{\Gam'_t;x:\IM}{\UM \mid \Del_t}$ by the
rule $(\introarrow)$. We conclude by setting $\Gam_o = Gam'_t$ and $\Del_o = \Del_t$.

\item   $(\appet)$   $o   =   t v$   and   $(t v)\isubs{x/u}   =
  t\isubs{x/u} v\isubs{x/u}$.                                       By
  construction  we have  that  $\Gam'  = \Gamt' \inter \Gamv'$ and $\Del' = \Delt' \union \Delv'$ and $\tri \muju{\Gamt'}{t\subxu:\UM_t}{\Delt'},~ \tri \muJu{\Gamv'}{v: \JM_v}{\Delv'}$ with $\UM_t = \umult{\JMk\ftype \UMk}_{\kK},~ \JM_v = \inter_{\kK} \choice{\JMk}$  (those types are of no matter here, except they satisfy the typing constraint of $(\appet)$).
  By the \ih\ there are:
  \begin{itemize}
  \item $\Gamt,~ \IMt,~\Delt,~ \Gamtu,~ \Deltu$ s.t. $\Gamt'=\Gamt\inter \Gamtu,~ \Delt'=\Delt \union \Deltu,~ \tri \muju{\Gamt;x:\IMt}{t:\UM_t}{\Delt}$ and $\tri \muJu{\Gamtu}{u:\IMt}{\Deltu}$.
  \item $\Gamv,~ \IMv,~\Delv,~ \Gamvu,~ \Delvu$ s.t. $\Gamv'=\Gamv\inter \Gamvu,~ \Delv'=\Delv \union \Delvu,~ \tri \muJu{\Gamv;x:\IMv}{v:\JM_v}{\Delv}$ and $\tri \muJu{\Gamvu}{u:\IMv}{\Delvu}$.
  \end{itemize}
  Thus, we can type $tv$ with :
  $$
  \infer[(\appet)]{\tri \muju{\Gamt;x:\IMt}{t:\UM_t}{\Delt}\\ \tri \muJu{\Gamv;x:\IMv}{v:\JM_v}{\Delv}}{\muju{\Gam_o;x:\IM}{tv:\UM}{\Del_o}}
  $$
  where $\Gam_o = \Gamt\inter \Gamu,~ \Del_o = \Delt\union \Delu,~ \IM=\IMt\union \IMv$.

  We obtain $\tri \muJu{\Gamu}{u:\IM}{\Delu}$ with $\Gamu=\Gamtu\inter \Gamvu,~ \Delu=\Deltu \union \Delvu$ by Lemma~\ref{l:decomposition}.
  
\item The other cases are similar.  \qedhere
\end{itemize}

\end{proof}

\noindent {\bf Lemma~\ref{l:reverse-replacement} (Reverse Replacement). }
Let $\Phi'\rhd \muju{\Gam'}{o\alu:\A}{\al:\VM;\Del'} $, where
$\al \notin \fn{u}$. 
Then there exist $\Gam_o, \Del_o, \Gamu,  \Delu, (\IMk)_{\kK}, (\VMk)_{\kK}$
       such that:
\begin{itemize}
     \item $\Gam' = \Gam \inter \Gamu$, 
     \item $\Del' =\Del \union \Delu$, 
     \item $\VM = \uVMk$,
     \item $\tri \muju{\Gam}{o:\A}{\al:\uIVMk;\Del}$, and
     \item $\tri \muJu{\Gamu}{u:\iIMsk}{\Delu}$
   \end{itemize}

  \begin{proof}
    We prove a more general statement, namely :
    \begin{itemize}
    \item If $\Phi'\rhd \muju{\Gam'}{o\alu:\A}{\al:\VM;\Del'} $, then \[\tri \muju{\Gam_o}{o:\A}{\al:\uIVMk;\Del_o},~ \tri \muJu{\Gamu}{u:\iIMsk}{\Delu}\] where $\Gam'=\Gam_o \inter \Gamu,~ \Del'=\Del_o \union \Delu,~ \VM=\uVMk$ for some $\Gam_o,~\Gamu,~ \Del_o,~ \Delu, (\VMk)_{\kK},~ (\IMk)_{\kK}$.
    \item If $\Phi'\rhd \muJu{\Gam'}{t\alu:\JM}{\al:\VM;\Del'} $, then \[\tri \muJu{\Gam_o}{t:\JM}{\al:\uIVMk;\Del_o},~ \tri \muJu{\Gamu}{u:\iIMsk}{\Delu}\] where $\Gam'=\Gam_o \inter \Gamu,~ \Del'=\Del_o \union \Delu,~ \VM=\uVMk$ for some $\Gam_o,~\Gamu,~ \Del_o,~ \Delu, (\VMk)_{\kK},~ (\IMk)_{\kK}$.
      \end{itemize}

    We proceed by induction on the structure of $\Phi$'.
\begin{itemize}
\item $(\ax)$ $o = x$ and $o\ire{\al}{u} = x$. Then $\Phi'$ is of the form
  $x: \mult{\UM} \vdash x: \UM \mid \es $ and we have
  $\VM=\umult{\phantom{.}}$ so that we set $\Gam_o = x: \mult{\UM}$,
  $\Gamu=\Delu=\Del_o = \es$, $K=\es$.  Notice that $\es  \Vdash
  u: \emul \mid \es$ always holds. \\
  
\item $(\introarrow  ) $  $o =  \l y.t$ and $(\l y.t)\ire{\al}{u} = \l y.t\ire{\al}{u}$. 
Then  $\Phi'$ is of the form
$$
\infer[(\introarrow)]{\muju{\Gam'; y: \JM}{ t\alu: \UM_t}{\al:\VM;\Del'}}
      {\muju{\Gam'}{ \l y.t\alu: \umult{\IM \ftype \UM_t}}{\al:\VM;\Del'}}
$$

      The \ih\ gives $\Gam';y:\JM = \Gamt \inter \Gamu,~ \VM= \uVMk,~ \Del'=\Del_t \union \Delu,~ \tri \muju{\Gamt}{t:\UM_t}{\al:\uIVMk;\Del_t}$ and $\tri \muJu{\Gamu}{u:\iIMsk}{\Delu}$. 
      By $\alpha$-conversion we can 
assume that $y \notin \fv{u}$, so that $y \notin \dom{\Gamu}$ holds by 
Lemma~\ref{l:relevance} and thus $\Gamt = \Gam'_t; y:\JM$. Hence, we obtain 
$$\infer{ \tri \muju{\Gam'_t;y:\JM}{t:\UM_t}{\al:\uIVMk;\Del_t}}{
 \tri \muju{\Gam'_t}{\l y.t:\umult{\JM\ftype \UM_t}}{\al:\uIVMk;\Del_t}}$$
From that, the desired conclusion is straightforward by setting
$\Gam_o = \Gam'_t$ and $\Del_o = \Del_t$.

\item $o = \co{\al} t$ and $o \ire{\al}{u} = \co{\al} t\ire{\al}{u} u$. Then $\Phi'$ has the
following form 
\[
\infer[(\mud)]{\infer[(\appet)]{\muju{\Gamt'}{t\alu:\uIVMkt }{\al:\VM_\al;\Delt'} \hspace{0.7cm}
              \muJu{\Gamtu}{u:\inter_{\kK_t} \choice{\IMk} }{\al:\VM_u; \Deltu}}
             {\muju{\Gamt'\inter \Gamtu}{t\alu u: \uVMkt}{\al:\VM_\al; \Del'}}}
      {\muju{\Gam'}{\co{\al}t\alu u: \TypCom}{\al:\uVMkt \union \VM_\al;\Del' }} 
\]
      where $\Gam'=\Gamt' \inter \Gamtu,~ \Del' = \Delt \vee \Deltu$ and $\VM = \uVMkt\union \VM_\al \union \VM_{u}$. Moreover, the hypothesis $\al \notin \fn{u}$
        implies $\VM_u = \umult{\, }$ by Lemma~\ref{l:relevance}. 

      The \ih\ gives $\tri \muju{\Gam_t}{t:\uIVMkt}{\al:\uIVMkal;\Del_t},~ \tri \muJu{\Gamalu}{u:\inter_{\kK_\al}\choice{\IMk}}{\Delalu}$ where $\Gamt'=\Gam_t\inter \Gamalu,~ \Delt' = \Del_t\union \Delalu$, and $\VM_\al = \uVMkal$.
W.l.o.g we can assume $K_\al \cap K_t = \es$.
            We then set $K=K_{\al} \uplus K_t$ and we  define :
      $$ \infer[(\mud)]{\tri \muju{\Gam_t}{t:\uIVMkt}{\al:\uIVMkal;\Del_t}}
      {\muju{\Gam_t}{\co{\al}t:\TypCom}{\al:\uIVMk;\Del_t}}$$

      By Lemma~\ref{l:decomposition}, we also have
      $\tri\muJu{\Gamu}{u:\inter_{\kK} \choice{\IMk}}{\Delu}$ with $\Gamu=\Gamtu \inter \Gamalu,~ \Delu=\Deltu \union \Delalu$. We
      can then conclude 
by setting $\Gam_o = \Gam_t$ and $\Del_o = \Del_t$
since  $\Gam_t \inter \Gamu = \Gam_t \inter (\Gamalu \inter \Gamtu)=\Gamt'\inter \Gamtu = \Gam'$ and 
likewise $\Del_t\union \Delu=\Del'$.\\

\item $o = tv$ so that $o \ire{\al}{u} =  t\ire{\al}{u} v \ire{\al}{u}$. Then $\Phi$
has the following form: 
$$\infer[(\appet)]{\tri \muju{\Gamt'}{t\alu: \UM_t}{\al:\VM_t;\Delt'} \hspace{0.7cm} 
         \tri \muJu{\Gamv'}{v\alu : \JM_v}{\al:\VM_v; \Delv'} }
        {\muju{\Gam'}{t\alu v\alu: \UM}{\al:\VM;\Del'}}$$
        where  $\VM=\VM_t\union \VM_\al$, $\Gam'  = \Gamt' \inter \Gamv'$,
        $\Del' = \Delt' \union \Delv'$, $\UM_t = \umult{\JMk\ftype \UMk}_{\kK}$
        and $\JM_v = \inter_{\kK} \choice{\JMk}$  (those types are of no matter here, except they satisfy the typing constraint of $(\appet)$).

         The property then trivially holds by the \ih\ (we proceed as in the complete proof of Lemma~\ref{l:reverse-substitution}, case $(\appet)$).\\

\item The other cases are similar.  \qedhere
\end{itemize}

\end{proof}


\noindent {\bf Property~\ref{l:se-s} (Subject Expansion for $\Slmu$).}
Assume $\Phi' \rhd \muju{\Gam'}{o':\A}{\Del'}$. If $o\Rew{} o'$
is a non-erasing step, then there is $\Phi \rhd
\muju{\Gam'}{o:\A}{\Del'}$.

\begin{proof} By induction on the reduction relation. 
We only show the main cases of reduction at the root,
the other ones being straightforward by induction. 
We can then assume $\Any = \UM$ for some union type $\UM$.

\begin{itemize}
\item If $o = (\l x.t)u$, then $o' = t\isubs{x/u}$ with $x\in \fv{t}$.
The Reverse Substitution Lemma~\ref{l:reverse-substitution} yields 
\begin{itemize}
\item $\Gam' = \Gam_o \inter \Gamu$,
\item $\Del' = \Del_o \inter \Delu$,
\item $\tri \muju{\Gam_o; x:\IM}{t:\UM}{\Del_o}$, and 
\item $\tri \muJu{\Gamu}{ u: \IM}{\Delu}$.
\end{itemize}

Moreover, $x \in \fv{t}$ implies by Lemma~\ref{l:relevance} that $\IM\neq \emul$, so that $\IM^*=\IM$.
We can then set :
$$ \Phi =
\infer[(\appet)]{\infer[(\introarrow)]{ \tri \muju{\Gam_o; x:\IM}{t:\UM}{\Del_o}}
              {\muju{\Gam_o}{\l x.t: \umult{\IM \ftype \UM}}{\Del_o}} \\
              \tri \muJu{\Gamu}{u: \IM}{\Delu}}
{\muju{\Gam'}{(\lambda x.t)u:\UM}{\Del'}
}$$

\item If $o = (\mu \al. \Com)u$, then $o' = \mu \al . \Com\ire{\al}{u}$ with $\al \in \fn{\Com}$. Moreover, $\al \in \fn{\Com \alu}$ and $\Phi'$ has the following form: 
  
$$\infer[(\mud)]{ \muju{\Gam'}{\Com\alu: \TypCom}{ \al: \UM ; \Del'}}
        {\muju{\Gam'}{\mu \al. \Com\ire{\al}{u}:  \UM}{\Del'} }$$
where $\UM \neq \umult{\,}$ holds by Lemma~\ref{l:relevance}, since $\al \in\fn{\Com\ire{\al}{u}}$, so that the $\mud$ rule is correctly applied.
Then the  Reverse Replacement Lemma~\ref{l:reverse-replacement} yields:
 \begin{itemize}
     \item $\Gam' = \GamCom \inter \Gamu$, 
     \item $\Del'=\DelCom \union \Delu$,
     \item $\UM=\uVMk$,
     \item $\tri \muju{\GamCom}{\Com:\TypCom}{\al:\uIVMk;\DelCom}$, and 
     \item $\tri \muJu{\Gamu}{u: \iIMsk}{\Delu}$.
   \end{itemize}
 Moreover, $\UM \neq \umult{\,}$ implies $K \neq \es$, thus
   $\choice{\uIVMk}=\uIVMk$ and 
 we conclude by  constructing the following derivation:  
$$\infer[(\appet)]{\infer[(\mud)]{\tri \muju{\GamCom}{\Com:\TypCom}{\al:\VMC;\DelCom}}{\muJu{\GamCom}{\mu \al.\Com:\VMC}{\DelCom}} \\ 
  \tri \muJu{\Gamu}{u: \IMu}{\Delu}  }{\muJu{\Gam'}{(\mu \al.\Com)u:\UM}{\Del'}}$$
where $\VMC=\uIVMk,\ \IMu=\iIMsk$ \qedhere
\end{itemize}
\end{proof}


Weighted Subject reduction for the $\lmuex$-calculus
(Lemma~\ref{l:psr}) is based on the fact that linear
substitution (Lemma~\ref{l:partial-substitution}) and linear replacement
 (Lemma~\ref{l:partial-replacement}) preserve types.

\begin{lem}[\textbf{Linear  Substitution}]
\label{l:partial-substitution}
Let $\Theu \tri \Gamu  \Vdash u: \IM \mid \Delu$. If 
$\Phi_{\cxot\cwc{x}} \tri \Gam; x: \IM \vdash \cxot\cwc{x}: \Any \mid
\Del$, then there exist $\IM_1, \IM_2, \Gamuu, \Gamud, \Deluu, \Delud$  s.t. 
\begin{itemize}
\item $\IM = \IM_1 \inter \IM_2$, where $\IM_1 \neq \emul$, 
\item $\Gamu = \Gamuu \inter \Gamud$ and $\Delu = \Deluu \union \Delud$, 
\item $\Theuu \tri \Gamuu \Vdash u: \IM_1 \mid \Deluu$,
\item $\Theud \tri \Gamud \Vdash u: \IM_2 \mid \Delud$,
\item $\Phi_{\ctx\cwc{u}} \tri \Gam \inter \Gamuu; x: \IM_2 \vdash
\cxot\cwc{u}: \Any \mid \Del \union \Deluu$, and 
\item $\sz{\Phi_{\cxot\cwc{u}}} = \sz{\Phi_{\cxot\cwc{x}}} + \sz{\Theuu}  - |  \IM_1|$.
\end{itemize}
\end{lem}

  \begin{proof}
    The proof is by induction on the context $\cxot$
so we need to prove the statement of the lemma 
for regular derivations simultaneously with the following one
for  \textit{non-empty} auxiliary derivations:
if $\Phi_{\cxtt\cwc{x}} \tri \Gam; x: \IM \Vdash \cxtt\cwc{x}: \JM \mid
\Del$ and $\JM \neq \emul$, then there exist $\IM_1, \IM_2, \Gamuu, \Gamud, \Deluu, \Delud$  s.t. 
\begin{itemize}
\item $\IM = \IM_1 \inter \IM_2$, where $\IM_1 \neq \emul$, 
\item $\Gamu = \Gamuu \inter \Gamud$ and $\Delu = \Deluu \union \Delud$, 
\item $\Theuu \tri \Gamuu \Vdash u: \IM_1 \mid \Deluu$,
\item $\Theud \tri \Gamud \Vdash u: \IM_2 \mid \Delud$,
\item $\Phi_{\cxtt\cwc{u}} \tri \Gam \inter \Gamuu; x: \IM_2 \Vdash
\cxtt\cwc{u}: \JM \mid \Del \union \Deluu$, and 
\item $\sz{\Phi_{\cxtt\cwc{u}}} = \sz{\Phi_{\cxtt \cwc{x}}} + \sz{\Theuu}  - | \IM_1|$.\\
\end{itemize}

Now, we can start the proof. 
Notice that 
$\IM\neq\emul$ by Lemma~\ref{l:relevance-bis}, since 
$x \in \fv{\cxot\cwc{x}} $ (resp. $x \in \fv{\cxtt\cwc{x}})$.
We only show the case $\cxot = \Box$ since all the other ones are straightforward. 
So assume $\cxot=\Box$. Then $\IM=\mult{\UM}$ for some $\UM$
and  the derivation $\Phi_{x}$ has the following form :
\begin{center}
$\Phi_x = \infer[(\ax)]{ }
  {\muju{x:\mult{\UM}}{x:\UM}{\es} }$
  \end{center}

   Thus, $\sz {\Phi_{x }} = 1$. We set then $\Theuu=\Theu$ and $\Theud=\infer[(\many)]{}{\muJu{\phdot }{u:\emul}{\phdot}}$\\
   We have 
$\sz{\Phi_{u}} =    \sz{\Theuu} = \sz{\Phi_x} + \sz{\Theu} - 
| \IM_1|  $
since $|\IM_1|=1$.
  \end{proof}
\medskip

\begin{lem}[\textbf{Linear Replacement}]
\label{l:partial-replacement}
Let $\Theu \tri \muJu{\Gamu}{u:\inter_{\lL} \choice{\IMl}}{\Delu}$ s.t. $\al \notin \fv{u}$. 
If $\Phi_{\cxoc{\cwc{\coal t}}} \tri \muju{\Gam}{\cxoc\cwc{\co{\al}t}:\Any}{\al: \umult{ \IMl \rew \VMl}_{\lL};\Del}$, 
then there exist $L_1, L_2, \Gamuu, \Gamud,  \Deluu,  \Delud, \Phi_{\cxoc\cwc{\co{\al'} tu}} $  s.t.
\begin{itemize}
\item $L=L_1\uplus L_2$, where $L_1\neq \es$.
\item $\Gamu = \Gamuu \inter \Gamud$ and $\Delu = \Deluu \union \Delud$, 
\item $\Theuu \tri \muJu{\Gamuu}{u: \inter_{\lL_1} \choice{\IMl}}{\Deluu}$,
\item $\Theud \tri \muJu{\Gamud}{u: \inter_{\lL_2} \choice{\IMl} }{\Delud}$,
\item $\Phi_{\cxoc\cwc{\co{\al'} tu}} \tri \muju{\Gam \inter \Gamuu}{\cxoc{\cwc{\co{\al'} tu}}:\Any}{\al: \umult{ \IMl \rew \VMl}_{\lL_2}; \al':\union_{\lL_1} \VMl \union \Del \union \Deluu}$, and 
\item $\sz{\Phi_{\cxoc\cwc{\co{\al'} tu}}} = \sz{\Phi_{\cxoc\cwc{\coal t}}} + \sz{\Theuu}$.
\end{itemize}
\end{lem}

\begin{proof}
The proof is by induction on the context $\cxoc$ 
so we need to prove the statement of the lemma 
for regular derivations simultaneously with the following one
for  \textit{non-empty} auxiliary derivations:
 if $\Phi_{\cxtc{\cwc{\coal t}}} \tri \muJu{\Gam}{\cxtc\cwc{\co{\al}t}:\JM}{\al: \umult{ \IMl \rew \VMl}_{\lL};\Del}$
and $\JM \neq \emul$, then there exist $ L_1, L_2,  \Gamuu,  \Gamud,  \Deluu,  \Delud, 
\Phi_{\cxtc\cwc{\co{al'} tu}} $  s.t.
\begin{itemize}
\item $L=L_1\uplus L_2$, where $L_1\neq \es$.
\item $\Gamu = \Gamuu \inter \Gamud$ and $\Delu = \Deluu \union \Delud$, 
\item $\Theuu \tri \muJu{\Gamuu}{u: \inter_{\lL_1} \choice{\IMl}}{\Deluu}$,
\item $\Theud \tri \muJu{\Gamud}{u: \inter_{\lL_2} \choice{\IMl} }{\Delud}$,
\item $\Phi_{\cxtc\cwc{\co{\al'} tu}} \tri \muJu{\Gam \inter \Gamuu}{\cxtc{\cwc{\co{\al'} tu}}:\JM}{\al: \umult{ \IMl \rew \VMl}_{\lL_2}; \al':\union_{\lL_1} \VMl \union \Del \union \Deluu}$, and 
\item $\sz{\Phi_{\cxtc\cwc{\co{\al'} tu}}} = \sz{\Phi_{\cxtc\cwc{\coal t}}} + \sz{\Theuu}$.\\
\end{itemize}

Now, we can start the proof.
Notice that $L \neq \es$ by Lemma~\ref{l:relevance-bis}, since $\alpha \in \fn{\cxoc\cwc{\co{\al}t}} $ (resp. $\alpha \in \fn{\cxtc\cwc{\co{\al}t}}$). 
We only show the case $\cxoc = \boxdot$ since all the other ones are straightforward.

So assume $\cxoc=\boxdot$. Then the derivation $\Phi_{\co{\al}t}$ has the
following form, where $K \neq \es$ holds by Lemma~\ref{l:non-empty-union}: 
\begin{center}
$
\infer[(\muu)]{\Phi_t \tri \tyj{t}{\Gam}{ \umult{ \IMk \rew \VMk}_{\kK} \mid \al: \umult{ \IMl \rew \VMl}_{\lL \sm K} ; \Del}}
   {\tyj{\co{\al}t}{\Gam}{\TypCom \mid \al: \umult{ \IMl \rew \VMl}_{\lL} ; \Del}}$
   \end{center}

Thus, $\sz {\Phi_{\co{\al}t }} =\sz{\Phi_t} + \ar{\umult{ \IMk \rew \VMk}_{\kK} } = \sz{\Phi_t} + |K| + \ar{\union_{\kK} \VMk}$.
We set $L_1=K$ and $L_2=L\setminus K$ and we write
$\inter_{\lL} \choice{\IMl}$ as $(\inter_{\lL_1} \choice{\IMl})\inter (\inter_{\lL_2} \choice{\IMl})$. Then by Lemma~\ref{l:decomposition} 
 there are $\Theuu \tri \muJu{\Gamuu}{u:\inter_{\lL_1} \choice{\IMl}}{\Deluu},~ \Theud \tri \muJu{\Gamud}{u:\inter_{\lL_2} \choice{\IMl}}{\Delud}$ s.t. $\Gamuu\inter \Gamud =\Gamu,~ \Deluu \union \Delud =\Delu$. 
We set $\VM=\union_{\lL_1} \VMl$
and then construct the following derivation $\Phi_{\co{\al'}tu}$: 
\begin{center}$ \infer[(\muu)]{\infer[(\appet)]{\Phi_t \hspace{1.5cm} \Theuu}
              { \muju{\Gam \inter \Gamuu}{ t\,u: \union_{\lL_1} \VMl}{\al: \umult{ \IMl \rew \VMl}_{\lL_2} ; \Del \union  \Deluu}}}
       { \muju{\Gam \inter \Gamuu}{\co{\al'}tu: \TypCom}{\al: \umult{ \IMl \rew \VMl}_{\lL_2}; \al': \VM \vee  \Del\union  \Deluu}}
$       \end{center}

We have: 
\begin{align*}
   &\sz{\Phi_{\co{\al'}tu}} =    \sz{\Phi_{tu}} + \ar{\union_{\kK} \VMk} \\
   &= \sz{\Phi_{t}} + \sz{\Theuu} + |K| + \ar{\union_{\kK} \VMk}   \\
   &= \sz{\Phi_{t}} + \sz{\Theuu} +  \ar{ \umult{ \IMk \rew \VMk}_{\kK} }   \\
   &=  \sz{\Phi_{\co{\al}t}} + \sz{\Theuu}  & \tag*{\qedhere}
   \end{align*}
\end{proof}

\noindent {\bf Property~\ref{l:psr} ({Weighted Subject Reduction for $\lmuex$}).}
Let $\Phi \tri \tyj{o}{\Gam}{\Any\mid\Del}$. If $o \Rew{} o'$ is a non-erasing step, then $\Phi' \tri \tyj{o'}{\Gam}{\Any\mid\Del}$ and  $\sz{\Phi} > \sz{\Phi'}$. \\

\begin{proof}
By induction on the reduction relation $\rew$. We only show the main cases of reduction at the root, the other ones being straightforward.

\begin{itemize}
\item If $o = \slist[(\lambda x.t)] u \Rew{} \slist[t[x/u]]= o'$: we proceed by induction on $\slist$,
by detailing only the case $\slist = \Box$ as the other one is straightforward. 

The derivation $\Phi$ has the following form:
$$\Phi = \infer[(\appet)]{ \infer[(\introarrow)]{\Phi_t\rhd \Gam_t ; x:\IM \vdash t:\UM~|~
    \Del_t} {\Gam_t \vdash \lambda x.t: \umult{\IM \rew \UM} \mid
    \Del_t } \hspace{0.7cm} \Theu \rhd \Gamu
    \Vdash u:  \choice{\IM} \mid  \Del}  {\Gam \vdash
  (\l x. t)u:\UM \mid \Del } $$ where $\Gam = \Gam_t \inter
\Gamu$, $\Del=\Del_t\union  \Delu$ and $\Any = \UM$. 
We then construct the following derivation $\Phi'$:
$$
\infer[(\subs)]{\Phit\rhd \Gam_t ;  x:\IM  \vdash t:\UM \mid \Delt \hspace{1cm}
       \Theu \rhd \Gam_u \Vdash u:\choice{\IM}   \mid  \Delu}
      {\Gamt \inter \Gamu  \vdash t[x/u]:\UM  \mid \Del_t\union \Delu  }$$
We  conclude since  $ 
   \sz{\Phi} =    \sz{\Phi_{t}} +  \sz{\Theu} + 2  
   > \sz{\Phi_t} +  \sz{\Theu}  =    \sz{\Phi'}$. 

\item If $o = \slist [\mu \al. \Com] u \Rew{}  \slist[\mu \al'. \Com\rempl{\al}{\al'}{u}] = o'$: we proceed by induction on $\slist$,
by detailing only the case $\slist = \Box$ as the other one is straightforward.
The  derivation $\Phi$ has the following form:
$${ \Phi= 
 {\infer[(\appet)]{ \infer[(\mud)]{\PhiCom \tri 
       \muju{\GamCom}{\Com:\TypCom}{\al:\VMC;\DelCom
     }}
                  {\muju{\GamCom}{\mu\al.\Com: \choice{\VMC}}{\DelCom}} \hspace{0.7cm}
                   \Theu \tri \muJu{\Gamu}{u:\choice{\IMu}}{\Delu}
                          }
         { \muju{\GamCom\inter \Gamu}{(\mu\al.\Com)u:\UM}{\Delu}}
  }}
$$
 where $\choice{ \VMC}= \umult{\IMl \rew \VMl}_{\lL}$, $\choice{\IMu}= \inter_{\lL}   \choice{\IMl}$, $\UM= \vee_{\lL} \VMl$, $\Gam = \GamCom \inter \Gamu$ and $\Del = \DelCom \union \Delu$.
 
Moreover, Lemma~\ref{l:non-empty-union} implies  $L \neq \es$,
so that $\inter_{\lL} \choice{\IMl} = \choice{(\inter_{\lL} \choice{\IMl})}$.

We then construct the following derivation $\Phi'$:
$$
\infer[(\mud)]{\infer[(\repl)]{\PhiCom \hspace{1.5cm} \Theu }
             {\Gam' \inter \Gamu \vdash \Com\rempl{\al}{\al'}{u}: \TypCom \mid \Del' \union  \Delu; \al': \UM }}
      {\Gam' \inter \Gamu  \vdash  \mu \al'. \Com\rempl{\al}{\al'}{u}:\UM  \mid \Del' \union  \Delu  }
$$
We conclude since $|L|\geq 1$ in the following equation:
\[ \begin{array}{l} 
   \sz{\Phi'}  
   = \sz{\Phi_{\Com\rempl{\al}{\al'}{u}} } + 1 \\
   = \sz{\Phi_{\Com}} + \sz{\Theu} +  |L| - \frac{1}{2} + 1  \\ 
   = \sz{\Phi_{\mu \al. \Com}} + \sz{\Theu} + |L| - \frac{1}{2}  \\
   < \sz{\Phi_{\mu \al. \Com}} + \sz{\Theu } + |L|    = \sz{\Phi}
         \end{array} \]

\item If $o = \cxtt\cwc{x}[x/u] \Rew{} \cxtt\cwc{u}[x/u] = o'$, with $|\cxtt\cwc{x}|_x>1$.
The derivation $\Phi$ has the following form:
{ 
  $$\infer[(\subs)]{\Phi'_{\cxtt\cwc{x}}
    \tri \Gam'; x:  \IM \vdash \cxtt\cwc{x}: \UM  \mid \Del'  \hspace{0.7cm} 
         \Theu   \tri  \Gamu  \vdash u :  \choice{\IM}  \mid \Delu } 
        {\Gam' \inter \Gamu \vdash \cxtt\cwc{x}[x/u]: \UM\mid \Del' 
       \union  \Delu}$$}
Moreover, $|\cxtt\cwc{x}|_x>1$ so that 
Lemma~\ref{l:relevance-bis} applied to $\Phi_{\cxtt\cwc{x}}$ gives $\IM \neq \emul$ and thus   $\choice{\IM} = \IM$.
We can then apply Lemma~\ref{l:partial-substitution} which gives  a derivation
\[\Phi_{\cxtt\cwc{u}} \tri \Gam' \inter \Gamuu; x: \IM_2   \vdash
\cxtt\cwc{u}: \UM \mid \Del' \union \Deluu \] 
where $\IM = \IM_1 \inter \IM_2$ and $\IM_1 \neq \emul$
and $\Gamu = \Gamuu \inter \Gamud$ and $\Delu = \Deluu \union \Delud$. 
Moreover $\Theuu \tri \Gamuu \Vdash u: \IM_1 \mid \Deluu$,
$\Theud \tri \Gamud \Vdash u: \IM_2 \mid \Delud$, 
and $\sz{\Phi_{\cxtt\cwc{u}}} = \sz{\Phi_{\cxtt\cwc{x}}} + \sz{\Theuu}  - | \IM_1 | $.

The  hypothesis $|\ctx\cwc{x}|_x>1$ implies  $|\ctx\cwc{u}|_x>0$, then 
$ \IM_2 \neq \emul$ by Lemma~\ref{l:relevance-bis} applied to $\Phi_{\cxtt\cwc{u}}$ so that $\choice{\IM_2} = \IM_2 $.
We can then  construct the derivation $\Phi'$ as follows:
$$\infer[(\subs)]{\Phi_{\cxtt\cwc{u}} \hspace{1cm}
         \Theud}
        { \Gam' \inter \Gam \vdash \cxtt\cwc{u}[x/u]: \UM \mid \Del' \inter \Del }
$$
We conclude since
$\sz{\Phi'} = \sz{\Phi_{\cxtt\cwc{u}}} +   \sz{\Theud} =_{Lemma~\ref{l:partial-substitution}}
\sz{\Phi_{\cxtt\cwc{x}}} + \sz{\Theuu}  - |  \IM_1 |  +   \sz{\Theud} = \sz{\Phi_{\cxtt\cwc{x}}} + 
\sz{\Theu} - | \IM_1 |  < \sz{\Phi}$. 

The step $<$ is justified by $\IM_1 \neq \emul$.
\item If $o = \cxtt\cwc{x}[x/u] \Rew{} \cxtt\cwc{u} = o'$, with $|\cxtt\cwc{x}|_x = 1$.
The derivation $\Phi$ has the following form:
{
$$\infer[(\subs)]{\Phi_{\cxtt\cwc{x}} \tri \Gam'; x:  \IM \vdash \cxtt\cwc{x}: \UM  \mid \Del'  \hspace{0.7cm} 
        \Theu  \tri \Gamu   \Vdash u : \choice{\IM}  \mid \Delu } 
        {\Gam' \inter \Gamu \vdash \cxtt\cwc{x}[x/u]: \UM\mid \Del'   \union \Delu}$$}
Lemma~\ref{l:relevance-bis} applied to $\Phi_{\cxtt\cwc{x}}$ gives $\IM \neq \emul$ and thus   $\choice{\IM} = \IM$.
We can then apply Lemma~\ref{l:partial-substitution} which gives  a derivation
\[\Phi_{\cxtt\cwc{u}} \tri \Gam' \inter \Gamuu; x: \IM_2   \vdash
\cxtt\cwc{u}: \UM \mid \Del' \union \Deluu \] 
where $\IM = \IM_1 \inter \IM_2$ and $\IM_1 \neq \emul$
and $\Gamu = \Gamuu \inter \Gamud$ and $\Delu = \Deluu \union \Delud$. 
Moreover $\Theuu \tri \Gamuu \Vdash u: \IM_1 \mid \Deluu$,
$\Theud \tri \Gamud \Vdash u: \IM_2 \mid \Delud$, 
and $\sz{\Phi_{\cxtt\cwc{u}}} = \sz{\Phi_{\cxtt\cwc{x}}} + \sz{\Theuu}  - |  \IM_1|$.
By hypothesis $|\cxtt\cwc{x}|_x=1$ so that $|\cxtt\cwc{u}|_x=0$, then 
$\IM_2= \es$ by Lemma~\ref{l:relevance-bis} applied to $\Phi_{\cxtt\cwc{u}}$. Thus $\IM = \IM_1$. 
We then set $\Phi' = \Phi_{\cxtt\cwc{u}}$ and conclude since
\[ \begin{array}{l}
   \sz{\Phi'} = \sz{\Phi_{\cxtt\cwc{u}}}  
   =_{Lemma\ref{l:partial-substitution}} \sz{\Phi_{\cxtt\cwc{x}}} +  \sz{\Theuu}  - | \IM_1| 
   = \sz{\Phi_{\cxtt\cwc{x}}} +  \sz{\Theu}    - |  \IM|   < \\
   = \sz{\Phi_{\cxtt\cwc{x}}} +  \sz{\Theu}      = \sz{\Phi}
   \end{array} \]

The step $<$ is justified by $\IM = \IM_1 \neq \emul$.

\item If $o = \cxcc\cwc{\co{\al}t} \rempl{\al}{ \al'}{u}   \Rew{} 
 \cxcc\cwc{\co{\al'}tu} \rempl{\al}{ \al'}{u} = o'$, with 
$|\cxcc\cwc{\co{\al}t}_\al > 1$. Then $\Phi$ has the following form

{
$$
  \infer[(\repl)]{\Phi_{\Com}
    \tri \GamCom  \vdash  \cxcc\cwc{\co{\al}t}: \TypCom \mid  \DelCom; \al: \VM' \hspace{0.7cm} \Theu \tri \Gamu  \Vdash  u: \IMu  \mid \Delu}    {\GamCom \inter \Gamu \vdash \cxcc\cwc{\co{\al}t} \rempl{\al}{ \al'}{u}: \TypCom \mid \DelCom \union \Delu \vee \al': \union_{\lL} \VMl}  $$}
where $\Com = \cxcc\cwc{\co{\al}t} $, $ \VM'= \umult{\IMl \rew
  \VMl}_{\lL}$, $\IMu= \choice{ (\inter_{\lL} \choice{\IMl})}$,
$\AM = \TypCom$, $\Gam = \GamCom \inter \Gamu$ and
$\Del = \DelCom \union \Delu\vee \al': \union_{\lL} \VMl$.
Since $|\cxcc\cwc{\co{\al}t}|_\al > 1$ implies $L \neq
\es$ by Lemma~\ref{l:relevance-bis}, we have that  $\IMu=\inter_{\lL}
\choice{\IMl}$.  By Lemma~\ref{l:partial-replacement} there are $L_1,\,
L_2,\ \Gamuu,\, \Gamud,\ \Deluu,\, \Delud,\ \Phi_{\cxcc\cwc{\co{\al'}
    tu}} $ s.t.
\begin{itemize}
\item $L=L_1\uplus L_2$, where $L_1\neq \es$.
\item $\Gamu = \Gamuu \inter \Gamud$ and $\Delu = \Deluu \union \Delud$, 
\item $\Theuu \tri \muJu{\Gamuu}{u: \inter_{\lL_1} \choice{\IMl}}{\Deluu}$,
\item $\Theud \tri \muJu{\Gamud}{u: \inter_{\lL_2} \choice{\IMl} }{\Delud}$,
\item $\Phi_{\cxcc\cwc{\co{\al'} tu}} \tri \muju{\GamCom \inter \Gamuu}{\cxcc{\cwc{\co{\al'} tu}}:\AM}{\al: \umult{ \IMl \rew \VMl}_{\lL_2}; \al':\union_{\lL_1} \VMl \union \DelCom \union \Deluu}$, and 
\item $\sz{\Phi_{\cxcc\cwc{\co{\al'} tu}}} = \sz{\Phi_{\cxcc\cwc{\coal t}}} + \sz{\Theuu}$.
\end{itemize}

   Moreover, $|\cxcc\cwc{\co{\al}t}|_\al > 1$ implies $|\cxcc\cwc{\co{\al'}tu}|_\al > 0$ so that
   $L_2 \neq \es$ holds by Lemma~\ref{l:relevance-bis} and thus
   $\inter_{\lL_2} \choice{\IMl}= \choice{(\inter_{\lL_2} \choice{\IMl})}$. Then we can build the
   following derivation $\Phi'$:

   $${ \infer[(\repl)]{\Phi_{\cxcc\cwc{\co{\al'}tu} } \hspace{1.5cm}
            \Theud }
     {\Gam' \vdash \cxcc\cwc{\co{\al'}tu} \rempl{\al}{ \al'}{u}: \TypCom \mid \Del' }
   }$$
   where $\Gam'=(\GamCom \inter \Gamuu)\inter \Gamud =\Gam$, $\Del'=(\al':\union_{\lL_1} \VMl \union \DelCom \union \Deluu) \union \Delud \union (\al':\union_{\lL_2} \VMl)=\Del$.

   We conclude since
 \[ \begin{array}{l}
       \sz{\Phi'} = \sz{\Phi_{\cxcc\cwc{\co{\al'}tu} }} + \sz{\Theud} + | L_2| - \frac{1}{2}  \\
       =_{Lemma~\ref{l:partial-replacement}}    \sz{\Phi_{\cxcc\cwc{\co{\al}t}}} + \sz{\Theuu} + \sz{\Theud} + |L_2| - \frac{1}{2}   \\
       = \sz{\Phi_{\cxcc\cwc{\co{\al}t}}} + \sz{\Theu} + |L_2| - \frac{1}{2}    \\
       < \sz{\Phi_{\cxcc\cwc{\co{\al}t}}} + \sz{\Theu} + | L | - \frac{1}{2}         = \sz{\Phi}      
 \end{array} \]

The step $<$ is justified because $L_1 \neq \es$ and thus $|L_2 | < |L|$.

\item If $o = \cxcc\cwc{\co{\al}t} \rempl{\al}{ \al'}{u}   \Rew{} 
 \cxcc\cwc{\co{\al'}tu}  = o'$, with 
$|\cxcc\cwc{\co{\al}t}|_\al =  1$. The derivation $\Phi$ has the following form

$${
  \infer[(\repl)]{\Phi_{\Com}
    \tri \GamCom  \vdash  \cxcc\cwc{\co{\al}t}: \TypCom \mid  \DelCom; \al: \VM' \hspace{0.7cm} \Theu \tri \Gamu  \Vdash  u: \IMu  \mid \Delu}    {\GamCom \inter \Gamu \vdash \cxcc\cwc{\co{\al}t} \rempl{\al}{ \al'}{u}: \TypCom \mid \DelCom \union \Delu \vee \al': \union_{\lL} \VMl}  }$$
 
 where $\Com = \cxcc\cwc{\co{\al}t} $, $ \VM'= \umult{\IMl \rew \VMl}_{\lL}$, $\IMu= \choice{ (\inter_{\lL}   \choice{\IMl})}$,
 $\AM = \TypCom$, $\Gam = \GamCom \inter \Gamu$ and $\Del = \DelCom \union \Delu\union \al': \union_{\lL} \VMl$. 
Since $|\cxcc\cwc{\co{\al}t}|_\al = 1$ implies $L \neq \es$ by Lemma~\ref{l:relevance-bis}, we have that $\IMu=\inter_{\lL}   \choice{\IMl}$.
By Lemma~\ref{l:partial-replacement} there are
$L_1,\, L_2,\ \Gamuu,\, \Gamud,\ \Deluu,\, \Delud,\ \Phi_{\cxcc\cwc{\co{\al'} tu}} $  s.t.
\begin{itemize}
\item $L=L_1\uplus L_2$, where $L_1\neq \es$.
\item $\Gamu = \Gamuu \inter \Gamud$ and $\Delu = \Deluu \union \Delud$, 
\item $\Theuu \tri \muJu{\Gamuu}{u: \inter_{\lL_1} \choice{\IMl}}{\Deluu}$,
\item $\Theud \tri \muJu{\Gamud}{u: \inter_{\lL_2} \choice{\IMl} }{\Delud}$,
\item $\Phi_{\cxcc\cwc{\co{\al'} tu}} \tri \muju{\GamCom \inter \Gamuu}{\cxcc{\cwc{\co{\al'} tu}}:\AM}{\al: \umult{ \IMl \rew \VMl}_{\lL_2}; \al':\union_{\lL_1} \VMl \union \DelCom \union \Deluu}$, and 
\item $\sz{\Phi_{\cxcc\cwc{\co{\al'} tu}}} = \sz{\Phi_{\cxcc\cwc{\coal t}}} + \sz{\Theuu}$.
\end{itemize}

   Moreover, $|\otx\cwc{\co{\al}t}|_\al = 1$ implies $|\otx\cwc{\co{\al'}tu}|_\al = 0$ so that
   $L_2 = \es$ and $L=L_1$ holds by Lemma~\ref{l:relevance-bis}. Thus, $\Theuu=\Theu$ and so on. We then set $\Phi' = \Phi_{\otx\cwc{\co{\al'}tu} }$
   and  conclude since

 \[ \begin{array}{l}
       \sz{\Phi'} 
       = \sz{\Phi_{\otx\cwc{\co{\al'}tu} }}  \\
       =_{Lemma~\ref{l:partial-replacement}}    \sz{\Phi_{\otx\cwc{\co{\al}t}}} + \sz{\Theu}   \\
       <  \sz{\Phi_{\otx\cwc{\co{\al}t}}} + \sz{\Theu}  + | L | - \frac{1}{2}   
       = \sz{\Phi} \\ 
   \end{array} \] 

 The step $<$ is justified because $L\neq \es$, so that $|L|\geqslant 1$
 implies  $|L|- \frac{1}{2}>0$. \qedhere

\end{itemize}
\end{proof}


\begin{lem}[Reverse Partial Substitution]
\label{l:reverse-partial-substitution} 
Let $\Phi \tri \Gam \vdash \cxot\cwc{u}: \Any \mid \Del$, where $x\notin \fv{u}$.
Then, there exist $\Gam_0, \Del_0, \IM_0\neq \emul, \Gamu, \Delu$ such that 
\begin{itemize}
   \item $\Gam = \Gam_0 \inter \Delu$, 
     \item $\Del = \Del_0 \union \Delu$, 
     \item $\Phi_{\cxot\cwc{x}} \tri \muju{\Gam_0 \inter x:\IM_0}{\cxot\cwc{x}:\Any }{\Del_0}$       
     \item $\tri \muJu{\Gamu}{u:\IM_0}{\Delu}$.
\end{itemize}
\end{lem}

\begin{proof}  The proof is by induction on the context $\cxot$. 
For this induction to work, we need as usual to adapt the statement for auxiliary derivations.
We only  show the  case $\cxot =  \Box$ since  all the other  ones are
straightforward and rely on suitable partitions of the contexts in the
premises. So assume $\cxot=\Box$,  then $\Any=\UM$ for some $\UM$.  We
set $\Gam_0 = \Del_0=\es$,  $\IM_0=\mult{\UM}$, $\Gamu = \Gam$, $\Delu
= \Del$  (so that  $\tri \muJu{\Gamu}{u:\IM_0}{\Delu}$ holds  by using
the $(\many)$ rule), and
$$ \Phi_{x} = \infer[(\ax)]{ }{\muju{x:\mult{\UM}}{x:\UM}{\es} }$$ 
  
The claimed set and context equalities trivially hold.   
\end{proof}

\begin{lem}[Reverse Partial Replacement]
\label{l:reverse-partial-replacement} 
Let $\muju{\Gam}{\cxoc\cwc{\co{\al'}tu}:\Any}{\al':\VM;\Del}$, where  $\al,\al' \notin \fn{u}$.  Then there exist $\Gam_0, \Del_0,  \VM_0, K \neq \es, (\IMk)_{\kK}, (\VMk)_{\kK}, \Gamu, \Delu$ such that 
\begin{itemize}
\item $\Gam = \Gam_0 \inter \Gamu$,
\item $\Del=\Del_0 \union \Delu$,
\item $\VM = \VM_0 \union_{\kK} \VMk$, 
\item $\tri \muju{\Gam_0}{\cxoc\cwc{\co{\al}t}: \Any}{\al':\VM_0;\al:\umult{\IMk\rew \VMk}_{\kK} \union   \Del_0}$,  and 
\item     $\tri \muJu{\Gamu}{u :  \inter_{\kK}\choice{\IMk} }{\Delu}$
\end{itemize}
\end{lem}

\begin{proof} 
  The proof is by induction on the context $\cxoc$. For this
    induction to work, we need as usual to adapt the statement for auxilary
    derivations.  Notice that $\VM \neq \umult{ \,}$ by
  Lemma~\ref{l:relevance-bis}, since $\al' \in
  \fn{\cxoc\cwc{\co{\al'}tu}} $.  We only show the case $\cxoc =
  \boxdot$ since all the other ones are straightforward.  So assume
  $\cxoc=\boxdot$. Then the derivation of $\co{\al'}tu$ has the
  following form, where $K\neq \es$: {
$$
\infer[(\muu)]{\infer[(\appet)]{\Phit  \tri  \Gam_0 \vdash t:\umult{\IMk\rew \VMk}_{\kK} \mid
         \al':\VM_0; \Del_0
    \hspace{0.7cm}
   \tri \Gamu \Vdash  u: \inter_{\kK}\choice{\IMk} \mid \Delu  }
  {\Gam_0 \inter \Gamu \vdash tu: \union_{\kK} \VMk \mid \al':\VM_0; \Del_0 \union \Delu }}
      {\Gam_0 \inter \Gamu \vdash \co{\al'}tu: \TypCom \mid \al': \VM_0\union_{\kK}\VMk; \Del_0 \union \Delu}
      $$}
where $\Gam =\Gam_0 \inter \Gamu$,
and $\Del=\Del_0 \union \Delu$ and $\VM = \VM_0 \union_{\kK}\VMk $. 

We then construct the following derivation :
$$ \infer[(\muu)]{\Phit  \tri  \Gam_0 \vdash t:\umult{\IMk\rew \VMk}_{\kK} \mid
         \al':\VM_0; \Del_0 }{
 \muju{\Gam}{\co{\al}t:\TypCom}{\al':\VM_0;\al:\umult{\IMk\rew \VMk}_{\kK}\union \Del_0}}
$$
Thus, we have all the claimed set and context equalities.
\end{proof}

\noindent {\bf Property~\ref{l:pse} (Subject Expansion for $\lmuex$).}
Let $\tingD{\Phi'}{\tyj{o'}{\Gam}{\Any\mid\Del}}$. If $o \Rew{\nonelmuex} o'$
(\ie\  a non-erasing $\lmuex$-step), 
then 
$\tingD{\Phi}{\tyj{o}{\Gam}{\Any\mid\Del}}$.

\begin{proof}
By induction on the non erasing reduction relation $\Rew{\nonelmuex}$. We only show the main cases of 
non-erasing reduction at the root, the other ones being straightforward. 

\begin{itemize}
\item If $o = (\slist [\lambda x.t]) u \Rew{} \slist[t[x/u]] = o'$, we proceed by induction on $\slist$, by detailing only the case $\slist = \Box$ as the other one is straightforward. 

The derivation $\Phi'$ has the following form :
$$\infer[(\subs)]{\Phit \tri  \muju{\Gamt;x:\IM}{t:\UM}{\Delt}
   \hspace{0.7cm} \Theu \tri \muJu{\Gamu}{u:\choice{\IM}}{\Delu}   }{\muju{\Gam}{t[x/u]:\UM}{\Del}}$$

We then construct the following derivation $\Phi$: 
$$\infer[(\appet)]{ \infer[(\introarrow)]{\Phit \tri \muju{\Gamt;x:\IM}{t:\UM}{\Delt}}{
    \muju{\Gamt}{\l x. t:\umult{\IM\rew \UM}}{\Delt}}
\hspace{0.7cm}\Theu \tri \muJu{\Gamu}{u:\choice{\IM}}{\Delu}
}{\muju{\Gam}{(\l x.t)u:\UM}{\Del}}
$$

\item If $o = (\slist[\mu \al. \Com]) u \Rew{} \slist[\mu
  \al'. \Com\rempl{\al}{\al'}{u}] = o'$, where $\al'$ is fresh,
  then we proceed by induction on $\slist$, by detailing only the
  case $\slist = \Box$ as the other one is straightforward.  Then
  $\Phi'$ has the following form :
$${
\infer[(\mud)]{\infer[(\repl)]{\PhiCom \tri \muju{\GamCom}{\Com:\TypCom }{\al:\VMal;\DelCom} \hspace{1.5cm} \Theu \tri \muJu{\Gamu}{u:\IMu}{\Delu}
    }{\muju{\GamCom \inter \Gamu}{\Com\rempl{\al}{\al'}{u}: \TypCom}{\DelCom \union \Delu; \al': \VMalp}
  }}{\muju{\GamCom \inter \Gamu}{\mu \al'. \Com\rempl{\al}{\al'}{u}:\choice{(\VMalp)}}{ \DelCom \union \Delu}}
}
$$
where $ \VMal= \umult{\IMl \rew \VMl}_{\lL}$, $\IMu= \choice{(\inter_{\lL}   \choice{\IMl})}$, $\VMalp= \vee_{\lL} \VMl$,
$\Any=\choice{(\VMalp)} = \choice{(\vee_{\lL} \VMl)}$, $\Gam = \GamCom \inter \Gamu$ and $\Del = \DelCom \union \Delu$.
Notice that the  name assignment of the judgment typing
$\Com\rempl{\al}{\al'}{u}$ has the form  $\DelCom \union \Delu; \al': \VMalp$
since $\al'$ is a fresh name  by hypothesis,  so that $\al' \notin \dom{\DelCom \union \Delu}$ 
holds by Lemma~\ref{l:relevance-bis}.
We now consider two cases:

\begin{itemize}
\item If $L \neq \es$, then  $\choice{\umult{\IMl\rew \VMl}_{\lL}} = \umult{\IMl\rew \VMl}_{\lL}$, 
$\choice{(\inter_{\lL}   \choice{\IMl})} = \inter_{\lL}   \choice{\IMl}$, 
$\Any = \choice{(\vee_{\lL} \VMl)} = \vee_{\lL} \VMl$, so that 
we construct the following derivation $\Phi$:
 $$
  {
   \infer[(\appet)]{\infer[(\mud)]{\PhiCom \tri \muju{\GamCom}{\Com:\TypCom }{\al:\VMal;\DelCom}}
                 {\muju{\Gam_\Com }
                       {\mu \al.\Com:\umult{\IMl\rew \VMl}_{\lL}}{\Del_\Com}  } \hspace{0.7cm}
          \Theu \tri \muJu{\Gamu}{u:\inter_{\lL}   \choice{\IMl} }{\Delu}}
         {\muju{\GamCom \inter \Gamu}
               { (\mu \al.\Com)u : \vee_{\lL} \VMl}
               {\DelCom\union \Delu}}
  } $$

  \item If $L = \es$, then, in the derivation above,
  $\choice{(\VMalp)}=\choice{(\vee_{\lL} \VMl)}=\umult{\xi}$ for some blind type $\xi$.
Then we choose  $\choice{\umult{\IMl\rew \VMl}_{\lL}}$ to be
$ \umult{\emul \rew \umult{\xi}}$, which is a blind type. We then construct the following derivation $\Phi$:
 $$
  { 
   \infer[(\appet)]{\infer[(\mud)]{\PhiCom\tri \muju{\GamCom}{\Com:\TypCom }{\al:\VMal;\DelCom}}
                 {\muju{\Gam_\Com }
                       {\mu \al.\Com:\umult{\emul \rew \umult{\xi}}}{\Del_\Com}  } \hspace{0.7cm}
          \Theu \tri \muJu{\Gamu}{u: \choice{\emul}}{\Delu}}
         {\muju{\GamCom \inter \Gamu}
               { (\mu \al.\Com)u : \umult{\xi}}
               {\DelCom\union \Delu}}
  } $$
  We conclude since $\Any = \umult{\xi}$.

\end{itemize}

\item If $o = \cxtt\cwc{x}[x/u] \Rew{} \cxtt\cwc{u}[x/u] = o'$, with $|\cxtt\cwc{x}|_x>1$.
The derivation $\Phi'$ has the following form:

$$
\infer[(\subs)]{\Phi_{\cxtt\cwc{u}}\tri \Gamstt;x:\IM \vdash \cxtt\cwc{u}:\Any \mid \Delstt \hspace{0.7cm}
         \Theu\tri \Gamu \Vdash u:\choice{\IM} \mid \Delu}
        { \muju{\Gamstt \inter \Gamu}{\cxtt\cwc{u}[x/u]: \Any}{\Delstt \union \Delu} }
$$
where $x \in \fv{\cxtt\cwc{u}}$ implies  $\IM\neq \emul$ by Lemma~\ref{l:relevance-bis}, so that $\choice{\IM}=\IM$.

        By Lemma~\ref{l:reverse-partial-substitution} applied to
        $\Phi_{\cxtt\cwc{u}}$, we have
$\Gam'_0,\, \Del_0,\, \IM_0\neq \emul,\, \Gamu', \, \Delu'$ such that 
\begin{itemize}
   \item $\Gamstt;x:\IM = \Gam'_0 \inter \Gamu'$, 
     \item $\Delstt = \Del_0 \union \Delu'$, 
     \item $\Phi_{\cxtt\cwc{x}} \tri \muju{\Gam'_0 \inter x:\IM_0}{\cxtt\cwc{x}:\Any }{\Del_0}$       
     \item $\tri \muJu{\Gamu'}{u:\IM_0}{\Delu'}$.
\end{itemize}
We set $\IM^{''}=\IM\inter \IM_0,\, \Gamu^{''}=\Gamu \inter \Gamu',\,
\Delu^{''}=\Delu \union \Delu'$. Thus, in particular, $\choice{({\IM^{''}})} =
\IM^{''}$.  By Lemma~\ref{l:relevance-bis}, $x\notin \dom{\Gamu'}$, so
that $\Gamo'=\Gamo;x:\IM$ for some $\Gamo$ and thus $\Gamo'\inter
x:\IM_0=\Gamo;x:\IM^{''}$.  By Lemma~\ref{l:decomposition}, there is a
derivation $\Theu^{''}\tri \muJu{\Gamu^{''}}{u:\IM^{''}}{\Delu^{''}}$. We then
construct the following derivation $\Phi$ :
$$\infer[(\subs)]{\Phi_{\cxtt\cwc{x}}   \hspace{0.7cm} 
         \Theu^{''} \tri \Gamu^{''}  \vdash u : \IM^{''}  \mid \Delu^{''} }
        {\Gam_0 \inter  \Gamu^{''} \vdash \cxtt\cwc{x}[x/u]: \Any\mid \Del_0
          \union \Delu^{''}}$$
        We conclude since $\Gam_0\inter \Gamu^{''}=\Gam_0\inter \Gamu' \inter \Gamu=\Gamstt\inter \Gamu=\Gam$ and $\Del_0 \union \Delu^{''}=\Del_0\union \Delu' \union \Delu=\Delstt\union \Delu=\Del$.

\item If $o = \cxtt\cwc{x}[x/u] \Rew{} \cxtt\cwc{u} = o'$, with $|\cxtt\cwc{x}|_x = 1$.
The derivation $\Phi'$ ends with $\muju{\Gam}{\cxtt\cwc{u}:\Any}{\Del}$ where  $x\notin \dom{\Gam}$ by Lemma~\ref{l:relevance-bis}.
        By Lemma~\ref{l:reverse-partial-substitution} applied to $\Phi'$, we have
$\Gam_0,\, \Del_0,\, \IM_0\neq \emul,\, \Gamu, \, \Delu$ such that 
\begin{itemize}
   \item $\Gam = \Gam_0 \inter \Gamu$, 
     \item $\Del = \Del_0 \union \Delu$, 
     \item $\Phi_{\cxtt\cwc{x}} \tri \muju{\Gam_0 \inter x:\IM_0}{\cxtt\cwc{x}:\Any }{\Del_0}$       
     \item $\tri \muJu{\Gamu}{u:\IM_0}{\Delu}$.
\end{itemize}
Thus in particular $\choice{\IM_0}=\IM_0$. Since $x\notin \dom{\Gam}$, $x\notin \dom{\Gam_0}$,
so that $\Gam_0 \inter x:\IM_0=\Gam_0; x:\IM_0$.
We then construct the following derivation $\Phi$ :
$$\infer[(\subs)]{\Phi_{\cxtt\cwc{x}}   \\ 
         \tri \muJu{\Gamu}{u:\IM_0}{\Delu} }
{ \muju{\Gamo\inter \Gamu}{\cxtt\cwc{x}[x/u]: \UM}{\Del_0 \union \Delu}}$$
We conclude since $\Gam = \Gam_0 \inter \Gamu$ and $\Del = \Del_0 \union \Delu$.\\

\item If $o = \cxcc\cwc{\co{\al}t} \rempl{\al}{ \al'}{u}   \Rew{} 
 \cxcc\cwc{\co{\al'}t u} \rempl{\al}{ \al'}{u} = o'$, with 
 $|\cxcc\cwc{\co{\al}t}_\al > 1$.
Then  $\Phi'$ has the following form :
 
   $${ \infer[(\repl)]{
    \Phio'\tri \Gamscc \vdash  \cxcc\cwc{\co{\al'}tu} :\Any \mid 
                \Delscc ;   \al': \VM_{\al'}; \al:\VM_{\al} \hspace{1cm}
              \Theu \tri \muJu{\Gamu}{u:\IMu }{\Delu}}
             {\muju{\Gamscc \inter \Gamu}
                   {\cxcc\cwc{\co{\al'}tu} \rempl{\al}{ \al'}{u}:\Any}
                   { (\Delscc; \al' : \VM_{\al'} ) \union \Delu\union \al':\VM}}
 } $$
where $\VMal= \umult{\IMl \rew \VMl}_{\lL}$, $\IMu= \choice{(\inter_{\lL}   \choice{\IMl})}$,
$\VM=  \vee_{\lL} \VMl$,
$\Gam = \Gamscc \inter \Gamu$ and
$\Del = (\Delscc; \al' : \VM_{\al'} ) \union \Delu \union \al':\VM =
(\Delscc \union \Delu; \al' : \VM_{\al'} \union \VM)$
since $\al'\notin \fn{u}$ implies $\al' \notin \dom{\Delu}$. Since $\al\in \fn{\cxcc\cwc{\co{\al'}tu}}$, then $L\neq \es$ by Lemma~\ref{l:relevance-bis}, so that $\IMu= \inter_{\lL}   \choice{\IMl}$.

By Lemma~\ref{l:reverse-partial-replacement} applied to $\Phio'$, we have
$ \Gamo,\, \Delo'$,$\, \Phio,\, \VM_0,\, K \neq \es, \, (\IMk)_{\kK},\, (\VMk)_{\kK},\, \Gamu',\Delu'$, and $\Theu'$ such that 
\begin{itemize}
\item $\Gamscc = \Gam_0 \inter \Gamu'$,
\item $\Delscc;\al:\VM_\al=\Delo' \union \Delu'$,
\item $\VMalp = \VM_0 \union_{\kK} \VMk$, 
\item $\Phio\tri \muju{\Gam_0}{\cxcc\cwc{\co{\al}t}: \Any}{\al':\VM_0;\al:\umult{\IMk\rew \VMk}_{\kK} \union  \Delo'}$  and 
\item     $\tri \muJu{\Gamu'}{u :  \inter_{\kK}\choice{\IMk} }{\Delu'}$
\end{itemize}
We set $L^{''}=L \uplus K,\, \Gamu^{''}=\Gamu \inter \Gamu',\, \Delu^{''}=\Delu
\union \Delu'$ and $\choice{(\IMu^{''})}=\IMu^{''}=\inter_{\lL^{''}} \choice{\IMl}$ (indeed, $L^{''}\supseteq K\neq \emptyset$). By
Lemma~\ref{l:relevance-bis}, $\al\notin \dom{\Delu'}$, so that
$\Delo'=\Delo;\al:\VM_{\al}$ for some $\Delo$ and
$\al':\VM_0;\al:\umult{\IMk\rew \VMk}_{\kK} \union
\Delo'=\al':\VM_0;\al:\umult{\IMl\rew \VMl}_{\lL^{''}};\Delo$ since
$\al'\notin \dom{\Delo'}$.  By Lemma~\ref{l:decomposition}, there is
$\Theu^{''}\tri \muJu{\Gamu^{''}}{u:\IMu^{''}}{\Delu^{''}}$.  We then construct
the following derivation $\Phi$:
$${
  \infer[(\repl)]
        {\Phio \tri \muju{\Gam_0}
                         {\cxcc\cwc{\co{\al}t}: \Any}
                         {\al':\VM_0;\al:\umult{\IMk\rew \VMk}_{\kK} \union  \Delo'} \quad
         \Theu^{''} \tri \muJu{\Gamu^{''}}{u:\IMu^{''}}{\Delu^{''}}}{\muju{\Gamo \inter  \Gamu^{''}}
  {\cxcc\cwc{\co{\al}t}\rempl{\al}{ \al'}{u}:\Any }{
    \al':\VM_0\union_{\lL^{''}}\VMl; \Delo\union \Delu^{''}}  }
}$$
We conclude since
$\Gamo \inter \Gamu^{''}=\Gamo \inter\Gamu'  \inter  \Gamu =\Gamscc \inter \Gamu =\Gam$,
$\Delo\union \Delu^{''}=\Delo \union \Delu' \union \Delu=\Delscc \union \Delu$ and
$\VM_0\union_{\lL^{''}} \VMl=\VM_0 \union_{\kK} \VMk \union_{\lL} \VMl=
\VMalp \union_{\lL} \VMl = \VMalp \union \VM$.

\item If $o = \cxcc\cwc{\co{\al}t} \rempl{\al}{ \al'}{u}   \Rew{} 
 \cxcc\cwc{\co{\al'}t u}  = o'$, with 
$|\cxcc\cwc{\co{\al}t}|_\al =  1$, then the derivation $\Phi'$ necessarily ends 
with the  judgment
 $\muju{\Gam}{\cxcc\cwc{\co{\al'}tu}:\Any}{\Delscc;\al':\VM}$, where $\Del=\Delscc;\al':\VM$.

By Lemma~\ref{l:reverse-partial-replacement} applied to  $\Phi'$,  
 we have $ \Gamo,\, \Delo,\, \VM_0,\, K \neq \es, \, (\IMk)_{\kK},\, (\VMk)_{\kK},\, \Gamu,\, \Delu$, and $\Theu$ such that 
\begin{itemize}
\item $\Gam = \Gamo \inter \Gamu$,
\item $\Delscc=\Delo \union \Delu$,
\item $\VM= \VMo \union_{\kK} \VMk$, 
\item $\Phio\tri \muju{\Gamo}{\cxcc\cwc{\co{\al}t}: \Any}{\al':\VM_0;\al:\umult{\IMk\rew \VMk}_{\kK} \union  \Delo} $  and 
\item $\Theu \tri \muJu{\Gamu}{u :  \inter_{\kK}\choice{\IMk} }{\Delu}$
\end{itemize}
Notice that  $K \neq \es$ implies $\choice{(  \inter_{\kK}\choice{\IMk})} =  \inter_{\kK}\choice{\IMk} $. 
Moreover, by Lemma~\ref{l:relevance-bis}, since $\al \notin  \fn{\cxcc\cwc{\co{\al'}t u}}$, then $\al \notin \dom{\Del}$, thus $\al \notin \dom{\Del_0}$
and $\al:\umult{\IMk\rew \VMk}_{\kK} \union \Delo =\al:\umult{\IMk\rew \VMk}_{\kK} ; \Delo$. 
We then construct $\Phi$ :
{
$$\infer[(\repl)]{
  \Phio\tri \muju{\Gamo}{\cxcc\cwc{\co{\al}t}: \Any}{\al':\VM_0;\al:\umult{\IMk\rew \VMk}_{\kK} \union  \Delo} \quad \Theu \tri \muJu{\Gamu}{u :  \inter_{\kK}\choice{\IMk} }{\Delu}}{
  \muju{\Gam}{ \cxcc\cwc{\co{\al}t} \rempl{\al}{ \al'}{u}:\Any}{(\Delo;\al':\VMo)\union \Delu \union \al':\union_{\kK}\VMk } }
$$
}

We conclude since $\al' \notin \fn{u}$ implies $\al' \notin \dom{\Delu}$
by Lemma~\ref{l:relevance-bis} so that  $(\Delo;\al':\VMo )\union \Delu\union \al': \union_{\kK} \VMk = \Delo\union \Delu; \al':\VMo\union_{\kK} \VMk=\Delscc;\al':\VM=\Del$
  as desired. \qedhere
\end{itemize}
\end{proof}





\end{document}